\documentclass[10pt, twocolumn, comsoc]{IEEEtran}

\def\argmax{\mathop{\rm arg\,max}}

\usepackage{graphicx,epsfig}
\usepackage[noadjust]{cite}
\usepackage{mcite}
\usepackage{booktabs}
\usepackage{amsfonts,helvet}
\usepackage{fancyhdr}
\usepackage{threeparttable}
\usepackage{epsf,epsfig}
\usepackage{amsthm}
\usepackage{amsmath}
\usepackage{siunitx}
\usepackage{amssymb}
\usepackage{dsfont}
\usepackage{subfigure}
\usepackage{color}
\usepackage[linesnumbered,ruled,noend]{algorithm2e}
\usepackage{algpseudocode}
\usepackage{algcompatible}
\usepackage{enumerate}
\usepackage{gensymb}
\usepackage{cancel}
\usepackage{graphicx}
\usepackage{wrapfig}
\usepackage{ragged2e}

\newtheorem{lemma}{Lemma}

\newtheorem{proposition}{Proposition}
\newtheorem{remark}{Remark}

\usepackage{bbm}
\usepackage{eucal}

\usepackage{multirow}
\usepackage{dsfont}
\usepackage{boldline}

\def\ba{{\bf a}}

\def\bc{{\bf c}}

\def\be{{\bf e}}
\def\bff{{\bf f}}

\def\bh{{\bf h}}

\def\bs{{\bf s}}

\def\bu{{\bf u}}
\def\bv{{\bf v}}
\def\bw{{\bf w}}

\def\bA{{\bf A}}
\def\bB{{\bf B}}
\def\bC{{\bf C}}
\def\bD{{\bf D}}
\def\bE{{\bf E}}
\def\bF{{\bf F}}

\def\bH{{\bf H}}
\def\bI{{\bf I}}

\def\bP{{\bf P}}
\def\bQ{{\bf Q}}
\def\bR{{\bf R}}

\def\bX{{\bf X}}
\def\bY{{\bf Y}}

\def\cC{\mbox{$\mathcal{C}$}}

\def\cL{\mbox{$\mathcal{L}$}}

\def\cN{\mbox{$\mathcal{N}$}}

\def\bbC{\mbox{$\mathbb{C}$}}

\def\bbE{\mbox{$\mathbb{E}$}}

\begin{document}

\title{
% Multi-RIS Scalable Beamforming for MU-MIMO Systems with Imperfect CSIT
Scalable Beamforming Design for Multi-RIS-Aided  MU-MIMO Systems with Imperfect CSIT
}
\author{Mintaek~Oh and Jinseok~Choi

\thanks{M. Oh and J. Choi are with the School of Electrical Engineering, Korea Advanced Institute of Science and Technology (KAIST), Daejeon, 34141, Republic of Korea (e-mail: {\texttt{\{ohmin, jinseok\}@kaist.ac.kr}}). 
}
}

\maketitle \setcounter{page}{1} 

\begin{abstract}
This paper presents a scalable beamforming design for maximizing the spectral efficiency (SE) of multi-reconfigurable intelligent surface (RIS)-aided communications through joint optimization of the precoder and RIS phase shifts in multi-user multiple-input multiple-output (MU-MIMO) systems under imperfect channel state information at the transmitter (CSIT).
To address key challenges of the joint optimization problem,
we first decompose it into two subproblems by deriving a proper lower bound.
We then leverage a generalized power iteration (GPI) approach to identify a superior local optimal precoding solution.
We further extend this approach to the RIS design using regularization;
we set a RIS regularization function to efficiently handle the unit-modulus constraints, and also find the superior local optimal solution for RIS phase shifts  under the GPI-based optimization framework.
Subsequently, we propose an alternating optimization method.
{\color{black}
Our proposed algorithm offers scalable multi-RIS beamforming in terms of computational complexity that scales linearly with the number of RISs, while achieving superior performance.
We further reduce the complexity with respect to the number of RIS elements by using diagonal approximation of the channel error covariance and avoiding direct matrix inversion.}
Simulations validate the proposed algorithm in terms of both the sum SE performance and the scalability.
% $\CMcal{O}(L)$
\end{abstract}

\begin{IEEEkeywords}
Reconfigurable intelligent surface, imperfect channel state information, beamforming, generalized power iteration, alternating optimization, regularization.
\end{IEEEkeywords}

\section{Introduction}
A reconfigurable intelligent surface (RIS) has been recognized as a promising technique for future wireless communication systems \cite{liu2021reconfigurable}.
By adjusting numerous passive reflecting elements, the RIS
enhances wireless channel conditions with flexible control and configuration, resulting in performance improvements  with marginal increase in network power consumption.
For instance, in millimeter-wave (mmWave) communications, signals can easily experience a severe attenuation due to blockage
\cite{heath2016overview}.
% \cite{niu2015survey, heath2016overview}.
To overcome this issue, the deployment of RIS can help creating an additional signal path, thereby increasing the coverage of mmWave systems  \cite{wang2020joint}.
% li2022ris
% With the support of RIS, a wide range of scenarios can be realized, from the simplest scenario where user equipment is localized in a single RIS-aided system to advanced configurations enabling joint beamforming with multiple RISs, gathering potential communication gains  \cite{do2021multi}.
% The RIS enables various deployment scenarios. 
% These range from simple single RIS-aided localization systems to advanced configurations with multiple RISs, which can maximize communication performance through joint beamforming.}
% \cite{zheng2021double}

Despite the significant potential of the RIS, obtaining precise wireless channel state information (CSI) poses a considerable practical issue in RIS-aided systems  \cite{lin2021channel}.
This challenge is exacerbated by the fact that the passive nature of RIS does not facilitate any signal transmission, reception, or processing \cite{yao2023robust}.
% Accordingly, channel estimation in RIS-aided systems remains a difficult issue, and active research focuses on minimizing channel estimation errors \cite{zheng2022survey}.
In addition, deploying more RISs 
% the large number of reflecting elements in RIS-aided systems 
introduces significant burden on channel estimation, which can lead to substantial performance degradation due to degraded channel estimation accuracy.
% \cite{gao2023robust}
% For this reason, the channel estimation is even more challenging in the systems utilizing multiple RISs.
This challenge motivates us to establish beamforming strategies for RIS-aided systems with inaccurate CSI.
In this regard, it is imperative to consider robust transmission design that accounts for imperfect CSI at the transmitter (CSIT) and a scalable beamforming strategy to offer a flexible beamforming solution tailored for multi-RIS-aided systems.

% ----------------------------------------------------------------
\subsection{Related Works}
% ----------------------------------------------------------------
% Multi-RIS
Because of the promising potential of utilizing the RIS, there exist various analytical works  investigating  multi-RIS-aided networks.
Specifically, wireless networks assisted by two RISs demonstrate that the effective channel benefits from the superposition of the double-reflection link \cite{mei2022intelligent}. 
In \cite{han2020cooperative}, it was shown that if the double-reflection channel follows a line-of-sight (LoS) path, it can attain a higher order of passive beamforming gain compared to individual single-reflection links.
% This advantage becomes more pronounced as the total number of reflecting elements increases asymptotically.
The outage probability and average symbol error probability for  Nakagami-$m$ fading of multi-RIS-aided systems were derived in \cite{aldababsa2023multiple}.
The result in \cite{aldababsa2023multiple} indicates that multi-RIS configurations can significantly enhance diversity order and system performance.
In \cite{li2023performance}, the spectral efficiency (SE) of multi-RIS-aided systems was analyzed with Poisson point processes.
% The results reveal that multiple RISs achieve greater performance than active nodes, which necessitates a deployment strategy of multiple RISs.
% The results demonstrate that systems with RISs outperform those with active nodes.
{\color{black}It was also confirmed in \cite{al2024performance} that deploying more RISs up to an optimal number increases the network sum SE, and the optimal number can increase with proper RIS phase optimization.}
% , indicating the advantages of deployment strategies.
Overall, the literature emphasizes the potential benefit achieved by deploying multiple RISs.

% ------------------------------------------------------
% Multi-RIS
Motivated by the potential of multi-RIS-aided systems, the RIS beamforming optimization has been actively investigated for multiple-input and multiple-output (MIMO) communication systems.
According to \cite{li2020weighted}, the problem of maximizing the weighted sum rate in a system supported by multiple RISs was tackled by jointly optimizing the precoder and RIS phase shifts.
In \cite{pan2020multicell}, the study extended the application of RIS-aided beamforming to multi-cell networks, advancing the methodology to encompass multi-RIS systems through the aggregation of all relevant channel matrices associated with RISs.
For a multiuser MIMO (MU-MIMO) downlink system, a codebook-based beamforming method was proposed to minimize transmit power while satisfying quality-of-service (QoS) constraints
% for individual users with reduced complexity  
\cite{huang2023multi}.
{\color{black}
A recent study has also explored optimization strategies for distributed RIS architectures \cite{chen2025joint}.}
% In addition, due to unfavorable propagation conditions, 
It is often considered that the direct link between the base station (BS) and users is not available.
For instance, in \cite{do2022line}, the RIS-aided communication system was handled through MIMO transmission by considering an upper limit for the channel capacity of the RIS for the blocked direct link scenario.
In \cite{wang2020joint}, the hybrid beamforming design was proposed in the RIS-assisted point-to-point MIMO system without the consideration of the direct link.
% ------------------------------------------------------
In addition, the RIS can play a pivotal role in supporting various next-generation wireless communication systems, such as integrated sensing and communication as well as low-Earth orbit satellites \cite{ zheng2024cooperative, li2024joint, toka2024ris}.
% The versatility and adaptability of RIS technology make it particularly attractive for diverse wireless applications, ranging from enhanced mobile broadband to ultra-reliable low-latency communications \cite{pala2022joint}.
Therefore, a comprehensive beamforming optimization framework is indispensable for a general multi-RIS-aided system to fully leverage these diverse applications and maximize system performance.

% ----------------------------------
% Imperfect CSIT
Considering the difficulty in channel estimation for the RIS-aided system, robust beamforming approaches have been proposed to compensate for channel estimation inaccuracies \cite{zhou2020framework,omid2021low,zeng2022joint,omid2023robust,jiang2023robust}.
{\color{black}Moreover, advanced mathematical frameworks for channel estimation \cite{wasfi2026unified} have been developed, broadening the scope of robust system design.}
In \cite{zhou2020framework}, a robust beamforming design was proposed to minimize transmit power under worst-case rate constraints with imperfect CSIT of cascaded channels.
% This work observed that a significant channel estimation error may lead to detrimental effects on overall system performance, also provided an engineering insight for the careful selection of the size of the RIS under imperfect CSIT.
The study found that substantial channel estimation errors can significantly degrade overall system performance, and  provided engineering insights for the size of the RIS when dealing with imperfect CSIT.
% In \cite{zeng2022joint}, a joint beamforming optimization was tackled for RIS-aided MU-MIMO systems with imperfect CSI, and this work demonstrated that the optimal placement of RIS tends to be closer to the BS.
In \cite{jiang2023robust}, it was emphasized that precise channel estimation is crucial in achieving max-min fairness and QoS for multi-group systems, noting that deploying RIS may not yield performance gains and could potentially degrade system performance without accurate CSI.
% ------------------------------------------------

Although several approaches have been successfully proposed by employing famous optimization frameworks such as manifold optimization and successive convex approximation techniques \cite{pan2020multicell,zhou2020intelligent,zargari2020energy,xiu2021irs,pan2022overview,zeng2022joint,shtaiwi2023sum,wang2023ris,yoon2023joint}, most approaches have achieved sub-optimal solutions, leaving a room for further performance improvement. 
In addition, some of the proposed approaches have limited scalability for multiple RISs in terms of the computational complexity.
Specifically, in \cite{pan2020multicell,pan2022overview,guo2020weighted},  manifold optimization (MO)-based algorithm and a majorization-minimization (MM)-based algorithm yield sub-optimal performance due to the inherently non-convex nature of the optimization problems.
Moreover, regarding the number of RISs $L$, the high-performing algorithms' complexity increases in the order of $\CMcal{O}(L^2)$ or even higher  \cite{li2020weighted, yoon2023joint, huang2023multi}.
Therefore, there is a need for more effective and scalable beamforming designs tailored for multi-RIS-aided systems which can surpass existing  sub-optimal and high-complexity methods, especially for imperfect CSIT scenarios.

\subsection{Contributions}
{\color{black}
We investigate multi-RIS-aided MU-MIMO downlink systems with imperfect CSIT.
% Although imperfect CSIT scenarios have been widely studied, our work introduces key innovations in both methodology and scalability for multi-RIS systems.
Specifically, we propose a novel optimization framework that jointly optimizes the precoder and RIS phase shifts with achieving scalability in  both the number of RISs and their elements, significantly advancing the state of the art. 
% in an alternating manner with a structured channel error model. 
% This approach ensures superior local optimality and achieves linear scalability with respect to the number of RISs, 
The main contributions are summarized as follows:
}

\begin{itemize}
    \item 
    In the considered system, we assume that estimated channels and their  estimation error covariance are known at the BS.
    Under this assumption, we adopt the instantaneous SE  to fully exploit the available channel knowledge, which is desirable to maximize the performance gain.
    This instantaneous SE is the short-term SE expression that averages  over channel estimation error distribution.
    To make this metric tractable, we  derive a lower bound of the instantaneous SE.
    Consequently, our metric involves the channel error covariance as well as the estimated channel, allowing us to exploit the partial CSIT.
    Based on this approach, we can achieve a robust solution aiming at maximizing the sum SE under imperfect CSIT.

    \item Based on the derived instantaneous SE bound, we first optimize the precoder based on a generalized power iteration (GPI) method  \cite{choi2019joint}.
    To this end, we begin with decomposing the problem into two subproblems: $(i)$ the precoder optimization and $(ii)$ the RIS phase shifts optimization.
    Subsequently, deriving the first-order optimality condition for the precoder with given RIS phase shifts, we interpret the precoding optimization problem as a generalized eigenvalue problem.
    Subsequently, applying the GPI method, we find the superior local optimal precoding solution that is a principal eigenvector.

    \item 
    Regarding RIS phase shifts, we develop a regularized GPI method to effectively resolve the unit-modulus constraint and achieve multi-RIS scalability by utilizing a block-diagonal structure in the GPI method.
    Introducing a regularization term as a function of RIS phase shifts,
    the unit-modulus condition is controlled by the  regularization function without the explicit constraint.
    Through smooth approximation and  GPI-friendly reformulation of the problem, we apply the GPI method, ultimately achieving superior local optimal solution.

    \item 
    {\color{black} By leveraging block-diagonal matrices of the GPI framework, the overall complexity becomes lower compared to state-of-the-art methods.
    In particular, our approach enables scalable beamforming by reducing the computational complexity to a linear growth with the number of RISs, i.e., $\CMcal{O}(L)$.
    % Moreover, for a fixed total number of RIS elements, the complexity scales with respect to $L$ as $\CMcal{O}(1/L^2)$.
    {\color{black}
    We further reduce the complexity with respect to the number of RIS elements by using diagonal approximation of the channel error covariance and avoiding direct matrix inversion.
    % Moreover, we propose a reduced-complexity variant that mitigates the cubic complexity with respect to the number of RIS elements.
    }
    }

    \item 
    Via simulations, we demonstrate that our proposed method outperforms baselines across various scenarios, achieving superior local solutions.
    We also empirically confirm that the regularized GPI approach nearly satisfies the unit-modulus constraint, which verifies the effectiveness of our regularized GPI approach.
    In addition, we show that our method exhibits its multi-RIS scalability     while achieving the highest sum SE.
    
\end{itemize}

% ----------------------------------------------------------------
\textit{Notation}:
% ----------------------------------------------------------------
The superscripts $(\cdot)^{\sf T}$, $(\cdot)^{\sf H}$, $(\cdot)^{*}$, and $(\cdot)^{-1}$ denote the transpose, Hermitian, complex conjugate, and matrix inversion, respectively. ${\bf{I}}_N$ is the identity matrix of size $N \times N$ and $\bf 0$ is a zero vector with proper dimension.
We use ${\rm diag}(\ba)$ for a diagonal matrix with $\ba$ on its diagonal elements.
Assuming that ${\bf{A}}_1, \dots, {\bf{A}}_N \in \mathbb{C}^{K \times K}$, ${\bf {A}} = {\rm blkdiag}\left({\bf{A}}_1,\dots, {\bf{A}}_N \right) \in \bbC^{KN\times KN}$ is a block-diagonal matrix.
$\|\bf A\|$ represents L2 norm.
We use ${\rm{tr}}(\cdot)$ for trace operator, ${\rm{vec}}(\cdot)$ for vectorization, $\otimes$ for Kronecker product, and ${\rm arg}(\cdot)$ for the argument of a complex number.
% ${\cU}(a,b)$ is a uniform distribution with two boundaries $a$ and $b$
With mean $m$ and variance $\sigma^{2}$, we use $\cC\cN(m,\sigma^{2})$ for a circularly symmetric complex Gaussian distribution and $\cN(m,\sigma^{2})$ for a Gaussian distribution.
% $\CMcal{O}(\cdot)$ implies the big-O notation.
% We also follow MATLAB-style notation.

%%%%%%%%%%%%%%%%%%%%%%%%%%%%%%%%%%%%%%%%%%%%%%%%%%%%%%%%%%%%
\section{System Model And Problem Formulation} \label{sec:sys_model}
%%%%%%%%%%%%%%%%%%%%%%%%%%%%%%%%%%%%%%%%%%%%%%%%%%%%%%%%%%%%
% ----------------------------------------------------------------
\subsection{Signal Model}\label{subsec:channel}
% ----------------------------------------------------------------
%%%%%%%%%%%%%%%%%%%%%%%%%%%%%%%%%%%%%%%%%%%%%%%%%%%%%%%%%%%%%%%%
\begin{figure}[!t]    
    {\centerline{\resizebox{0.95\columnwidth}{!}{\includegraphics{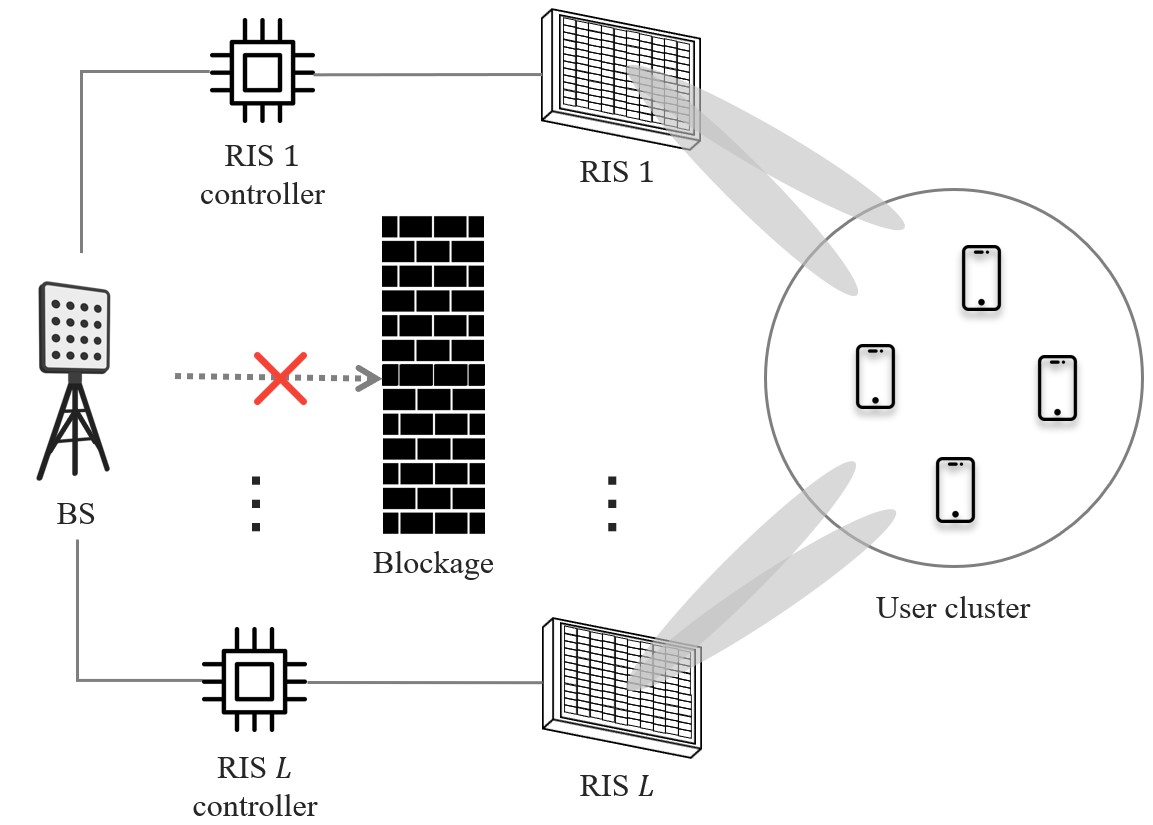}}}
    \caption{\color{black}An illustration of the   multi-RIS-aided downlink MU-MIMO system architecture.}
    \label{fig:system}}
\end{figure}
%%%%%%%%%%%%%%%%%%%%%%%%%%%%%%%%%%%%%%%%%%%%%%%%%%%%%%%%%%%%%%%%
{\color{black}As shown in Fig.~\ref{fig:system},} we consider a multi-RIS-aided MU-MIMO downlink system where the BS equipped with $N$ antennas serves $K$ single-antenna users assisted with $L$ RISs.
% We assume that all RISs have the same number of phase shifts $M$.
Each RIS is equipped with $M$ reflecting elements: ${\boldsymbol{\Phi}}_{\ell} = {\rm diag}(\boldsymbol{\phi}_{\ell})$ where $\boldsymbol{\phi}_{\ell} = [\phi_{1,\ell}, \phi_{2,\ell}, \ldots, \phi_{M,\ell}]^{\sf T}$  with $\phi_{m,\ell} = e^{j\theta_{m,\ell}}, \forall m, \forall \ell$.
% We denote the user set, RIS set, and RIS phase shifts set as $\CMcal{K}=\left\{1,\cdots,K\right\}$, $\CMcal{L}=\left\{1,\cdots,L\right\}$, and $\CMcal{M}=\left\{1,\cdots,M\right\}$, respectively.
{\color{black}We denote the user set, RIS set, and set of indices of RIS reflecting elements as $\CMcal{K}=\{1,\dots,K\}$, $\CMcal{L}=\{1,\dots,L\}$, and $\CMcal{M}=\{1,\dots,M\}$, respectively.}
% We assume that the BS has accurate knowledge of the channel state information (CSI) for all channels. 
The BS determines the optimal phase shifts and conveys the information back to the corresponding RIS controller.
{\color{black}Assuming that the RIS is deployed in the vicinity of the BS, the control signaling can be implemented by an error-free link with negligible latency via wired links \cite{yao2023superimposed}.}
% The data symbol for user $k$ defined as $s_k$ is drawn from a Gaussian distribution with a zero mean and variance of $\bbE [|s_k|^2] = P$, $\forall k \in \CMcal{K}$.
The BS broadcasts the data symbols $s_k \sim \cC \cN (0, P)$, $\forall k \in \CMcal{K}$ to each legitimate user via a precoder $\bF = [{\bf f}_1 , \ldots , {\bf f}_K] \in \bbC^{N \times K}$
where ${\bf f}_{k} \in \bbC^{N}$ indicates a precoding vector for $s_k$.
Then, a transmitted signal vector is given by
\begin{align}
    {\bf{x}} = \sum_{k=1}^{K}{\bf{f}}_k s_k = {\bF\bs} \in \bbC^{N},
\end{align}
where $\bs = [s_1, \dots, s_K]^{\sf T} \in \bbC^{K}$.

Let ${\bH}_{1,\ell} \in \mathbb{C}^{N \times M}$ and ${\bH}_{2,\ell} \in \mathbb{C}^{M \times K}$ denote an $\ell$th BS-RIS channel matrix and a $\ell$th RIS-user channel matrix, respectively.
For the link between RIS $\ell$ and user $k$, we consider $\bh_{2,k,\ell} \in \bbC^{N}$ based on a user channel matrix ${\bH}_{2,\ell} = [\bh_{2,1,\ell}, \cdots, \bh_{2,K,\ell}]$.
We assume that the direct link between transmitter and receiver is not available.
The systems assisted by multiple RISs often ignore secondary reflections between the surfaces, which is a reasonable oversight when the RISs are in each other’s far-field.
Thus, we assume a far-field propagation model\footnote{{\color{black}
In multi-RIS scenarios, inter-RIS interference from double-reflection links may occur.
However, compounded pathloss factors substantially attenuate the inter-RIS interference gain relative to the primary links. 
This theoretical insight justifies our assumption as both analytically tractable and practically acceptable for multi-RIS-aided systems.
}} 
so that the double-reflection link can be ignored \cite{pan2022overview}.
Accordingly, a received signal at user $k$ is given by
% \begin{align}
%     {\bf y} = \left(\sum_{\ell=1}^{L} \bH_{1,\ell}{\boldsymbol{\Phi}}_{\ell} \bH_{2, \ell}\right)^{\sf H}\bF \bs + \bn,
% \end{align} 
% Then, the received signal at user $k$ is written as 
\begin{align}
    \label{eq:y_k}
    y_{k} \!=\! \left(\!\sum_{\ell=1}^{L}\bH_{1,\ell}{\boldsymbol{\Phi}}_{\ell} \bh_{2,k,\ell}\!\right)^{\sf H}\!\!\!{\bf{f}}_{k} s_{k} \!+\!\!\! \sum_{i = 1, i \neq k}^{K}\!\!\left(\!\sum_{\ell=1}^{L} \bH_{1,\ell} {\boldsymbol{\Phi}}_{\ell} \bh_{2,k,\ell}\!\right)^{\sf H}\!\!\!{\bf{f}}_{i} s_{i} \!+\! n_k,
\end{align}
where $n_k \sim \mathcal{CN}(0,\sigma^2)$ is additive white Gaussian noise (AWGN) of the $k$th user. 
%%%%%%%%%%%%%%%%%%%%%%%%%%%%%%%%%%%%%%%
% \begin{figure}[t]    
%     {\centerline{\resizebox{0.9\columnwidth}{!}{\includegraphics{Figures/Fig_system_model.jpg}}}}
%     % \vspace{-0.2cm}
%     \caption{The considered  multi-RIS-aided MU-MIMO communication system where the direct link is blocked.}
%     \label{fig:system}
% \end{figure}
%%%%%%%%%%%%%%%%%%%%%%%%%%%%%%%%%%%%%%%

Now, let us define the effective channel matrix between the BS and the $K$ users as 
\begin{align}
    \bH = \sum_{\ell=1}^{L} \bH_{1,\ell} {\boldsymbol{\Phi}}_{\ell} \bH_{2,\ell} = [{\bf h}_1, \dots, {\bf h}_K].
\end{align}
Then denoting the cascaded channel for user $k$ at RIS $\ell$ as $\bH^{\sf r}_{k,\ell} = \bH_{1,\ell} {\rm diag}(\bh_{2,k,\ell}) \in \bbC^{N \times M}$, we rewrite the effective channel vector of user $k$ as
\begin{align}
    \label{eq:effect_ch_k}
    \bh_k = \sum_{\ell=1}^{L} \bH_{1,\ell} {\boldsymbol{\Phi}}_{\ell}\bh_{2,k,\ell} = \sum_{\ell=1}^{L} \bH_{k,\ell}^{\sf r} {\boldsymbol {\phi}}_{\ell}.
\end{align}
We consider a block fading model where the channel is invariant within each transmission block.

% ----------------------------------------------------------------
\subsection{Channel Acquisition Model} \label{subsec:channel_est}
% ----------------------------------------------------------------
% Now, we consider that the BS performs an uplink channel estimation. 
% When the channel estimation is implemented, the BS has imperfect knowledge of all channels.
% Hence, utilizing a channel estimation scheme is important. 
We  assume that the BS performs uplink channel estimation and thus, the BS only has partial knowledge of the channels due to estimation imperfections. 
% Hence, using an effective channel estimation method is crucial.
In \cite{kundu2021channel,pan2022overview}, the RIS channel estimation is implemented for the cascaded channel under a time-division duplex mode.
For the cascaded channel, the estimated CSIT is given by
\begin{align}
    \hat{\bH}^{\sf r}_{k,\ell} = \bH_{k,\ell}^{\sf r} - \bE_{k,\ell},
\end{align}
or equivalently in a vectorized form, we have
\begin{align}
    \hat{\bc}_{k,\ell} = \bc_{k,\ell} - \be_{k,\ell},
\end{align}
where $\bE_{k,\ell}$ is a channel estimation error matrix, $\hat{\bc}_{k,\ell} = {\rm vec}(\hat{\bH}^{\sf r}_{k,\ell})$, $\bc_{k, \ell} = {\rm vec}(\bH_{k,\ell}^{\sf r})$, and $\be_{k,\ell} = {\rm vec}(\bE_{k,\ell})$.

% We note that the channel estimation error is independent of the estimated channel.

% Since the unknown channel $\bH_{k,\ell}^{\sf r}$ does not generally follow a Gaussian distribution, obtaining the posterior distribution poses a significant challenge.
% In the channel acquisition model, we adopt a LMMSE channel estimation scheme.
% Since the LMMSE estimator is the unbiased estimator, a mean of the error vector is zero regardless of the mean of channels, i.e., $\bbE[\be_{k,\ell}] = \bbE[\bc_{k,\ell}] - \bbE[\hat{\bc}_{k,\ell}] 
 % = \boldsymbol{0}, \forall k \in \CMcal{K}, \forall \ell \in \CMcal{L}$.
According to \cite{pan2022overview}, using the linear minimum mean square error (LMMSE) estimator with $T_{\sf UL}$ uplink training length, the error covariance matrix is given by
\begin{align}
\label{eq:Err_Cov}
    \bR^{\sf e}_{k,\ell} \!=\! \bbE \left[\be_{k,\ell} \be_{k,\ell}^{\sf H}\right] \!=\! \left(\bC_{k,\ell}^{-1} + \frac{\rho_{\sf UL}}{\gamma_{k, \ell}\sigma^2}{\boldsymbol{\Psi}}_{\ell}^{\sf H} {\boldsymbol{\Psi}}_{\ell}\otimes \bI_{N} \right)^{-1}\!\!\!\!\!,
\end{align}
where $\rho_{\sf UL}$ is uplink training power, $\bC_{k,\ell}$ is the  cascaded channel covariance matrix, ${\boldsymbol{\Psi}}_{\ell}$ is the training phase shift matrix of $\ell$th RIS, and $\gamma_{k, \ell}$ denotes the pathloss of the cascaded channel determined by $\gamma_{k, \ell} = \gamma_{1,\ell} \gamma_{2,k,\ell}$ where $\gamma_{1,\ell}$ and $\gamma_{2,k,\ell}$ are the pathloss of the $\ell$th RIS-BS and RIS-user $k$, respectively.
% Note that the error covariance matrix in \eqref{eq:Err_Cov} is invariant to the training signals.
% Hence, the optimization scheme of the channel estimation needs to be handled.
% One representative scheme is to design the known training matrix ${\boldsymbol{\Psi}}_{k,\ell}$.
% The simplest scheme is the on-off scheme which is to switch each reflecting element on and off.
% In addition, to enhance the estimation performance, we consider the discrete discrete Fourier transform (DFT) training scheme \cite{kundu2021channel}.
% It is noteworthy that the error covariance matrix in \eqref{eq:Err_Cov} remains unaffected by the choice of training signals
While the error covariance matrix in \eqref{eq:Err_Cov} remains unaffected by the choice of training signals, the performance of channel estimation is influenced by other factors, particularly the pattern of ${\boldsymbol{\Psi}}_{\ell}$.
Hence, careful design of an appropriate training matrix  
 ${\boldsymbol{\Psi}}_{\ell}$ is crucial for optimizing the channel estimation process.
In this paper, we adopt the discrete Fourier transform (DFT)-based training scheme \cite{kundu2021channel}.
% Using the $T_{\sf UL} \times T_{\sf UL}$ DFT matrix, the training phase shift matrix of $\ell$th RIS is given by
% Based on the DFT method, the training phase shift matrix ${\boldsymbol{\Psi}}_{k,\ell}$ is equal to the first $M$ columns of a $T_{\sf UL} \times T_{\sf UL}$ DFT matrix as
% \begin{align}
    % [{\boldsymbol{\Psi}}_{k,\ell}]_{t,m} = e^{-j \frac{2\pi (t-1)(m-1)}{T_{\sf UL}}},\; t= 1, \ldots, T_{\sf UL},\; m = 1, \ldots, M.
% \end{align}
With the DFT scheme, the error covariance matrix is given by
\begin{align}
    \label{eq:error_cov}
    \bR^{\sf e}_{k,\ell} = \left(\bC_{k,\ell}^{-1} +  \frac{T_{\sf UL} \rho_{\sf UL}}{\gamma_{k, \ell}\sigma^2}  \bI_{N M}\right)^{-1}, 
\end{align}
where the sequence length for channel training needs to satisfy  $T_{\sf UL} \geq M$.
As the uplink training length and power increase to infinity, the error covariance in \eqref{eq:error_cov} goes to zero; thereby the CSIT error eventually vanishes.

In the presence of multiple RISs, {\color{black}
we assume the on-off scheme \cite{pan2022overview,xu2022channel} for each individual RIS, where each RIS is independently switched on and off by exploiting reflection and absorption characteristics of its elements \cite{basar2024reconfigurable}.}
% When switched off, the RIS is configured to operate in an absorption mode, yielding near-zero reflection so that the incident signals are effectively absorbed rather than reflected [REF2].
Specifically, the $\ell$th cascaded channel is estimated while turning off all reflecting elements of RISs related to other cascaded channels.
Then, each cascaded channel estimation process is divided into $M$ stages in which each stage only estimates one column vector of $\bH^{\sf r}_{k,\ell}$ \cite{wei2021channel}.
Once the channel estimation is completed, all cascaded channels are combined to obtain \eqref{eq:effect_ch_k} for all users.
We adopt this channel estimation scheme with the LMMSE estimator by taking advantage of leveraging the error covariance matrix in $\eqref{eq:error_cov}$ for our system.

% Although we employ a simplified estimation process, our approach is not limited to the considered estimation with the on-off scheme for multiple RISs. 
% Developing an advanced channel estimation method for multi-RIS-aided systems is not the primary contribution of this work. 
% Therefore, we utilize the DFT-based LMMSE estimator along with its error covariance matrix.

% -------------------------------------------------------------
\subsection{Performance Metrics and Problem Formulation} \label{subsec:metrics}
% -------------------------------------------------------------
% While the BS is unable to predict the short-term (instantaneous) SE, it has access to the average SE 
With the effective channel in \eqref{eq:effect_ch_k}, the SE of user $k$ is
\begin{align}
    \label{eq:SE_1}
    R_k = \log_2 \left(1+\frac{\left|\bh_k^{\sf H}\bff_k\right|^2}{\sum_{i=1,i \neq k}^{K}\left|\bh_k^{\sf H}\bff_i\right|^2 + \frac{\sigma^2}{P}}\right).
\end{align}
% Unfortunately, the BS cannot predict the short-term SE for perfect CSI; yet, it has {\color{red}access} to the long-term SE defined as $\bbE_{\{\bH_{k,\ell}^{\sf r}|\hat{\bH}^{\sf r}_{k,\ell}\}}[\left.R_k \right| \hat{\bH}^{\sf r}_{k,\ell}]$ for the known RIS phase shifts in the DFT-based training scheme for all $\ell$.
Unfortunately, the BS cannot predict the SE in \eqref{eq:SE_1} for imperfect CSI.
To overcome this, we consider the instantaneous SE defined as $\bbE_{\{\bH_{k,\ell}^{\sf r}|\hat{\bH}^{\sf r}_{k,\ell}\}}\![\!\left.R_k \right| \hat{\bH}^{\sf r}_{k,\ell}\!]$ \cite{joudeh2016sum, lee2023max}.
While an ergodic SE represents the long-term SE achievable when channel coding spans extensive channel blocks, the instantaneous SE refers to the short-term SE expression that averages channel estimation error distribution in each channel realization.
We exploit the instantaneous SE to express the ergodic SE in our problem with imperfect CSIT.
% In this regard, we can take into account the ergodic SE by averaging the SE over the variations in the estimated channel.
% The ergodic SE is expressed by averaging the average SE over the variations in the estimated channel.

According to \cite{joudeh2016sum}, the ergodic SE of user $k$ is defined as
\begin{align} 
    \label{eq:SE_2}
    \bbE_{\{\bh_k\}}\!\!\left[R_{k}\right]
    &\!=\! \bbE_{\{{\bH}^{\sf r}_{k,\ell}, {{\hat{\bH}^{\sf r}_{k,\ell}}}\}}\!\left[R_k\right] 
    % \!\bbE_{\{{{\hat{\bH}^{\sf r}_{k,\ell}}}\}}\!\left[\left.\!
    % \bbE_{\{{{{\bE}_{k,\ell}}}\}}\!\left[\log_2\!\left(1 \!+\! \frac{\left|\left(\sum_{\ell=1}^{L}\bH^{\sf r}_{k,\ell}{\boldsymbol {\phi}}_{\ell}\right)^{\sf H}\!\!{\bf{f}}_k \right|^2}{\sum_{i \neq k}^{K}\!\left| \left(\sum_{\ell=1}^{L}\!\bH^{\sf r}_{k,\ell}{\boldsymbol {\phi}}_{\ell}\right)^{\sf H}\!\!{\bf{f}}_i \right|^2 \!+\! \frac{\sigma^2}{P}} \right)\right|\!{{\hat{\bH}^{\sf r}_{k,\ell}}}\!\right]\!\right],
    \nonumber \\ \nonumber
    &\!=\! \bbE_{{{\{\hat{\bH}^{\sf r}_{k,\ell}}}\!\}}\!\!\left[\bbE_{\{{\bH}^{\sf r}_{k,\ell}\!|{{\hat{\bH}^{\sf r}_{k,\ell}}}\!\}}\!\!\left[R_k|\hat{\bH}^{\sf r}_{k,\ell}\right]\!\right] 
    \\
    &=\! \bbE_{\{{{\hat{\bH}^{\sf r}_{k,\ell}}}\!\}}\!\left[R^{\sf ins}_k\right],
    % &= \bbE_{{{\{\hat{\bH}^{\sf r}_{k,\ell}}}\}}\!\left[\bbE_{\{\bE_{k,\ell}\}}\left[R_k|\hat{\bH}^{\sf r}_{k,\ell}\right]\right] 
    % \!=\! \bbE_{\{{{\hat{\bH}^{\sf r}_{k,\ell}}}\}}\!\left[R^{\sf ins}_k\right],
\end{align}
where the expectation is taken over the randomness associated with the imperfect knowledge of the channel fading process and $R_k^{\sf ins}$ is the instantaneous SE.
% averaging the SE when taking into account the variability due to the channel estimation error per channel realization.
From \eqref{eq:SE_2}, we define this instantaneous SE as
\begin{align}
    \label{eq:ins_SE}
    R^{\sf ins}_k \!=\! \bbE_{\{{{{\bE}_{k,\ell}}}\}}\!\!\left[\left.\log_2 \left(1 \!+\! \frac{\left|\left(\sum_{\ell=1}^{L}\bH^{\sf r}_{k,\ell}{\boldsymbol {\phi}}_{\ell}\right)^{\sf H}\!\!\!{\bf{f}}_k \right|^2}{\sum_{i=1, i \neq k}^{K}\!\left|\! \left(\sum_{\ell=1}^{L}\!\bH^{\sf r}_{k,\ell}{\boldsymbol {\phi}}_{\ell}\right)^{\sf H}\!\!\!{\bf{f}}_i \right|^2 \!\!\!+\! \frac{\sigma^2}{P}}\right)\right|\hat{\bH}^{\sf r}_{k,\ell}\right]\!\!.
\end{align}
% where the expectation accounts for the variability due to the CSIT error within one particular coherence block.
% In this regard, we derive the lower bound of the average rate of user $k$ as 
% $\bar{R}^{\sf ins}_k(\bff_k, \boldsymbol{\Phi})$ in \eqref{eq:ins_SE_f}.
% By introducing the following proposition, we derive the lower bound of the average rate for all users.
However, \eqref{eq:ins_SE} is still not tractable because there exists no closed-form expression for the expectation of the CSIT error. 
To achieve a closed-form expression for the instantaneous SE, we introduce the following proposition.
\begin{proposition}
    \label{pro:ins_SE}
    Using the LMMSE estimator, we consider all channel estimation errors as uncorrelated noise with the transmitted signal.
    Accordingly, the lower bound of the instantaneous SE is derived as
    \begin{align}
        R_{k}^{\sf ins}\!&\geq\!
        \log_2 \!\left(\!1 \!+\! \frac{\left|\left(\sum_{\ell=1}^{L}\hat{\bH}^{\sf r}_{k,\ell}{\boldsymbol {\phi}}_{\ell}\right)^{\sf H} {\bf{f}}_k \right|^2}
        {\sum_{i=1, i \neq k}^{K}\!\left|\left(\sum_{\ell=1}^{L}\hat{\bH}^{\sf r}_{k,\ell}{\boldsymbol {\phi}}_{\ell}\right)^{\sf H}\!\!\!{\bf{f}}_i \right|^2 
        \!\!+\! 
        \sum_{i = 1}^{K}{\bf{f}}^{\sf{H}}_i
        {\boldsymbol{\Xi}}_k {\bf{f}}_i
        \!+\! \frac{\sigma^2}{P}}\!\right)
        \nonumber \\ \label{eq:ins_SE_f}
        &=\bar{R}_k^{\sf ins}({\bf F}, {\boldsymbol{\Phi}}),
\end{align}
where ${\boldsymbol{\Xi}}_k = \sum_{\ell=1}^{L}\left({\boldsymbol {\phi}}^{\sf T}_{\ell} \otimes \bI_N\right) \bR_{k,\ell}^{\sf e}\left({\boldsymbol {\phi}}^{\sf T}_{\ell} \otimes \bI_N\right)^{\sf H}$.
\begin{proof}
Using the LMMSE estimator, we can treat channel estimation errors as uncorrelated noise:
\begin{align}
    &y_{k} = 
    \left(\sum_{\ell=1}^{L}\hat{\bH}^{\sf r}_{k,\ell}\boldsymbol{\phi}_{\ell}\right)^{\sf H}\!\!\!{\bf{f}}_{k} s_{k} + \sum_{i = 1, i \neq k}^{K}\left(\sum_{\ell=1}^{L}\hat{\bH}^{\sf r}_{k,\ell}\boldsymbol{\phi}_{\ell}\right)^{\sf H}\!\!\!{\bf{f}}_{i} s_{i}
    \nonumber\\
    &\qquad\qquad\qquad + \underbrace{\sum_{i=1}^{K}\!\left(\sum_{\ell=1}^{L}{\bE}_{k,\ell}\boldsymbol{\phi}_{\ell}\right)^{\sf H}\!\!\!{\bf{f}}_{i} s_{i}}_{{\rm uncorrelated\; errors}}  + n_k.
    % y_{k} \!=\! {\pmb{\psi}}_k^{\sf H}{\bf{f}}_{k} s_{k} \!+\!\!\! \sum_{i = 1, i \neq k}^{K}\!\!{\pmb{\psi}}_k^{\sf H}{\bf{f}}_{i} s_{i} \!+\! \sum_{i=1}^{K}\pmb{\epsilon}_k^{\sf H}{\bf{f}}_{i} s_{i} \!+\! n_k,
\end{align}
% where the first term is the desired signal for user $k$, the second term is the inter-user-interference from all other users, and the third term reflects the impact of channel estimation errors. 
Then, the lower bound of the instantaneous SE is derived by the following procedure from \eqref{eq:R_lb_a} to \eqref{eq:R_lb_c} {\color{black}at top of the next page.}
 % as shown at the top of this page.
\begin{figure*}[h]
\vspace{-2em}
\begin{align} 
    % \label{eq:R_lb_a}
    \label{eq:R_lb_a}
    \bbE_{\{\hat{\bH}^{\sf r}_{k,\ell}\}}\left[R^{\sf ins}_k\right]
   % &\stackrel{(a)}\geq
      % \bbE_{\{{{\hat{\bH}^{\sf r}_{k,\ell}}}\}}\left[
   % \bbE_{\{{{{\bE}_{k,\ell}}}\}}\left[\log_2 \left.\left(1 + \frac{\left|\left(\sum_{\ell=1}^{L}\hat{\bH}^{\sf r}_{k,\ell}{\boldsymbol {\phi}}_{\ell}\right)^{\sf H} {\bf{f}}_k \right|^2}
   % {\sum_{i=1, i \neq k}^{K} \left|\left(\sum_{\ell=1}^{L}\hat{\bH}^{\sf r}_{k,\ell}{\boldsymbol {\phi}}_{\ell}\right)^{\sf H} {\bf{f}}_i \right|^2 
   % + \sum_{i = 1}^{K}\left|\left(\sum_{\ell=1}^{L}{\bE}_{k,\ell}{\boldsymbol {\phi}}_{\ell}\right)^{\sf H} {\bf{f}}_i \right|^2 
   % + \frac{\sigma^2}{P}} \right)  \right|{{\hat{\bH}^{\sf r}_{k,\ell}}}, \right] \right]
   %  \\
    &{\color{black}\stackrel{(a)}\geq \bbE_{\{{{\hat{\bH}^{\sf r}_{k,\ell}}}\}}\!\!\left[\bbE_{\{{{{\bE}_{k,\ell}}}\}}\left[\left.\log_2 \left(1 \!+\! \frac{\left|\left(\sum_{\ell=1}^{L}\hat{\bH}^{\sf r}_{k,\ell}{\boldsymbol {\phi}}_{\ell}\right)^{\sf H}\!{\bf{f}}_k \right|^2}{\sum_{i=1, i \neq k}^{K}\!\left|\! \left(\sum_{\ell=1}^{L}\hat{\bH}^{\sf r}_{k,\ell}{\boldsymbol {\phi}}_{\ell}\right)^{\sf H}\!{\bf{f}}_i \right|^2 \!+\! \sum_{i=1}^{K}\!\left|\! \left(\sum_{\ell=1}^{L}\hat{\bE}_{k,\ell}{\boldsymbol {\phi}}_{\ell}\right)^{\sf H}\!{\bf{f}}_i \right|^2 \!+\! \frac{\sigma^2}{P}}\right)\right|\hat{\bH}_{k,\ell}^{\sf r}\right]\right]}
    \\
    &{\color{black}\stackrel{(b)}\geq
    \label{eq:R_lb_b}
    % &\stackrel{(b)}\geq  
    \bbE_{\{{{\hat{\bH}^{\sf r}_{k,\ell}}}\}}\left[\log_2 \left(1 \!+\! \frac{\left|\left(\sum_{\ell=1}^{L}\hat{\bH}^{\sf r}_{k,\ell}{\boldsymbol {\phi}}_{\ell}\right)^{\sf H} {\bf{f}}_k \right|^2}
    {\sum_{i=1, i \neq k}^{K} \left|\left(\sum_{\ell=1}^{L}\hat{\bH}^{\sf r}_{k,\ell}{\boldsymbol {\phi}}_{\ell}\right)^{\sf H} {\bf{f}}_i \right|^2 
    +  \bbE_{\{{{{\bE}_{k,\ell}}}\}}
    \left[
    \sum_{i = 1}^{K} {\bf{f}}^{\sf{H}}_i\left(\sum_{\ell=1}^{L}\bE_{k,\ell} {\boldsymbol {\phi}}_{\ell} {\boldsymbol {\phi}}_{\ell}^{\sf H} \bE_{k,\ell}^{\sf H}\right) {\bf{f}}_i\right]
    \!+\! \frac{\sigma^2}{P}}\right)
    \right]}
    % \\ 
    % &=
    % \bbE_{\{{{\hat{\bH}^{\sf r}_{k,\ell}}}\}}\!\!\left[\log_2 \!\left(1 \!+\! \frac{\left|\left(\sum_{\ell=1}^{L}\hat{\bH}^{\sf r}_{k,\ell}{\boldsymbol {\phi}}_{\ell}\right)^{\sf H} {\bf{f}}_k \right|^2}
    % {\sum_{i=1, i \neq k}^{K}\! \left|\left(\sum_{\ell=1}^{L}\hat{\bH}^{\sf r}_{k,\ell}{\boldsymbol {\phi}}_{\ell}\right)^{\sf H}\!\!\!{\bf{f}}_i \right|^2 
    % \!\!+\!  \bbE_{\{{{{\bE}_{k,\ell}}}\}}
    % \left[
    % \sum_{i = 1}^{K}\sum_{\ell=1}^{L}{\bf{f}}^{\sf{H}}_i
    % \left(\boldsymbol{\phi}^{\sf T}_{\ell} \otimes \bI\right)\be_{k,\ell}\be_{k,\ell}^{\sf H}\left(\boldsymbol{\phi}^{\sf T}_{\ell} \otimes \bI\right)^{\sf H}\!\!\! {\bf{f}}_i\right]
    % \!+\! \frac{\sigma^2}{P}}\right)
    % \right]
    % \\
    % &\stackrel{(b)}=
    % \label{eq:R_lb_2}
    % \bbE_{\{{{\hat{\bH}^{\sf r}_{k,\ell}}}\}}\left[\log_2 \left(1 \!+\! \frac{\left|\left(\sum_{\ell=1}^{L}\hat{\bH}^{\sf r}_{k,\ell}{\boldsymbol {\phi}}_{\ell}\right)^{\sf H} {\bf{f}}_k \right|^2}
    % {\sum_{i=1, i\neq k}^{K} \left|\left(\sum_{\ell=1}^{L}\hat{\bH}^{\sf r}_{k,\ell}{\boldsymbol {\phi}}_{\ell}\right)^{\sf H} {\bf{f}}_i \right|^2 
    % + 
    % {\color{black} 
    % \sum_{i = 1}^{K}\sum_{\ell=1}^{L}{\bf{f}}^{\sf{H}}_i
    % \left({\boldsymbol {\phi}}^{\sf T}_{\ell} \otimes \bI_N \right)\bR_{k,\ell}^{\sf e}\left({\boldsymbol {\phi}}^{\sf T}_{\ell} \otimes \bI_N\right)^{\sf H} {\bf{f}}_i
    % }
    % \!+\! \frac{\sigma^2}{P}}\right)
    % \right]
    \\
    \label{eq:R_lb_c}
    &{\color{black}\stackrel{(c)}= \bbE_{\{{{\hat{\bH}^{\sf r}_{k,\ell}}}\}}\left[\log_2 \!\left(\!1 \!+\! \frac{\left|\left(\sum_{\ell=1}^{L}\hat{\bH}^{\sf r}_{k,\ell}{\boldsymbol {\phi}}_{\ell}\right)^{\sf H} {\bf{f}}_k \right|^2}
    {\sum_{i=1, i \neq k}^{K}\!\left|\left(\sum_{\ell=1}^{L}\hat{\bH}^{\sf r}_{k,\ell}{\boldsymbol {\phi}}_{\ell}\right)^{\sf H}\!\!\!{\bf{f}}_i \right|^2 
    \!\!+\! 
    \sum_{i = 1}^{K}{\bf{f}}^{\sf{H}}_i
    \sum_{\ell=1}^{L}\left({\boldsymbol {\phi}}^{\sf T}_{\ell} \otimes \bI_N\right) \bR_{k,\ell}^{\sf e}\left({\boldsymbol {\phi}}^{\sf T}_{\ell} \otimes \bI_N\right)^{\sf H} {\bf{f}}_i\!+\! \frac{\sigma^2}{P}}\!\right)\right]
    % =  \bbE_{\{{{\hat{\bH}^{\sf r}_{k,\ell}}}\}}\left[\bar{R}^{\sf ins}_k\right]
    }
\end{align}
\hrule
\vspace{-1.2em}
\end{figure*}
%%%%%%%%%%%%%%%%%%%%%%%%%%%%%%%%%%%
%%%%%%%%%%%%%%%%%%%%%%%%%%%%%%%%%%%
{\color{black}
% Specifically, $(a)$ in \eqref{eq:R_lb_a} is obtained by treating the channel estimation error as an uncorrelated Gaussian random variable.
% When the estimated channel and the channel error are uncorrelated, the capacity lower bound can be derived \cite{hassibi2003much,yoo2006capacity} and has the same analytical form as \eqref{eq:R_lb_a}
% Furthermore, under the assumption of Gaussian signaling, this lower bound becomes tight and has the same analytical form as \eqref{eq:R_lb_a}.
{\color{black}Specifically, $(a)$ in \eqref{eq:R_lb_a} is obtained by treating the channel estimation error as worst-case uncorrelated additive Gaussian noise. 
According to \cite{hassibi2003much}, under the uncorrelation condition, the worst-case noise that minimizes mutual information for a given covariance is Gaussian, leading to the worst-case capacity expression in the form of \eqref{eq:R_lb_a}.}
For $(b)$ in \eqref{eq:R_lb_b}, we use Jensen's inequality\footnote{
{\color{black}When the number of BS antennas $N$ is sufficiently large, the gap between the exact SE and the Jensen's bound becomes negligible, making our bound a robust and reasonable approximation for large-scale MIMO systems \cite{park2017dynamic}.}}.
For $(c)$ in \eqref{eq:R_lb_c}, we apply the vectorization, i.e., ${\rm vec}(\bA \bB) = (\bB^{\sf T} \otimes \bI_a) {\rm vec}(\bA)$ with  $\bA \in \bbC^{a\times b}$ and $\bB \in \bbC^{b \times c}$, to obtain $\boldsymbol{{\bE}}_{k,\ell}\boldsymbol{\phi}_{\ell} = (\boldsymbol{\phi}_{\ell}^{\sf T} \otimes \bI_N) \be_{k,\ell}$, and then express the channel estimation error term as its covariance matrix $\bR^{\sf e}_{k,\ell}$.}
Using this result, we can then calculate the expectation across all channel error terms in the second term of the denominator in \eqref{eq:R_lb_a} as $\bbE\left[(\boldsymbol{\phi}_{\ell}^{\sf T}\!\otimes \bI_N) \be_{k,\ell}\be^{\sf H}_{k,\ell}(\boldsymbol{\phi}_{\ell}^{\sf T} \otimes \bI_N)^{\sf H}\right] \!=\! (\boldsymbol{\phi}_{\ell}^{\sf T}\! \otimes \bI_N)\bR_{k,\ell}^{\sf e}(\boldsymbol{\phi}_{\ell}^{\sf T} \otimes \bI_N)^{\sf H},\; \forall \ell,k$.
This completes the proof.
\end{proof}
\end{proposition}

% {\color{black}It is well known that Jensen’s inequality yields a tight bound at a low power regime {\color{red}[REF]}. 
% However, when the number of BS antennas $N$ is sufficiently large, the gap between the exact SE and the lower bound becomes negligible, making our bound a robust and reasonable approximation for large-scale MIMO systems {\color{red}[REF]}.}

Thanks to the expression in \eqref{eq:ins_SE_f}, we are able to leverage the partial channel knowledge, i.e., the channel estimate and its error covariance.
Aiming to maximize the sum SE by jointly optimizing the precoder and RIS phase shifts, we formulate the optimization problem with \eqref{eq:ins_SE_f} as
\begin{align} 
    \label{eq:main_problem}
    \mathop{{\text{maximize}}}_{{\bf{f}}_1, \cdots,{\bf{f}}_K, {\boldsymbol{\Phi}}}& \;\; \sum_{k = 1}^{K} \bar{R}^{\sf ins}_{k}({\bf F}, {\boldsymbol{\Phi}})
    \\
    \label{eq:power_constraint}
    {\text{subject to}} & \;\;  \sum_{k = 1}^{K} \left\| {\bf{f}}_k \right\|^2  \le 1,
    \\
    \label{eq:unitmodulus_constraint}
    & \;\; |\phi_{\ell,m}| = 1,\;\; \forall \ell \in \CMcal{L}, \forall m \in \CMcal{M},
\end{align}
where ${\boldsymbol{\Phi}} = [{\boldsymbol{\Phi}}_1, {\boldsymbol{\Phi}}_2, \cdots, {\boldsymbol{\Phi}}_L] \in \bbC^{M \times LM}$. 
Since \eqref{eq:main_problem} is non-convex, it is challenging to obtain the globally optimal solution.
Additionally, the non-convex unit-modulus constraint in \eqref{eq:unitmodulus_constraint} further aggravates the challenge.
These characteristics render the joint optimization of precoding and RIS phase shifts a highly challenging task.
Thus, it is necessary to establish the efficient optimization framework to tackle the problem.
% To this end, we introduce the GPI-based algorithm for solving \eqref{eq:main_problem} with respect to $\bF$ and $\boldsymbol{\Phi}$ in the following section.

%%%%%%%%%%%%%%%%%%%%%%%%%%%%%%%%%%%%%%%%%%%%%%%%%%%%%%%%%%%%
\section{Proposed Beamforming Design} \label{sec:Tx_design}
%%%%%%%%%%%%%%%%%%%%%%%%%%%%%%%%%%%%%%%%%%%%%%%%%%%%%%%%%%%%
In this section, we aim to maximize the sum SE by solving the joint optimization problem in \eqref{eq:main_problem}.
% subject to the transmit power constraint  \eqref{eq:power_constraint} and the unit-modulus constraint  \eqref{eq:unitmodulus_constraint} on the RIS elements.
To this end, we propose a joint optimization framework that decomposes the problem into two subproblems: $(i)$ the precoder $\bF$ optimization and $(ii)$ the RIS phase shifts  $\boldsymbol{\Phi}$ optimization.
We then solve these subproblems in an alternating manner under a unified framework to identify a superior local optimal solution.
{\color{black}Although this alternating strategy is inherently iterative and adds additional complexity, it is indispensable for decomposing the non-convex joint optimization problem into tractable subproblems.
Beyond the specific application to our system model, this approach has been widely validated, consistently yielding effective, high-quality solutions \cite{zhou2020intelligent, pan2022overview}.}

\subsection{Optimizing  Precoder ${\bF}$} \label{subsec:precoding}
In this subsection, we optimize the precoder $\bF$ while fixing the RIS phase shifts $\boldsymbol{\Phi}$.
To highlight this, we omit the notation ${\boldsymbol{\Phi}}$ from ${\bar{R}^{\sf ins}}_{k}({\bf F},{\boldsymbol{\Phi}})$ and thus, we have
\begin{align} 
    \label{eq:AO_precoder}
    \mathop{{\text{maximize}}}_{{\bf{f}}_1, \cdots,{\bf{f}}_K}& \; \sum_{k = 1}^{K}\;{\bar{R}^{\sf ins}}_{k}({\bf F})
    \\
    \label{eq:power_const2} 
    {\text{subject to}} & \;\sum_{k = 1}^{K} \left\| {\bf{f}}_k \right\|^2  \le 1.
\end{align}
Thanks to the reformulation of the original problem with the lower bound of instantaneous SE in \eqref{eq:ins_SE_f}, we can utilize the GPI approach in \cite{choi2019joint}.
To adopt the method, we first vectorize the precoding matrix as
\begin{align} 
    \label{eq:precodingvector}
    \bar{\mathbf{f}} = {\rm{vec}}(\bF) = \left[\mathbf{f}^\top_{1},\mathbf{f}^\top_{2},\ldots,\mathbf{f}^\top_{K}\right]^\top \in \mathbb{C}^{NK}.
\end{align}
Considering the maximum transmit power $P$, we set $\|\bar{\bf f}\|^2 = 1$ to achieve the maximum sum SE.
With the estimated cascaded channel of user $k$ denoted as $\hat{\bh}_k = \sum_{\ell=1}^{L} 
\hat{\bH}_{k,\ell}^{\sf r}{\boldsymbol {\phi}}_{\ell}$, we replace $\frac{\sigma^2}{P}$ with $\frac{\sigma^2}{P} \|{\bff}\|^2$ in \eqref{eq:ins_SE_f}, and rewrite the problem in \eqref{eq:AO_precoder} into more tractable expression, the function of Rayleigh quotient form, as
\begin{align} 
    \label{eq:Rayleigh_quotient}
    \mathop{{\text{maximize}}}_{\bar{\mathbf{f}}} \sum_{k=1}^{K}\log_2  \left(\frac{\bar{\mathbf{f}}^{\sf H}\mathbf{A}_{k}\bar{\mathbf{f}}}{\bar{\mathbf{f}}^{\sf H}\mathbf{B}_{k}\bar{\mathbf{f}}}\right),
\end{align}
where
\begin{align}
    \label{eq:A_k}
    &\mathbf{A}_k = \mathrm{blkdiag}\!\left(\hat{\bh}_{k} \hat{\bh}^{\sf H}_{k} + {\boldsymbol{\Xi}}_k,\cdots,\hat{\bh}_{k} \hat{\bh}^{\sf H}_{k} + {\boldsymbol{\Xi}}_k\right) + \mathbf{I}_{NK}\frac{\sigma^2}{P},
    \\
    \label{eq:B_k}
    &\mathbf{B}_k = \mathbf{A}_k - \mathrm{blkdiag} \left(\mathbf{0},\cdots, \hat{\bh}_{k}\hat{\bh}^{\sf H}_{k},\cdots, \mathbf{0}\right).
    % \\
    % \label{eq:Xi_k}
    % &{\boldsymbol{\Xi}}_k =  \sum_{\ell=1}^{L} \left({\boldsymbol {\phi}}^{\sf T}_{\ell} \otimes \bI_N \right)\bR_{k,\ell}^{\sf e}\left({\boldsymbol {\phi}}^{\sf T}_{\ell} \otimes \bI_N \right)^{\sf H}.
\end{align}
% Recall that $\bR^{\sf e}_{k,\ell}$ is the error covariance matrix.
The second term on the right hand side in \eqref{eq:B_k} has a nonzero block located at the $k$th block entry.
We can ignore the transmit power constraint \eqref{eq:power_constraint} because of $\|\bar{\bf f}\|^2 = 1$ and scaling invariance of $\bar{\bf f}$ in \eqref{eq:Rayleigh_quotient}.
% The product of Rayleigh quotient form in \eqref{eq:Rayleigh_quotient} is derived under the assumption $\|\bar{{\bf f}}\|^2=1$, and the problem in \eqref{eq:Rayleigh_quotient} is invariant up to the scaling of $\bar{\bf f}$.
% We focus on identifying the local points of the problem in \eqref{eq:Rayleigh_quotient} to find the superior point among all candidates. 
For simplicity, we define the objective function in \eqref{eq:Rayleigh_quotient} as
\begin{align}
    \label{eq:lamb_obj}
    \cL_{\sf BS}(\bar{\bf f}) = \log_2\prod_{k=1}^{K} \left(\frac{\bar{\mathbf{f}}^{\sf H}\mathbf{A}_{k}\bar{\mathbf{f}}}{\bar{\mathbf{f}}^{\sf H}\mathbf{B}_{k}\bar{\mathbf{f}}}\right)
    = \log_2 \lambda_{\sf BS}(\bar{\bf f}).
\end{align}
Then, we derive Lemma~\ref{lem:NEP} to find the  stationary points of \eqref{eq:lamb_obj}.
\begin{lemma}
    \label{lem:NEP}
    The stationary condition of the problem \eqref{eq:Rayleigh_quotient} is satisfied if the following holds:
    \begin{align}
        \label{eq:general_eigen}
        {\bar{\bB}}^{-1}(\bar{\mathbf{f}})\bar{\bA}(\bar{\mathbf{f}})\bar{\mathbf{f}} = \lambda_{\sf BS}(\bar{\mathbf{f}})\bar{\mathbf{f}},
    \end{align}
    where 
    \begin{align}
        \label{eq:A_KKT}
        &{\bar{\bA}}(\bar{\mathbf{f}}) = \lambda_{{\sf BS}, {\sf num}}(\bar{\mathbf{f}}) \sum_{k=1}^{K}\left(\frac{\mathbf{A}_k}{\bar{\mathbf{f}}^{\sf H}\mathbf{A}_{k}\bar{\mathbf{f}}}\right), 
        \\
        \label{eq:B_KKT}
        &{\bar{\bB}}(\bar{\mathbf{f}}) = \lambda_{{\sf BS}, {\sf den}}(\bar{\mathbf{f}}) \sum_{k=1}^{K}\left(\frac{\mathbf{B}_k}{\bar{\mathbf{f}}^{\sf H}\mathbf{B}_{k}\bar{\mathbf{f}}}\right),
        % \\
        % \label{eq:lamb_1}
        % &\lambda_{\sf BS}(\bar{\bff}) = \lambda_{{\sf BS}, \sf{num}}(\bar{\bff}) / \lambda_{{\sf BS}, \sf{den}}(\bar{\bff}),
        % \\
        % &\lambda_{{\sf BS}, {\sf num}}(\bar{\mathbf{f}}) = \prod_{k=1}^{K} \left(\bar{\mathbf{f}}^{\sf H}\mathbf{A}_{k}\bar{\mathbf{f}}\right), \lambda_{{\sf BS}, {\sf den}}(\bar{\mathbf{f}}) = \prod_{k=1}^{K} \left(\bar{\mathbf{f}}^{\sf H}\mathbf{B}_{k}\bar{\mathbf{f}}\right).
    \end{align}
    and  $\lambda_{{\sf BS}, {\sf num}}(\bar{\mathbf{f}})$ and $ \lambda_{{\sf BS}, {\sf den}}(\bar{\mathbf{f}})$ are any functions that meet $ \lambda_{{\sf BS}} (\bar{\mathbf{f}})=  \lambda_{{\sf BS}, {\sf num}}(\bar{\mathbf{f}})/ \lambda_{{\sf BS}, {\sf den}}(\bar{\mathbf{f}})$.
\end{lemma}
\begin{proof}
    Please refer to Appendix~\ref{pf:NEPv_precoder}. 
\end{proof}
% \textit{Proof.} See Appendix~\ref{pf:NEPv_precoder}. 

We note that the stationary condition in \eqref{eq:general_eigen} can be interpreted as a generalized eigenvalue problem.
Here, $\lambda_{\sf BS}(\bar{\bf f})$ is  as an eigenvalue of ${\bar{\bB}}^{-1}(\bar {\bf{f}}) \bar{\bA}(\bar {\bf{f}})$ with $\bar{\bf f}$ as a corresponding eigenvector.
As a result, maximizing the objective function $\cL_{\sf BS} (\bar{\bf f})$ is equivalent to maximizing $\lambda_{\sf BS}(\bar{\bf f})$.
Therefore, it is desirable to find the principal eigenvalue of \eqref{eq:general_eigen} to maximize \eqref{eq:lamb_obj}, which is equivalent to finding the superior local optimal solution of \eqref{eq:AO_precoder}.
Based on \eqref{eq:general_eigen}, we propose the sum SE maximization precoding algorithm by employing the GPI method \cite{choi2019joint, choi2022energy}
As described in Algorithm~\ref{alg:algorithm_1}, we first initialize $\bar \bff^{(0)}$ and update $\bar \bff^{(t)}$ at each iteration with given $\boldsymbol{\Phi}$; the algorithm computes $\bar{\bA} (\bar {\bf{f}}^{(t-1)})$ and $\bar{\bB} (\bar {\bf{f}}^{(t-1)})$ according to \eqref{eq:A_KKT} and \eqref{eq:B_KKT}. 
% Then, the algorithm updates $\bar{\bff}_{(t)}$ as
% \begin{align}
%     \bar{\bff} \leftarrow \bar{\bB}^{-1} (\bar {\bf{f}}^{(t-1)})\bar{\bA} (\bar {\bf{f}}^{(t-1)}) \bar {\bf{f}}^{(t-1)}.
% \end{align}
Then, the algorithm updates $\bar{\bff} \leftarrow \bar{\bB}^{-1} (\bar {\bf{f}}^{(t-1)})\bar{\bA} (\bar {\bf{f}}^{(t-1)}) \bar {\bf{f}}^{(t-1)}$.
The algorithm normalizes the updated precoding vector by $\bar {\bf{f}}^{(t)} = \bar {\bf{f}}^{(t)}/ \left\| \bar {\bf{f}}^{(t)} \right\|$.
We repeat these steps until either $\bar {\bf{f}}^{(t)}$ converges to a tolerance level (e.g. $\|\bar{\bff}_{(t)} - \bar{\bff}_{(t-1)}\| < \varepsilon_1$ for $\varepsilon_1>0$) or the algorithm reaches  $t_{1, \rm max}$.

%%%%%%%%%%%%%%%%%%%%%%%%%%%%%%%%%%%%%%%%%%%%%%%%%%%%%%%%%%%%
%%%%%%%%%%%%%%%%%%%%%%%%%%%%%%%%%%%%%%%%%%%%%%%%%%%%%%%%%%%%
\begin{algorithm}[t]
    \caption{GPI-Based Precoding Algorithm}
    \label{alg:algorithm_1} 
    {\bf{initialize}}: $\bF^{(0)}$.
    \\
    Set $\bar {\bf{f}}^{(0)} = {\rm vec}(\bF^{(0)})$ and $t= 1$.
    \\
    \While {$\left\|\bar {\bf{f}}^{(t)} - \bar {\bf{f}}^{(t-1)} \right\| > \varepsilon_1$ {\rm or} $t \leq t_{1,\max} $}{
    Build $\bar{\bA}(\bar {\bf{f}}^{(t-1)})$ and $\bar{\bB}  (\bar {\bf{f}}^{(t-1)})$ according to \eqref{eq:A_KKT} and \eqref{eq:B_KKT} for given $\boldsymbol{\Phi}$.
    \\
    Compute $\bar {\bf{f}}^{(t)} = \bar{\bB}^{-1} (\bar {\bf{f}}^{(t-1)}) \bar{\bA} (\bar {\bf{f}}^{(t-1)}) \bar {\bf{f}}^{(t-1)}$. 
    \\
    Normalize $\bar {\bf{f}}^{(t)} = \bar {\bf{f}}^{(t)}/\left\| \bar {\bf{f}}^{(t)}\right\|$.
    \\
    $t \leftarrow t+1$.}
    % $\bar {\bf{f}}^{\star} = \left[{\bf{f}}_1^{\sf T}, {\bf{f}}_2^{\sf T}, \dots, {\bf{f}}_K^{\sf T} \right]^{\sf T} \leftarrow \bar {\bf{f}}^{(t)}$.
    % \\
    \Return{\ }{$\bF^{\star} \leftarrow \bar {\bf{f}}^{(t)}$}.
\end{algorithm}
%%%%%%%%%%%%%%%%%%%%%%%%%%%%%%%%%%%%%%%%%%%%%%%%%%%%%%%%%%%%
%%%%%%%%%%%%%%%%%%%%%%%%%%%%%%%%%%%%%%%%%%%%%%%%%%%%%%%%%%%%

% ----------------------------------------------------------------
\subsection{Optimizing RIS Phase Shifts ${\boldsymbol{\Phi}}$} \label{subsec:RIS}
% ----------------------------------------------------------------
In this subsection, we put forth a multi-RIS phase-shift optimization method for given $\bF$.
{\color{black}
We have two primary motivations for developing the RIS optimization method within the GPI framework.
First, the GPI approach offers a significant advantage over traditional optimization methods by enabling us to identify not just any local optimal solution, but a superior local optimum. This capability enhances the overall performance of our RIS system.
Second, we aim to achieve efficient scaling with multiple RISs. 
This is made possible by leveraging the block-diagonal structure of matrices within the GPI method, for example, \eqref{eq:A_KKT} and \eqref{eq:B_KKT}. 
This structural property, which we will discuss later in Remark~\ref{rm:complexity}, allows us to handle multiple RISs  without a prohibitive increase in computational complexity.
% We are also motivated to  develop the RIS optimization method under the GPI framework because: ($i$) the GPI approach allows to identify not just a local optimal solution, but a superior local optimal solution, and ($ii$)
% we aim to achieve multi-RIS scalability by utilizing a block-diagonal structure in the GPI method which is observed in $\bar {\bf B}(\bar {\bf f})$ \eqref{eq:B_KKT}, which will be discussed later in Remark~\ref{rm:complexity}.
}

Unlike the precoding optimization, however, it is required to  handle the unit-modulus constraint on ${\boldsymbol{\Phi}}$.
To this end, we first reformulate \eqref{eq:ins_SE_f} to a quadratic form with respect to $\boldsymbol{\phi}_{\ell}$, providing a more tractable approach in optimizing the RIS phase shifts.
% {\color{black}Using ${\rm vec}(\bA \bB) = (\bB^{\sf T} \otimes \bI_a) {\rm vec}(\bA)$ with  $\bA \in \bbC^{a\times b}$ and $\bB \in \bbC^{b \times c}$} 
Using the vectorization and introducing a permutation matrix, we convert the second term of the denominator in \eqref{eq:R_lb_a} to
\begin{align}
    % &\left|\sum_{i = 1}^{K}{\bf{f}}^{\sf{H}}_i\left(\sum_{\ell=1}^{L}\bE_{k,\ell} {\boldsymbol {\phi}}_{\ell}\right)\right|^2
    &\sum_{i = 1}^{K} {\bf{f}}^{\sf{H}}_i\left(\sum_{\ell=1}^{L}\bE_{k,\ell} {\boldsymbol {\phi}}_{\ell} {\boldsymbol {\phi}}_{\ell}^{\sf H} \bE_{k,\ell}^{\sf H}\right) {\bf{f}}_i
    \nonumber\\
    &\!=\! \sum_{\ell = 1}^{L}\!\sum_{i=1}^{K}
    {\boldsymbol {\phi}}_{\ell}^{\sf H} \left({\bff_i}^{\sf T}\!\!\otimes \bI_M \right)\!{\rm vec}\left(\bE^{\sf T}_{k,\ell}\right)^{*}\!\!\left({\rm vec}\left(\bE^{\sf T}_{k,\ell}\right)^{*}\right)^{\sf H}\!\! \left({\bff_i}^{\sf T}\!\!\otimes \bI_M \right)^{\sf H}\!\!\!{\boldsymbol {\phi}}_{\ell}
    \nonumber\\
    \label{eq:qd_phi}
    &\!=\! \sum_{\ell = 1}^{L}\!\sum_{i=1}^{K}
    {\boldsymbol {\phi}}_{\ell}^{\sf H} \left({\bff_i}^{\sf T} \otimes \bI_M \right)\!\bP \be_{k,\ell}^{*}\!\left(\bP \be_{k,\ell}^{*}\right)^{\sf H}\!\! \left({\bff_i}^{\sf T} \otimes \bI_M \right)^{\sf H}\!\!\!{\boldsymbol {\phi}}_{\ell},
\end{align}
where $\bP$ is a permutation matrix satisfying ${\rm vec}(\bE^{\sf T}) = \bP {\rm vec}(\bE)$.
By applying \eqref{eq:qd_phi} to the similar process of \eqref{eq:R_lb_a}-\eqref{eq:R_lb_c}, the lower bound of instantaneous SE can be reformulated as {\color{black}$\tilde{R}^{\sf ins}_{k}({\boldsymbol{\Phi}})=\bar{R}^{\sf ins}_{k}({\bf F}, {\boldsymbol{\Phi}})$, where}
\begin{align}
    \label{eq:ins_SE_phi}
    &\tilde{R}^{\sf ins}_k({\boldsymbol{\Phi}}) = 
    \\ \nonumber
    & \log_2 \!\left(\!1 \!+\! \frac{\left|\left(\sum_{\ell=1}^{L}\hat{\bH}^{\sf r}_{k,\ell}{\boldsymbol {\phi}}_{\ell}\right)^{\sf H} {\bf{f}}_k \right|^2}
    {\sum_{i\neq k}^{K}\!\left|\left(\sum_{\ell=1}^{L}\!\hat{\bH}^{\sf r}_{k,\ell}{\boldsymbol {\phi}}_{\ell}\!\right)^{\sf H}\!\!\!{\bf{f}}_i \right|^2 
    \!\!+\! 
    \sum_{\ell = 1}^{L}
    {\boldsymbol {\phi}}_{\ell}^{\sf H}\boldsymbol{\Theta}_{k,\ell} {\boldsymbol {\phi}}_{\ell}
    \!+\! \frac{\sigma^2}{P}}\!\right),
    % &\log_2 \!\left(\frac{\sum_{i=1}^K \left|{\bf{f}}_i^{\sf H} \left(\sum_{\ell=1}^{L}\hat{\bH}^{\sf r}_{k,\ell}{\boldsymbol {\phi}}_{\ell}\right) \right|^2 +\! 
    % \sum_{\ell = 1}^{L}
    % {\boldsymbol {\phi}}_{\ell}^{\sf H}\boldsymbol{\Theta}_{k,\ell} {\boldsymbol {\phi}}_{\ell}
    % \!+\! \frac{\sigma^2}{P}}
    % {\sum_{i=1, i\neq k}^{K}\!\left|{\bf{f}}_i^{\sf H}\left(\sum_{\ell=1}^{L}\hat{\bH}^{\sf r}_{k,\ell}{\boldsymbol {\phi}}_{\ell}\right) \right|^2 
    % \!\!+\! 
    % \sum_{\ell = 1}^{L}
    % {\boldsymbol {\phi}}_{\ell}^{\sf H}\boldsymbol{\Theta}_{k,\ell} {\boldsymbol {\phi}}_{\ell}
    % \!+\! \frac{\sigma^2}{P}}\!\right),
\end{align}
and $\boldsymbol{\Theta}_{k,\ell} = \sum_{i=1}^{K} \left({\bff_i}^{\sf T} \otimes \bI_M \right) \bP\left(\bR^{\sf e}_{k,\ell}\right)^{\sf *} \bP^{\sf T} \left({\bff_i}^{\sf T} \otimes \bI_M \right)^{\sf H}$.
With given $\bF$, our problem in \eqref{eq:main_problem} is rewritten as
\begin{align} 
    \label{eq:main_problem_2}
    \mathop{{\text{maximize}}}_{{\boldsymbol{\Phi}}}& \;\; \sum_{k = 1}^{K} \tilde{R}^{\sf ins}_{k}({\boldsymbol{\Phi}})
    \\
    \label{eq:unitmodulus_constraint_2}
    {\text{subject to}} & \;\; |\phi_{\ell,m}| = 1,\; \forall \ell \in \CMcal{L}, \forall m \in \CMcal{M}.
\end{align}
% In this problem, it is challenging to deal with the unit-modulus constraint as well as the non-convex objective function.

Now, to resolve the challenge of the unit-modulus constraint, we introduce a regularization approach.
We first stack all individual RIS elements as a vector form as 
\begin{align}
    \bar{\boldsymbol{\phi}} \!=\! \left[{\boldsymbol{\phi}}_1^{\sf T}, \boldsymbol{\phi}_2^{\sf T}, \cdots, \boldsymbol{\phi}_L^{\sf T} \right]^{\sf T} \!\!\!=\!  \left[{\phi}_{1,1}, {\phi}_{1,2}, \cdots, {\phi}_{L,M}\right]^{\sf T}\!\! \in \bbC^{LM}\!\!.
\end{align}
% where $\boldsymbol{\phi}_{\ell} = {\rm vec}({\boldsymbol{\Phi}}_{\ell})$.
We relax the unit-modulus constraint considering the difference between maximum and minimum values among all RIS phase shifts of $\bar{\boldsymbol{\phi}}$.
Then, the regularized optimization problem can be formulated without the unit-modulus constraint as
\begin{align}
    \label{eq:AO_RIS}
    \mathop{{\text{maximize}}}_{{\boldsymbol{\Phi}}, \mu}& \frac{1}{R_{\Sigma}}\!\sum_{k = 1}^{K}\!\tilde{R}^{\sf ins}_{k}({\boldsymbol{\Phi}}) \!-\! \frac{\mu}{\tau}\left(\max_{\substack{\ell \in \CMcal{L}\\ m\in \CMcal{M}}} |{\phi}_{\ell,m}|^2 \!-\!\min_{\substack{\ell \in \CMcal{L}\\ m\in \CMcal{M}}} |{\phi}_{\ell,m}|^2 \right),\!\!
    % \\
    % \label{eq:unitmodulus_constraint_3}
    % {\text{subject to}}  & \; |\phi_{\ell,m}| = 1,\;\; \ell \in \CMcal{L},\; m\in \CMcal{M},
\end{align}
where $\mu$ denotes a parameter of regularization, $\tau$ is a pre-defined normalization factor for the RIS phase shifts, and $R_{\Sigma}$ is a pre-defined normalization factor for the sum SE obtained by any existing state-of-the-art precoder with randomly generated $\boldsymbol{\Phi}$.
We remark that the normalization factors are introduced to make the SE and regularization term comparable in scale, thereby reducing the effective range of $\mu$.
In \eqref{eq:AO_RIS}, the second term, commonly referred to as the penalty term, serves to enforce the unit-modulus constraint by minimizing the difference between  maximum and minimum elements of $\bar{\boldsymbol{\phi}}$. 
% The regularization parameter $\mu$ controls the importance of adhering to the unit-modulus constraint relative to maximizing the sum SE.
The regularization parameter $\mu$ controls the degree of adherence to the unit-modulus constraint, ensuring that the amplitudes of relaxed RIS phase shifts remain as homogeneous as possible before projecting them into a feasible solution set, i.e., a complex unit circle.

We now tackle \eqref{eq:AO_RIS} with respect to ${\bf w}$ for given $\mu$. 
To transform \eqref{eq:AO_RIS} into a more tractable form, we first normalize the RIS phase shifts vector as
\begin{align}
    {\bw}=\frac{1}{\sqrt{L M}}\bar{\boldsymbol{\phi}} = [w_1, w_2, \cdots, w_{LM}]^{\sf T} \in \bbC^{LM}.
\end{align}
With some abuse of notation, we then rewrite the objective function in \eqref{eq:AO_RIS} as
\begin{align}
    \label{eq:L_2}
    \CMcal{L}_{\sf RIS}(\bw) \!=\! \frac{1}{R_{\Sigma}}\!\sum_{k = 1}^{K}\! \tilde{R}^{\sf ins}_{k}(\bw) \!-\! \frac{\mu}{\tau} \left(\max_{i = 1,\ldots, LM}\!|w_i|^2 \!-\! \min_{i = 1,\ldots,LM}\!|w_i|^2 \right)\!.
\end{align}
{\color{black}We set $\tau = (LM)^{-1}$ in this problem.}
Let us now introduce a $LM \times LM$ diagonal matrix as 
\begin{align}
    \label{eq:X_mat}
    \bX_i = {\rm diag}\big(0,\ldots, \underbrace{1}_{i{\rm th\;term}}, \ldots, 0\big),\; i=1,\ldots, LM,
    % \bX_i = {\rm diag}(0,\ldots, 1, \ldots, 0),\; i=1,\ldots, LM,
\end{align}
where the non-zero entry appears in $i$th diagonal position.
By leveraging \eqref{eq:X_mat}, we can rewrite the power of the normalized RIS phase shifts as a quadratic form as $|w_i|^2 = \bw^{\sf H} \bX_i \bw$.
To transform \eqref{eq:L_2} into the GPI-friendly form for ${\bf w}$, we assume $\|\bw\|=1$ {\color{black} which is naturally true when the unit-modulus constraint is met.}
% Although we relax the unit-modulus constraint, this assumption is valid for an optimized value of $\bw$.
As a result, we can reformulate the objective function in \eqref{eq:AO_RIS} as 
\begin{align}
    \label{eq:L_2_Rayleigh}
    \CMcal{L}_{\sf RIS}&(\bw) =
    \frac{1}{R_{\Sigma}}\sum_{k = 1}^{K} \log_2 \left(\frac{{\bw}^{\sf H}\bC_k {\bw}}{{\bw}^{\sf H}\bD_k {\bw}}\right) 
    \nonumber 
    \\
    & - \frac{\mu}{\tau} \left(\max_{i= 1,\ldots, LM} \left\{{\bw}^{\sf H} \bX_{i}{\bw}\right\} - \min_{i= 1,\ldots, LM} \left\{{\bw}^{\sf H}\bX_{i}{\bw}\right\}\right),
    % \CMcal{L}_{\sf RIS}&(\bw) =
    % \frac{1}{R_{\Sigma}}\sum_{k = 1}^{K} \log_2 \left(\frac{{\bw}^{\sf H}\bC_k {\bw}}{{\bw}^{\sf H}\bD_k {\bw}}\right) 
    % \nonumber 
    % \\
    % & - \mu \left(\max_{i= 1,\ldots, LM} \left\{\frac{{\bw}^{\sf H} \bX_{i}{\bw}}{\bw^{\sf H}\bw}\right\} - \min_{i= 1,\ldots, LM} \left\{\frac{{\bw}^{\sf H} \bX_{i}{\bw}}{\bw^{\sf H}\bw}\right\}\right),
\end{align}
where
\begin{align}
    \label{eq:C_k}
    &\bC_k \!\!=\! L M \mathrm{blkdiag}\!\left(\Upsilon_{k,1} \!+\! \boldsymbol{\Theta}_{k,1}, \!\cdots\!, \Upsilon_{k,L} \!+\! \boldsymbol{\Theta}_{k,L}\right) \!+\! \frac{\sigma^2}{P} \bI_{LM},
    \\
    \label{eq:D_k}
    &\bD_k \!\!=\! L M \mathrm{blkdiag}\!\left(\bar{\Upsilon}_{k,1} \!+\! \boldsymbol{\Theta}_{k,1}, \!\cdots\!, \bar{\Upsilon}_{k,L} \!+\! \boldsymbol{\Theta}_{k,L}\right) \!+\! \frac{\sigma^2}{P} \bI_{LM},
    \\
    &\Upsilon_{k,\ell} = {\hat{\bH}_{k,\ell}^{\sf r H}} \bQ \hat{\bH}^{\sf r}_{k,\ell},\; \bar{\Upsilon}_{k,\ell} = {\hat{\bH}_{k,\ell}^{\sf r H}} \bar{\bQ}_k \hat{\bH}^{\sf r}_{k,\ell},
    \\
    &\bQ = \sum_{i=1}^K \bff_i \bff_i^{\sf H},\; \bar{\bQ}_k = \bQ - \bff_k \bff_k^{\sf H}.
\end{align}
In \eqref{eq:L_2_Rayleigh}, we need to approximate the non-smooth functions such as  $\rm{max}(\cdot)$ and $\rm{min}(\cdot)$ for finding the optimality condition. 
To this end, we adopt a LogSumExp approach \cite{shen2010dual}:
\begin{align}
    \label{eq:logsumexp1}
    &\min_{i = 1,...,J}\{x_i\}  \approx -\alpha \ln\left(\sum_{i = 1}^{J} \exp\left( \frac{x_i}{-\alpha}  \right)\right),
    \\
    \label{eq:logsumexp2}
    &\max_{i = 1,...,J}\{x_i\}  \approx \alpha \ln \left(\sum_{i = 1}^{J} \exp\left( \frac{x_i}{\alpha} \right)\right),
\end{align}
where the approximation becomes tight as $\alpha \rightarrow +0$ for both cases. 
% Using the LogSumExp, \eqref{eq:AO_RIS} is approximated as
% \begin{align}
%     \nonumber 
%     &\CMcal{L}_2({\bw}) \approx \Tilde{\CMcal{L}}_2({\bw})
%     =  \frac{1}{R_{\Sigma}}\sum_{k = 1}^{K} \log_2 \left(\frac{{\bw}^{\sf H}\bC_k {\bw}}{{\bw}^{\sf H}\bD_k {\bw}}\right) -
%     \nonumber % \label{eq:L_tilde} 
%     \\
%     \nonumber
%     & \frac{\mu}{\tau}\! \left\{\!\frac{1}{\alpha_1}\!\ln\!\left(\sum_{i=1}^{LM} \exp\left(\alpha_1{\bw}^{\sf H}\bX_{i}{\bw}\right)\right) 
%     \!-\! \alpha_2\!\ln\!\left(\sum_{i=1}^{LM} \exp\left({\frac{{\bw}^{\sf H}\bX_{i}{\bw}}{-\alpha_2}}\right)\right)\!\right\}
%     \\
%     \label{eq:Rayleigh_2}
%     &= \log_2 \prod_{k=1}^{K} \left(\frac{{\bw}^{\sf H}\bC_k {\bw}}{{\bw}^{\sf H}\bD_k {\bw}}\right)^{\frac{1}{R_{\Sigma}}} \left(\sum_{i=1}^{LM} \exp\left(\alpha_1{\bw}^{\sf H}\bX_{i}{\bw}\right)\right)^{-\frac{\mu \ln 2}{\tau \alpha_1}}
%     \nonumber \\
%     &\qquad \times \left(\sum_{i=1}^{LM} \exp\left({\frac{{\bw}^{\sf H}\bX_{i}{\bw}}{-\alpha_2}}\right)\right)^{\frac{\alpha_2 \mu \ln 2}{\tau}} = \log_2 \lambda_2 (\bw)
% \end{align}
Using the LogSumExp, \eqref{eq:L_2_Rayleigh} is approximated as
\begin{align}
    \label{eq:Rayleigh_2}
    \cL_{\sf RIS}({\bw}) \approx \Tilde{\cL}_{\sf RIS}({\bw}) = \log_2 \lambda_{\sf RIS} (\bw),
\end{align}
where
\begin{align}
     % \nonumber
     % &\lambda_{\sf RIS} (\bw) \!=\! \prod_{k=1}^{K} \left(\frac{{\bw}^{\sf H}\bC_k {\bw}}{{\bw}^{\sf H}\bD_k {\bw}}\right)^{\frac{1}{R_{\Sigma}}} \!\left(\sum_{i=1}^{LM} \exp\left(\alpha_1\frac{{\bw}^{\sf H} \bX_{i}{\bw}}{\bw^{\sf H}\bw}\right)\right)^{-\frac{\mu \ln 2}{\tau \alpha_1}} 
     % \\\label{eq:lambda_RIS}
     % &\qquad\qquad \times \left(\sum_{i=1}^{LM} \exp\left(\frac{{\bw}^{\sf H} \bX_{i}{\bw}}{-\alpha_2 \bw^{\sf H}\bw}\right)\right)^{-\frac{\alpha_2 \mu \ln 2}{\tau}}.
    \nonumber
    &\lambda_{\sf RIS} (\bw) \!=\! \prod_{k=1}^{K} \left(\frac{{\bw}^{\sf H}\bC_k {\bw}}{{\bw}^{\sf H}\bD_k {\bw}}\right)^{\frac{1}{R_{\Sigma}}} \!\left(\sum_{i=1}^{LM} \exp\left(\frac{{\bw}^{\sf H}\bX_{i}{\bw}}{\alpha_1}\right)\right)^{-\frac{\alpha_1 \mu \ln 2}{\tau}} 
    \\ \label{eq:lambda_RIS}
    &\qquad\qquad \times \left(\sum_{i=1}^{LM} \exp\left({\frac{{\bw}^{\sf H}\bX_{i}{\bw}}{-\alpha_2}}\right)\right)^{-\frac{\alpha_2 \mu \ln 2}{\tau}}.
\end{align}
\begin{figure}[!t]
    \centering
    $\begin{array}{c}
    {\resizebox{0.75\columnwidth}{!}{\includegraphics{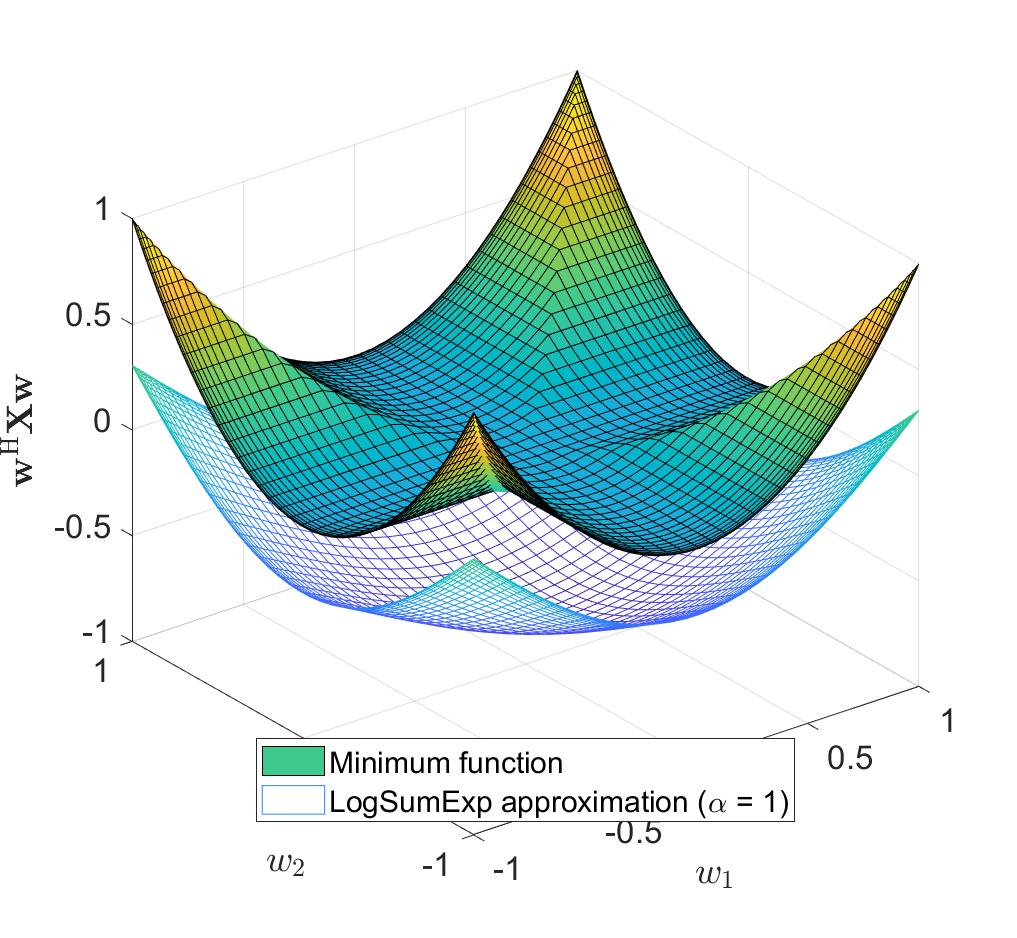}}
    } \\ \mbox{\small (a) $\min (\cdot)$ vs. LogSumExp with $\alpha = 1$}
    \end{array}$
    \vspace{-0.1cm}
    $\begin{array}{c}
    {\resizebox{0.75\columnwidth}{!}{\includegraphics{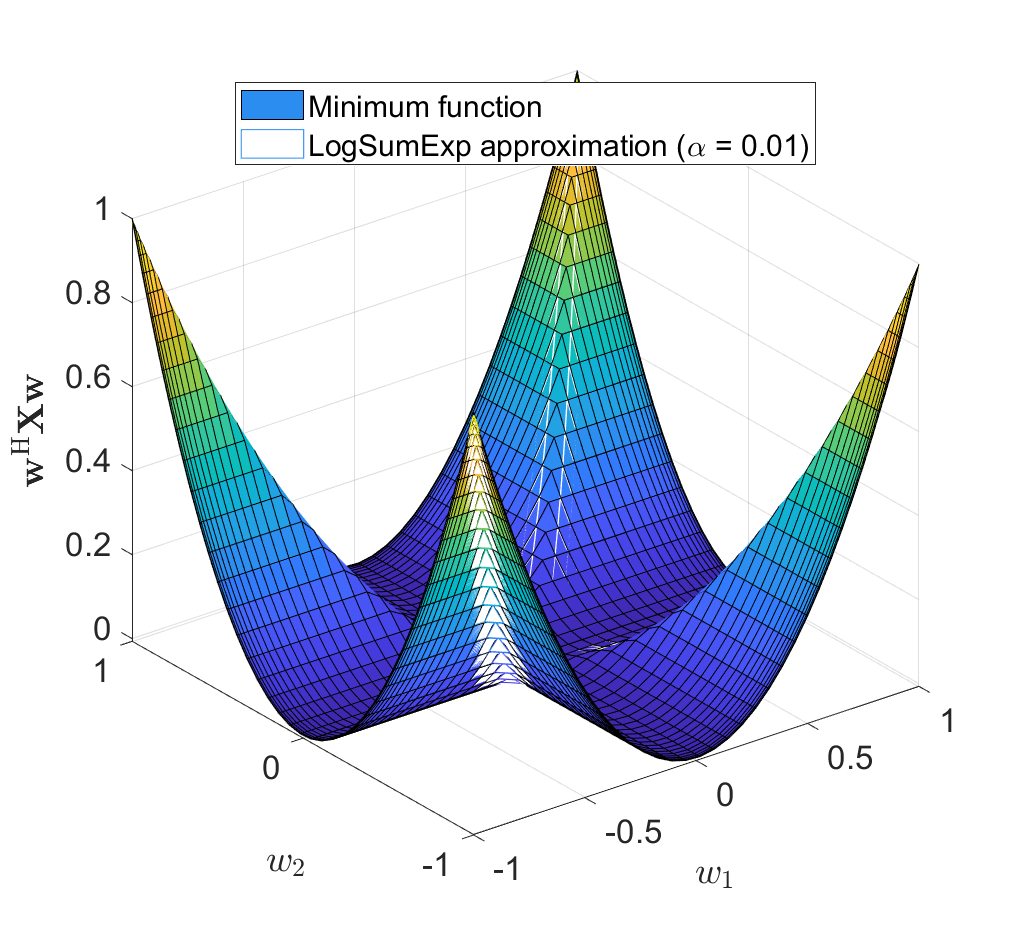}}
    }\\  \mbox{\small (b) $\min (\cdot)$ vs. LogSumExp with $\alpha = 0.01$}
    \end{array}$
    \caption{\color{black}An illustration of the comparison between the true minimum function and the LogSumExp approximation.}
    \label{fig:LSE}
\end{figure}
% To further illustrate the behavior of the non-smooth  function used in the approximation, we present Fig.~\ref{fig:LSE}, which compares LogSumExp with the exact minimum function.}
% {\color{magenta}
% This comparison demonstrates the trade-off in the LogSumExp approximation: a lower value of $\alpha$ tightens the approximation accuracy but induces steep gradients that may compromise numerical stability, whereas a higher $\alpha$ ensures smoother gradients for stability at the cost of accuracy.}
% {\color{black}Consequently, the {\color{black}chosen parameter succeeds in offering a good balance between precision and stability}.

{\color{black}
To further illustrate the behavior of the non-smooth  function used in the approximation, we present Fig.~2, which compares LogSumExp with the exact minimum function.
This comparison demonstrates the trade-off in the LogSumExp approximation: a lower value of $\alpha$ tightens the approximation accuracy yet induces steep gradients that may compromise numerical stability, whereas a higher $\alpha$ ensures smoother gradients for stability at the cost of accuracy.}

{\color{black}
\!} {\color{black}
We further analyze the convergence properties of our GPI-based algorithm from a \textit{a fixed-point iteration perspective} \cite{nesterov2005smooth} to rigorously explain why an extremely small value of $\alpha$ leads to instability.
% Specifically, the GPI-based approach updates the solution iteratively based on the first-order stationarity condition. 
% This can be viewed as a fixed-point iteration mapping $\CMcal{T}: \bbC^{LM} \to \bbC^{LM}$:
% \begin{align}
%     \bw^{(t+1)} = \CMcal{T}(\bw^{(t)}),
% \end{align}
% where the mapping $\CMcal{T}$ is derived from the stationary condition $\nabla \tilde{\cL}_{\sf RIS}(\bw) = 0$. 
Our GPI-based algorithm updates the solution iteratively using a mapping function $\CMcal{T}: \bbC^{LM} \to \bbC^{LM}$. 
% derived from the first-order stationarity condition
Since the update rule takes the form $\bw^{(t+1)} = \CMcal{T}(\bw^{(t)})$, the algorithm constitutes a fixed-point iteration.
For the algorithm to converge stably to a fixed point $\bw^{\star}$, the mapping $\CMcal{T}$ must be a contraction mapping in the neighborhood of $\bw^{\star}$. 
A sufficient condition for local convergence is that the spectral radius of the Jacobian matrix of $\CMcal{T}$ must be less than 1:
\begin{align}
    \rho({\boldsymbol{\CMcal{J}}}_{\CMcal{T}}(\bw^{\star})) < 1,
\end{align}
where $\boldsymbol{\CMcal{J}}_{\CMcal{T}}$ indicates the Jacobian of $\CMcal{T}$ defined as $\boldsymbol{\CMcal{J}}_{\CMcal{T}}(\bw) = \nabla \CMcal{T} (\bw)$ and $\rho (\cdot)$ indicates the spectral radius.

Although our GPI update is not a simple gradient descent step, we can draw a direct analogy to the standard gradient-based fixed-point iteration to understand the impact of curvature, i.e., 
\begin{align}
    \bw^{(t+1)} = \bw^{(t)} + \eta \nabla \tilde{\cL}_{\sf RIS}(\bw),
\end{align}
where $\eta$ is a step size.
In this case, the Jacobian is governed by the Hessian of the objective function:
\begin{align}
    \boldsymbol{\CMcal{J}}_{\CMcal{T}}(\bw) \approx \bI + \eta \nabla^2 \tilde{\cL}_{\sf RIS}(\bw).
\end{align}
Stability requires that the eigenvalues of the Hessian, $\lambda_i\left(\nabla^2 \tilde{\cL}_{\sf RIS}(\bw)\right)$, are not excessively large. 
However, the Hessian of LogSumExp is inversely proportional to the smoothing parameter $\alpha$ where $\alpha = \alpha_1 = \alpha_2$:
\begin{align}
    \big\|\nabla^2 \text{LogSumExp}(\bw; \alpha)\big\| \propto \frac{1}{\alpha}.
\end{align}
This implies that as $\alpha \to 0$, the curvature of the objective function becomes arbitrarily large, i.e., forming a sharp edge.
Consequently, an extremely small $\alpha$ induces a sharp curvature in the optimization landscape. In such regions, the Jacobian $\boldsymbol{\CMcal{J}}_{\CMcal{T}}(\bw)$ typically exhibits large eigenvalues, causing the spectral radius to exceed 1 ($\rho(\boldsymbol{\CMcal{J}}_{\CMcal{T}}) > 1$).
Under these conditions, the fixed-point iteration becomes locally expansive rather than contractive, leading to oscillations or divergence.
Therefore, there exists a fundamental trade-off: while a smaller $\alpha$ reduces the approximation error, it increases the numerical sensitivity of the problem, compromising the stability of the iterative algorithm.}
%%%%%%%%%%%%%%%%%%%%%%%%%%%%%%
% Consequently, the chosen parameter $\alpha$ in LogSumExp strikes a balance between precision and stability.

Finally, the regularized optimization problem in \eqref{eq:AO_RIS} is reformulated and approximated for given $\mu$ as
\begin{align} 
    \label{eq:AO_RIS_2}
    \mathop{{\text{maximize}}}_{\bw}&\; \Tilde{\cL}_{\sf RIS}({\bw}).
    % \\
    % \label{eq:unit_modulus_w}
    % {\text{subject to}}  & \; \sqrt{LM}|w_{i}| = 1,\;\; i= 1, \ldots, LM.
\end{align}
Then, similar to the precoding optimization, we derive {Lemma~\ref{lem:NEP_RIS}} for the approximated problem in \eqref{eq:AO_RIS_2}  to find the stationary points of \eqref{eq:lambda_RIS}.
\begin{lemma}
    \label{lem:NEP_RIS}
The stationary condition of \eqref{eq:AO_RIS_2} is also satisfied if the following holds:
\begin{align}
    \label{eq:general_eigen_2}
    \bar{\bD}^{-1}(\bw) \bar{\bC}(\bw){\bw} = \lambda_{\sf RIS}(\bw)\bw,
\end{align}
where
\begin{align}
    \label{eq:C_KKT}
    &\bar{\bC}(\bw) = \lambda_{{\sf RIS},{\sf num}}(\bw)\times
    \nonumber\\ 
    &\quad \left[\frac{1}{R_{\Sigma}\ln{2}}\sum_{k = 1}^{K} \left(\frac{\bC_k}{{\bw}^{\sf H}\bC_k {\bw}}\right) + \frac{\mu}{\tau} \frac{\sum_{i=1}^{LM} {\bX}_i e^{\frac{{\bw}^{\sf H}{\bX}_{i}{\bw}}{-\alpha_2}}}{\sum_{i=1}^{LM} e^{\frac{{\bw}^{\sf H}{\bX}_{i}{\bw}}{-\alpha_2}}}\right], 
    \\
    \label{eq:D_KKT}
    &\bar{\bD}(\bw) = \lambda_{{\sf RIS},{\sf den}}(\bw) \times
    \nonumber\\ 
    &\quad \left[\frac{1}{R_{\Sigma}\ln{2}}\sum_{k = 1}^{K} \left(\frac{\bD_k}{{\bw}^{\sf H}\bD_k {\bw}}\right) + \frac{\mu}{\tau} \frac{\sum_{i=1}^{LM} {\bX}_i e^{\frac{{\bw}^{\sf H}{\bX}_{i}{\bw}}{\alpha_1}}}{\sum_{i=1}^{LM} e^{\frac{{\bw}^{\sf H}{\bX}_{i}{\bw}}{\alpha_1}}}\right].
    % \label{eq:C_KKT}
    % &\bar{\bC}(\bw) = \lambda_{{\sf RIS},{\sf num}}(\bw)\times
    % \nonumber\\ 
    % &\left[\frac{1}{R_{\Sigma}\ln{2}}\sum_{k = 1}^{K} \left(\frac{\bC_k}{{\bw}^{\sf H}\bC_k {\bw}}\right) + \mu \frac{\sum_{i=1}^{LM} \frac{1}{\bw^{\sf H}\bw}\left(\bX_i - \frac{\bw^{\sf H}\bX_i \bw}{\bw^{\sf H}\bw}\right) e^{\frac{\bw^{\sf H}\bX_i \bw}{-\alpha_2 \bw^{\sf H}\bw}}}{\sum_{i=1}^{LM} e^{\frac{\bw^{\sf H}\bX_i \bw}{-\alpha_2 \bw^{\sf H}\bw}}}\right], \\
    % \label{eq:D_KKT}
    % &\bar{\bD}(\bw) = \lambda_{{\sf RIS},{\sf den}}(\bw) \times
    % \nonumber\\ 
    % &\left[\frac{1}{R_{\Sigma}\ln{2}}\sum_{k = 1}^{K} \left(\frac{\bD_k}{{\bw}^{\sf H}\bD_k {\bw}}\right) + \mu \frac{\sum_{i=1}^{LM} \frac{1}{\bw^{\sf H}\bw}\left(\bX_i - \frac{\bw^{\sf H}\bX_i \bw}{\bw^{\sf H}\bw}\right)e^{\frac{\alpha_1 \bw^{\sf H}\bX_i \bw}{\bw^{\sf H}\bw}}}{\sum_{i=1}^{LM} e^{\frac{\alpha_1 \bw^{\sf H}\bX_i \bw}{\bw^{\sf H}\bw}}}\right].
\end{align}
\end{lemma}
\begin{proof}
    Please refer to Appendix~\ref{pf:NEPv_RIS}. 
\end{proof}
As discussed in Section~\ref{subsec:precoding}, this stationary condition can be interpreted as the generalized eigenvalue problem.
Thus, the problem is transformed for finding its principal eigenvector of \eqref{eq:general_eigen_2}.
In this regard, such eigenvector can also be found by using the GPI method to solve  \eqref{eq:general_eigen_2}.

We describe our regularized GPI-based RIS phase-shift optimization method in Algorithm~\ref{alg:algorithm_2}.
The algorithm first initializes $\bar{\boldsymbol{\phi}}$.
To optimize $\bar{\boldsymbol{\phi}}$ with the GPI method, the algorithm builds $\bar{\bC} ({\bw}^{(t-1)})$ and $\bar{\bD}({\bw}^{(t-1)})$ for given $\bF$,  $\mu$, $R_{\Sigma}$, and  $\tau$.
From the stationary condition in \eqref{eq:general_eigen_2}, the algorithm updates ${\bw}^{(t)} $ by calculating ${\bw}^{(t)} = \bar{\bD}^{-1} ({\bw}^{(t-1)}) \bar{\bC} ({\bw}^{(t-1)}){\bw}^{(t-1)}$ and normalizing ${\bw}^{(t)} = {\bw}^{(t)}/\left\| {\bw}^{(t)}\right\|$.
We repeat these steps until either ${\bw}$ converges to a tolerance level or  $t=t_{2, \rm max}$.
% After the GPI algorithm of Algorithm~\ref{alg:algorithm_2}, we project the optimized vector $\bw^{(t)}$ into the 
Lastly, due to the unit-modulus constraint, the optimized $\bw^{\star}$ is projected onto the feasible solution set by computing $\bar{\boldsymbol{\phi}} = e^{j {\rm arg}(\bw^{\star})}$.

% \begin{remark}
%     [{\color{black}Selection of $\mu$}]\normalfont
%     Finding the optimal value of $\mu$ in the regularized GPI problem is crucial for balancing a trade-off between maximizing the sum SE and minimizing the deviation of the values of RIS elements to satisfy the unit-modulus constraint.
%     However, this balancing behavior complicates our method because the optimization problem may feature multiple local minima.
%     To address this, it is necessary to use an effective search method that involves defining a range of candidate value for $\mu$, and then to select $\mu$ that yields the best performance. 
%     For this purpose, we adopt the line search method for identifying the optimal $\mu$.
%     In addition, we note that the optimized $\bw^{\star}$ in Algorithm~\ref{alg:algorithm_2} may not fit the unit-modulus constraint before projecting the feasible solution set with ${\rm arg}(\cdot)$.
%     We further discuss this effect in Section~\ref{sec:Numerical} to verify the feasibility of our proposed algorithm.
%     \label{remark:1}
% \end{remark}

%%%%%%%%%%%%%%%%%%%%%%%%%%%%%%%%%%%%%%%%%%%%%%%%%%%%%%%%%%%%
\begin{algorithm}[t]
    \caption{{Regularized GPI-Based RIS Phase-Shift Optimization Algorithm}}
    \label{alg:algorithm_2} 
    {\bf{initialize}}: ${\bar{\boldsymbol{\phi}}}^{(0)} \leftarrow \boldsymbol{\Phi}^{(0)}$.
    \\
    Set $\bw^{(0)} = \frac{1}{\sqrt{LM}}{\bar{\boldsymbol{\phi}}}^{(0)}$ and $t= 1$.
    \\
    \While {$\left\|{\bw}^{(t)} - {\bw}^{(t-1)} \right\| > \varepsilon_2$ {\rm or} $t \leq t_{2, \max} $}{
    Build $\bar{\bC} ({\bw}^{(t-1)}\!)$ and $\bar{\bD} ({\bw}^{(t-1)}\!)$  with given $\bF$ and $\mu$.
    \\
    Compute ${\bw}^{(t)} \!=\! \bar{\bD}^{-1} ({\bw}^{(t-1)}) \bar{\bC} ({\bw}^{(t-1)}) {\bw}^{(t-1)}$. 
    \\
    Normalize ${\bw}^{(t)} = {\bw}^{(t)}/\left\| {\bw}^{(t)}\right\|$.
    \\
    $t \leftarrow t+1$.}
    {\color{black}
    $\bw^{\star} \leftarrow \bw^{(t)}$.}
    \\
    ${\bar{\boldsymbol{\phi}}}^{\star}  \leftarrow e^{j {{\rm arg}(\bw^{\star})}}$.
    \\
    % \Return{\ }{${{\boldsymbol{\Phi}}}^\star = \left[{\rm diag}({\boldsymbol{\phi}}_1), {\rm diag}({\boldsymbol{\phi}}_2), \cdots, {\rm diag}({\boldsymbol{\phi}}_L)\right]$}.
    \Return{\ }{${{\boldsymbol{\Phi}}}^\star \leftarrow \bar{\boldsymbol{\phi}}^{\star}$}.
\end{algorithm}
%%%%%%%%%%%%%%%%%%%%%%%%%%%%%%%%%%%%%%%%%%%%%%%%%%%%%%%%%%%%

% ----------------------------------------------------------------
\subsection{Joint Optimization Algorithm} \label{sec:proposed_AA}
% ----------------------------------------------------------------
In this subsection, we propose the alternating algorithm as described in Algorithm~\ref{alg:algorithm_3} for joint optimization of the precoder and RIS phase shifts by putting together the results in Section~\ref{subsec:precoding} and \ref{subsec:RIS}.
% Our aim is to maximize the sum ergodic SE by jointly optimizing the precoder and RIS phase shifts.
% To this end, we use the GPI approach for each subproblem and then combine them in an alternating manner.
% To maximize the sum SE by jointly optimizing the precoder and RIS phase shifts, we use the GPI approach for each subproblem and then combine them in an alternating manner.
At the beginning of Algorithm~\ref{alg:algorithm_3}, the algorithm initializes the precoder and the RIS phase shifts as $\bF^{(0)}$ and $\boldsymbol{\Phi}^{(0)}$ while setting the outer iteration count $i=1$.
Then, the algorithm computes $R_{\Sigma}$ for given $\bF^{(0)}$ to normalize the penalty term in \eqref{eq:AO_RIS}.
% {\color{black}Recall that we set $\tau= (LM)^{-1}$ for normalization of the RIS regularization term.}

For finding an optimal value of $\mu$, we adopt a line search method.
With the line search, the algorithm identifies the optimal value of $\mu$ within the range of $\mu^{(i)} \in [\mu_{\rm min}, \mu_{\rm max}]$ with $T_{\mu}$ linearly spaced points and the increasing step $\Delta_{\mu}$, i.e., $\mu_{\rm max} = \mu_{\rm min} + \Delta_{\mu}(T_{\mu} - 1)$, which is updated in the outer loop.
In the inner loop of the algorithm, the precoder is optimized by Algorithm~\ref{alg:algorithm_1} for given $\boldsymbol{\Phi}$.
% Subsequently, next step is to optimize the RIS phase shifts $\boldsymbol{\Phi}$ with the derived $\bF$. 
Subsequently, using Algorithm~\ref{alg:algorithm_2}, the RIS phase shifts matrix is optimized with updated $\bF$.
We repeat these steps until either it converges or  $t = t_{3, \rm  max}$ in the alternating manner.
For the convergence level, we compare the current solution with the previous solution as  $\frac{|f(\bF^{(t)}, \boldsymbol{\Phi}^{(t)}) - f(\bF^{(t-1)}, \boldsymbol{\Phi}^{(t-1)})|}{f(\bF^{(t-1)}, \boldsymbol{\Phi}^{(t-1)})}$ where  $f(\bF, {\boldsymbol{\Phi}})$ is the objective function of the original problem in \eqref{eq:main_problem}.
% with respect to $\bF$ and $\boldsymbol{\Phi}$.
We use this convergence level with a tolerance threshold $\varepsilon_3>0$.
% Consequently, we decide $\bF$ and $\boldsymbol{\Phi}$ in this alternating manner.
To identify the optimal value of $\mu$, the algorithm computes $f\left(\Tilde{\bF}^{(i)}, \Tilde{\boldsymbol{\Phi}}^{(i)}\right)$ at $\mu^{(i)}$. 
Consequently, the algorithm decides ${\boldsymbol{\Phi}}^{\star}$ and $\bF^{\star}$ that maximizes $f\left(\Tilde{\bF}^{(i)}, \Tilde{\boldsymbol{\Phi}}^{(i)}\right)$ for $i = 1, \cdots, N_{\mu}$.
% We further analyze the convergence behavior of the proposed method in Section~\ref{sec:Numerical}.
% Mu search
% Finding the optimal value of $\mu$ in the regularized GPI problem is crucial for balancing a trade-off between maximizing the sum SE and minimizing the deviation of the values of RIS elements to satisfy the unit-modulus constraint.
% However, this balancing behavior complicates our method because the optimization problem may feature multiple local minima.

% XXX{\color{magenta}
% % To identify the optimal balance of our regularized problem in \eqref{eq:AO_RIS}, it is necessary to use an effective search method that involves defining a range of candidate value for $\mu$, and then to select $\mu$ that yields the best performance. 
% To identify the optimal balance of our regularized problem in \eqref{eq:AO_RIS}, we need an effective search method to select $\mu$ that yields the best performance.
% For this purpose, we adopt the line search method for finding the optimal $\mu$.
% Furthermore, we note that the optimized $\bw^{\star}$ in Algorithm~\ref{alg:algorithm_2} may not fit the unit-modulus constraint before projecting the feasible solution set with ${\rm arg}(\cdot)$.
% We further discuss this effect in Section~\ref{sec:Numerical} to verify the feasibility of our proposed algorithm.}

%%%%%%%%%%%%%%%%%%%%%%%%%%%%%%%%%%%%%%%%%%%%%%%%%%%%%%%%%%%%
%%%%%%%%%%%%%%%%%%%%%%%%%%%%%%%%%%%%%%%%%%%%%%%%%%%%%%%%%%%%
\begin{algorithm}[t]
\caption{GPI-Based Precoding and RIS Phase-Shift Optimization Algorithm (GPI-PRIS)} 
\label{alg:algorithm_3} 
%%%%%%%%%%%%%%%%%%%%%%%%%%%%%%%%%%%%%%%%%%%%%%%%%%%%%%%%%%%%
{\bf{initialize}}: $\bF^{(0)}, {\boldsymbol{\Phi}}^{(0)}$, and $\mu^{(0)}$
\\
Set iteration count for outer loop $i= 1$.
\\
Compute  $\tau= (LM)^{-1}$  and $R_{\Sigma}$ with  ($\bF^{(0)}$,  ${{\boldsymbol{\Phi}}}^{(0)}$). 
\\
\While{$i \leq {\color{black}T_{\mu}}$}{
Set $\mu^{(i)} \leftarrow \mu^{(i-1)} + \Delta_{\mu} \in [\mu_{\rm min}, \mu_{\rm max}]$.
\\
Set iteration count for inner loop $t=1$.
\\
\While {$\frac{|f(\bF^{(t)}\!,\! \boldsymbol{\Phi}^{(t)}) - f(\bF^{(t-1)}\!,\! \boldsymbol{\Phi}^{(t-1)}\!)|}{f(\bF^{(t-1)},\boldsymbol{\Phi}^{(t-1)}\!)} \!>\!\varepsilon_3$ {\rm or}  $t \!\leq\! t_{3, \rm max}$}
{
${\bF}^{(t)} \leftarrow \text{Algorithm~\ref{alg:algorithm_1}}\left(\bF^{(t-1)};{\boldsymbol{\Phi}}^{(t-1)}\right)$.
\\
${\boldsymbol{\Phi}}^{(t)} \leftarrow \text{Algorithm~\ref{alg:algorithm_2}}\left({\boldsymbol{\Phi}}^{(\!t-1\!)}; \bF^{(t)}, \mu^{(i)}\right).$
\\
$t \leftarrow t + 1$.
}
$\Tilde{\bF}^{(i)} \leftarrow \bF^{(t)}$ and $\Tilde{\boldsymbol{\Phi}}^{(i)} \leftarrow \boldsymbol{\Phi}^{(t)}$.
\\
Compute $f\left(\Tilde{\bF}^{(i)}, \Tilde{\boldsymbol{\Phi}}^{(i)}\right)$.
\\
$i \leftarrow i + 1$.
}
Select $i^{\star}=\argmax_{i =1,\cdots, {\color{black}T_{\mu}}} {f\left(\Tilde{\bF}^{(i)}, \Tilde{\boldsymbol{\Phi}}^{(i)}\right)}$.
\\
\Return{\ }$\bF^\star \leftarrow \Tilde{\bF}^{(i^{\star})}$ and ${{\boldsymbol{\Phi}}}^\star \leftarrow \Tilde{\boldsymbol{\Phi}}^{(i^{\star})}$.
\end{algorithm}
%%%%%%%%%%%%%%%%%%%%%%%%%%%%%%%%%%%%%%%%%%%%%%%%%%%%%%%%%%%%
%%%%%%%%%%%%%%%%%%%%%%%%%%%%%%%%%%%%%%%%%%%%%%%%%%%%%%%%%%%%

% ----------------------------------------------------------------
\subsection{Complexity Analysis} \label{sec:complexity}
% ----------------------------------------------------------------
% Algorithm 1
Now, we analyze the complexity of the proposed algorithms.
The complexity of Algorithm~\ref{alg:algorithm_1} is dominated by the inversion in $\bar{\bB}^{-1} (\bar {\bf{f}})$.
Since $\bar{\bB} (\bar {\bf{f}})$ is a block-diagonal and symmetric matrix, we need  $\CMcal{O}(KN^3)$ instead of $\CMcal{O}(K^3 N^3)$ to obtain the inverse of $K$ sub-matrices in $\bar{\bB}(\bar {\bf{f}})$.
Hence, the total complexity of Algorithm~\ref{alg:algorithm_1} is $\CMcal{O}(T_1 KN^3)$ where $T_1$ is the number of its iterations.
We note that this is substantially lower compared to the existing precoding schemes.
{\color{black}For instance, weighted mean square error (WMMSE) \cite{shi2011iteratively} requires $\CMcal{O}(KN^3)$ complexity, which is the same as that of GPI-based precoding. 
However, GPI-based precoding accounts for the effects of channel errors by incorporating error covariance.
When considering channel uncertainty, WMMSE with sample average approximation  \cite{joudeh2016sum} requires $\CMcal{O}((KN)^{3.5})$ complexity based on a quadratically constrained quadratic programming problem.}
% In addition, unlike the existing schemes, the proposed method is not needed to use CVX.

%%%%%%%%%%%%%%%%%%%%%%%%%%%%%%%%%%%%%%%%%%%%%%%%%%%%%%%%%%%%%%%%
\begin{figure}[!t]    
    {\centerline{\resizebox{0.99\columnwidth}{!}{\includegraphics{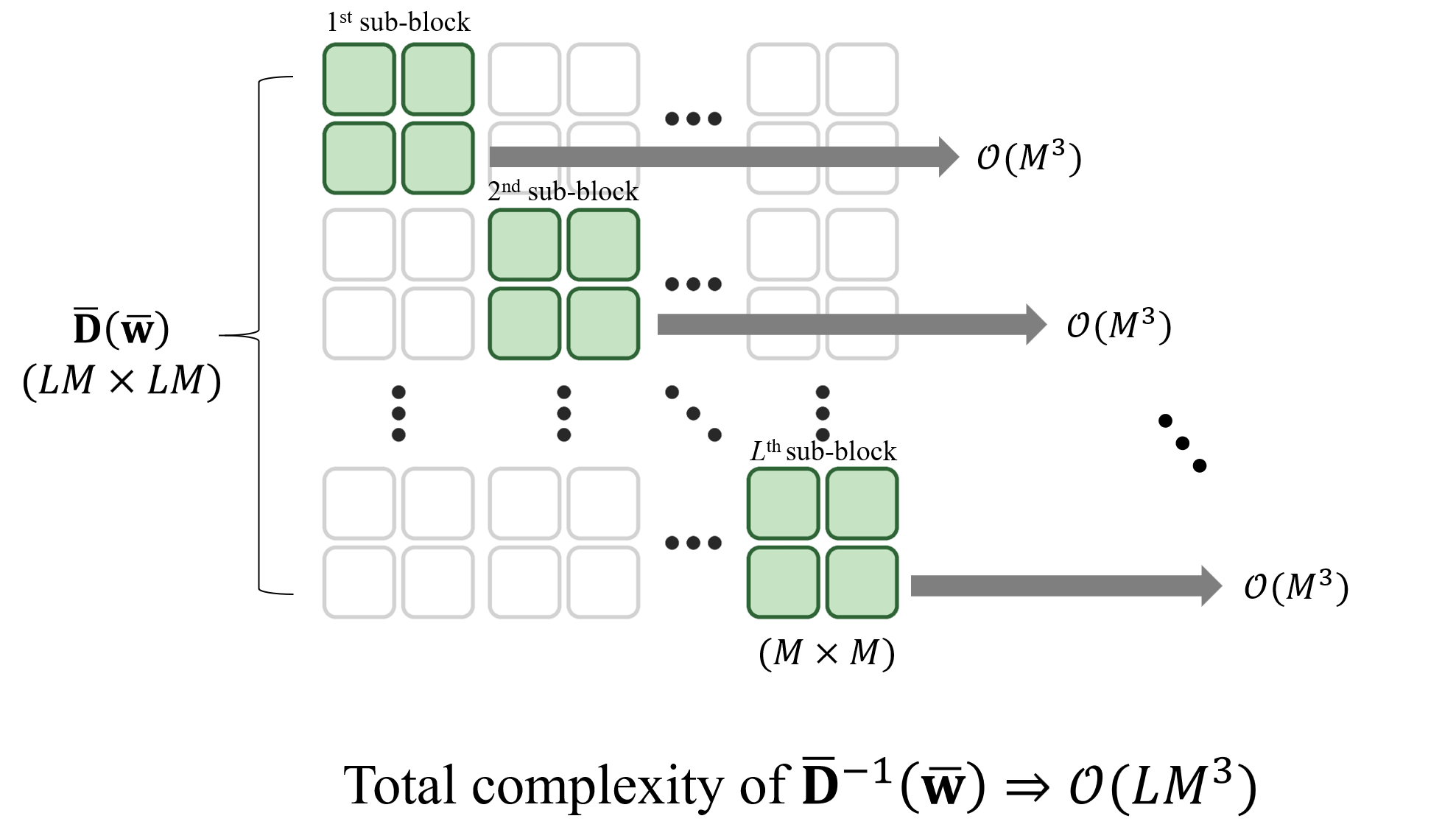}}}
    \caption{\color{black}An illustration of the computationally efficient matrix inversion of $\bar{\bD}(\bar{\bw})$ in Algorithm 2 via block-wise decomposition.}
    \label{fig:Alg2_complexity}}
\end{figure}
%%%%%%%%%%%%%%%%%%%%%%%%%%%%%%%%%%%%%%%%%%%%%%%%%%%%%%%%%%%%%%%%
% Algorithm 2
Similarly, Algorithm~\ref{alg:algorithm_2} complexity is dominated by the inversion in $\bar{\bD}^{-1} (\bw)$ for the regularized GPI method.
Since $\bar{\bD}(\bw)$ is also a block-diagonal and symmetric matrix, we only need $\CMcal{O}(LM^3)$ to obtain the inverse of each sub-matrix in $\bar{\bD}(\bw)$ {\color{black}as depicted in Fig.~\ref{fig:Alg2_complexity}.}
% Based on a line search manner, we need $\CMcal{O}(I_{\mu}) = \CMcal{O}(N_{\mu})$, where $N_{\mu}$ is the linearly spaced points within the interval of $\mu$.
% Based on line search, we need $\CMcal{O}(T_{\mu}) = \CMcal{O}(N_{\mu})$.
% If a bisection method is used, we can reduce this complexity to $\CMcal{O}(\log T_{\mu})$.
Accordingly, the total complexity of Algorithm~\ref{alg:algorithm_2} is $\CMcal{O}(T_2 LM^3)$ where $T_2$ is the number of its iterations.
The state-of-the-art RIS phase-shift optimization scheme such as the MM algorithm or the MO algorithm have the complexity order of $\CMcal{O}(M^3 + T_{i} M^2)$ where $T_i$ denotes the number of iterations required for the MM algorithm or the MO algorithm with $i\in \{\rm{MM, MO}\}$ in single-RIS-aided systems \cite{pan2020multicell}.
Extension to multi-RIS systems, they further need the complexity order of $\CMcal{O}((LM)^3 + T_{i}(LM)^2)$ \cite{yoon2023joint}.

% Algorithm 3
Finally, let us denote the number of iterations set for the line search of $\mu$ as $T_{\mu}$.
Then, the total complexity of Algorithm~\ref{alg:algorithm_3} (GPI-PRIS) is $\CMcal{O}(T_{\mu}(T_1 KN^3 + T_2 LM^3))$.
Considering $M > N$, the main bottleneck arises from Algorithm~\ref{alg:algorithm_2}, and the complexity becomes  $\CMcal{O}(T_{\mu}T_2 LM^3)$.
\begin{remark}
    [Complexity Comparison and Scalability]\normalfont
    % Our algorithm has a comparable complexity order to representative RIS optimization schemes in terms of $M$.
    % However, owing to the special structure of $\bar{\bD}(\bw)$, a block-diagonal and symmetric matrix, our algorithm exhibits linear scaling with respect to $L$.
    % This linear growth makes the algorithm particularly advantageous for scalable beamforming in multi-RIS-aided systems when efficient handling of a large number of RISs is essential.
    % {\color{black}In addition, since it was observed that increasing the number of RISs with fixed total RIS elements $M_{\rm tot}$ (for example, $M= M_{\rm tot}/L$) is advantageous up to a certain optimal $L$ \cite{al2024performance}, our algorithm provides a particularly more significant RIS scalability for the fixed $M_{\rm tot}$ case whose complexity scales as $\sim M_{\rm tot}^3/L^2$, allowing for flexible deployment of multiple RISs to maximize the sum SE.
    % This is also numerically verified in Fig.~\ref{fig:Scalability} in Section~\ref{subsec:evalution}.}

    Our algorithm has a comparable complexity order to representative RIS optimization schemes in terms of $M$.
    However, owing to the special structure of $\bar{\bD}(\bw)$, a block-diagonal and symmetric matrix, our algorithm exhibits linear scaling with respect to $L$.
    This linear growth makes the algorithm particularly advantageous for scalable beamforming in multi-RIS-aided systems when efficient handling of a large number of RISs is essential.
    {\color{black}Particularly, it was observed that increasing the number of RISs with fixed total RIS elements $M_{\rm tot} = L \times M$ is advantageous up to a certain optimal $L$ \cite{al2024performance}.
    For the case of a single RIS $(L=1)$ with a large number of RIS elements, the complexity of our method scales as similarly to conventional approaches, requiring a matrix inversion of size $M_{\rm tot} \times M_{\rm tot}$, which results in $\CMcal{O}(M_{\rm tot}^3)$ complexity.
    In contrast, for multi-RIS-aided systems $(L>1)$, the block-diagonal structure allows us to decompose the problem into $L$ independent sub-inversions, each corresponding to an RIS with $M$ elements.
    Considering $M= M_{\rm tot}/L$, the complexity of our method scales as $\CMcal{O}(LM^3) = \CMcal{O}\left(L \cdot \left(\frac{M_{\rm tot}}{L}\right)^3 \right) = \CMcal{O}\left(\frac{M_{\rm tot}^3}{L^2}\right)$, allowing for flexible deployment of multiple RISs to maximize the sum SE.}
    This is also numerically verified in Fig.~\ref{fig:Scalability} in Section~\ref{subsec:evalution}.

    \label{rm:complexity}
\end{remark}

{\color{black}
% ===============================================================
\subsection{Reduced-Complexity Approach} \label{subsec:reduction}
% ===============================================================
% ------------------------------------------------------------------------
% To reduce the computational burden, we can explore an iterative Gauss-Seidel (GS) approach {\color{red}[REF]} for approximating the required matrix inversions. 
% This approach reduces the inversion complexity of $\bar{\bD}(\bw)$ from $\CMcal{O}(M^3)$ to $\CMcal{O}(T_{GS} M^2)$ where $T_{GS}$ is the number of GS iteration, is typically much smaller than $M$, thereby offering a significant advantage for large-scale RIS deployments.
% ------------------------------------------------------------------------
{\color{black}
While our proposed algorithm exhibits linear scalability in $L$, the cubic dependence on $M$ incurs high computational cost.
To mitigate this issue, we introduce a reduced-complexity GPI approach based on the framework in \cite{oh2025generalized}.
As a first step, we approximate $\bC_{k,\ell}$ by diagonal matrix: $\bC_{k,\ell} \approx a_{k,\ell}\bI_{NM}$. 
Here, $a_{k,\ell}$ is the large-scale fading of user $k$ associated with RIS $\ell$.
Consequently, from \eqref{eq:error_cov}, we have $\bR^{\sf e}_{k,\ell} \approx \left(a_{k,\ell} \bI_{NM} + \frac{T_{\sf UL}\rho_{\sf UL}}{\gamma_{k,\ell}\sigma^2}\bI_{NM} \right)^{-1} = b_{k,\ell}\bI_{NM}$.}
This approximation implies that $\boldsymbol{\Xi}_{k,\ell}$ and $\boldsymbol{\Theta}_{k,\ell}$ can also be approximated as diagonal matrices.
Subsequently, we reorganize the following expressions:
\begin{align}
    &\bar{\bD}(\bw) = \sum_{k=1}^{K}\frac{\bD_k}{\bw^{\sf H}\bD_k\bw} + \bX(\bw),
    \\
    &\bX(\bw) = \frac{\mu}{\tau} \frac{\sum_{i=1}^{LM} {\bX}_i e^{\alpha_1 {\bw}^{\sf H}{\bX}_{i}{\bw}}}{\sum_{i=1}^{LM} e^{\alpha_1 {\bw}^{\sf H}{\bX}_{i}{\bw}}}.
\end{align}
Since $\bar{\bD}(\bw) = {\rm blkdiag}\left(\bar{\bD}_{1}(\bw), \cdots, \bar{\bD}_{L}(\bw)\right)$ is block-diagonal, its $\ell$th sub-block is defined as 
\begin{align}
    \nonumber
    \bar{\bD}_{\ell}(\bw) &\!=\! 
    % \sum_{k=1}^{K}\frac{1}{\bw^{\sf H}\bD_k\bw} \left(\bar{\Upsilon}_{k,\ell} + \boldsymbol{\Theta}_{k,\ell}\right) + \bX_{\ell}(\bw)
    \sum_{k=1}^{K}\frac{1}{\bw^{\sf H}\bD_k\bw} \left({\hat{\bH}_{k,\ell}^{\sf r H}} \sum_{i\neq k}^K \bff_i \bff_i^{\sf H} \hat{\bH}^{\sf r}_{k,\ell} + \boldsymbol{\Theta}_{k,\ell}\right) + \bX(\bw)
    \\
    &\!=\! \sum_{k=1}^{K}\sum_{i\neq k}^K c_k(\bw)\bv_{i,k,\ell}\bv_{i,k,\ell}^{\sf H} + \sum_{k=1}^{K}\boldsymbol{\Theta}_{k,\ell} + \bX(\bw),
\end{align}
where $\bv_{i,k,\ell} = {\hat{\bH}_{k,\ell}^{\sf r H}}{\bff}_i \in \bbC^{M}$ and $c_k(\bw) = \frac{1}{\bw^{\sf H}\bD_k \bw}$.
Approximating each $\boldsymbol{\Theta}_{k,\ell}$ by a diagonal matrix as above, we can approximate $\bar{\bD}_{\ell}(\bw)$ as
\begin{align}
    \label{eq:D_sub_tilde}
    \bar{\bD}_{\ell}(\bw) &\approx \sum_{k=1}^{K}\sum_{i\neq k}^K c_k(\bw)\bv_{i,k,\ell}\bv_{i,k,\ell}^{\sf H} + \bY_{\ell}(\bw) = \tilde{\bD}_{\ell}(\bw),
\end{align}
where $\bY_{\ell}(\bw) = \sum_{k=1}^{K}\boldsymbol{\Theta}_{k,\ell} + \bX(\bw)$ is a diagonal matrix.
The inverse of the $\ell$th sub-block is given by
\begin{align}
    \tilde{\bD}^{-1}_{\ell}(\bw) = \left[\sum_{k=1}^{K}\sum_{i\neq k}^K c_k(\bw)\bv_{i,k,\ell}\bv_{i,k,\ell}^{\sf H} + \bY_{\ell}(\bw)\right]^{-1}.
\end{align}
By applying the Sherman-Morrison-Woodbury (SMW) formula: $\left(\bA + \bu \bv^{\sf H}\right)^{-1} = \bA^{-1} - \frac{\bA^{-1}\bu\bv^{\sf H}\bA^{-1}}{1 + \bv^{\sf H}\bA^{-1}\bu}$,  we can easily compute the inverse of \eqref{eq:D_sub_tilde} in a recursive manner.
Specifically, the inverse of $\bar{\bD}_{\ell}(\bw)$ is obtained by recursively computing the previously obtained inverse via the SMW formula as
\begin{align}
    \label{eq:D_sub_tilde_inv}
    &\tilde{\bD}_{\ell, (p)}^{-1}(\bw) =\left[\tilde{\bD}_{\ell, (p-1)}(\bw) + c_{(p)}(\bw)\bv_{\ell,(p)}\bv_{\ell,(p)}^{\sf H}\right]^{-1} 
    \\
    \nonumber
    &=\! \tilde{\bD}^{-1}_{\ell, (p-1)}(\bw) \!-\! \frac{c_{(p)}(\bw)\tilde{\bD}_{\ell, (p-1)}^{-1}(\bw)\bv_{\ell,(p)}\bv_{\ell,(p)}^{\sf H}\tilde{\bD}_{\ell, (p-1)}^{-1}(\bw)}{1 + c_{(p)}(\bw)\bv^{\sf H}_{\ell,(p)}\tilde{\bD}_{\ell, (p-1)}^{-1}(\bw)\bv_{\ell,(p)}},
\end{align}
where $\tilde{\bD}^{-1}_{\ell, (p)}(\bw)$  denotes the inverse obtained after the $p$th rank-one update in \eqref{eq:D_sub_tilde}, $\bv_{\ell,(p)}$  is the 
$p$-th vector selected from the set $\{\bv_{i,k,\ell},\forall i,k \}$, and $c_{(p)}$ is the corresponding scalar associated with $\bv_{\ell,(p)}$.
By using this recursively rank-one update strategy, the complexity of obtaining \eqref{eq:D_sub_tilde} is $\CMcal{O}(K(K-1)M^2)$ in a divide-and-conquer manner.
Specifically, we accumulate $K(K-1)$ rank-one updates $\bv_{i,k,\ell} \bv_{i,k,\ell}^{\sf H}$, where each of which takes $\CMcal{O}(M^2)$ complexity.
Since there are $L$ sub-blocks, the overall complexity becomes $\CMcal{O}(LK(K-1)M^2)$.
For $M \gg K$, this typically much lower than $\CMcal{O}(LM^3)$ complexity of directly inverting the block-diagonal matrix $\bar{\bD}(\bw)$.
% Consequently, replacing the direct matrix inversion of $\bar{\bD}(\bw)$ with the proposed recursive approach yields a computational advantage.
}

%%%%%%%%%%%%%%%%%%%%%%%%%%%%%%%%%%%%%%%%%%%%%%%%%%%%%%%%%%%%%%%%%%
\section{Numerical Results} \label{sec:Numerical}
%%%%%%%%%%%%%%%%%%%%%%%%%%%%%%%%%%%%%%%%%%%%%%%%%%%%%%%%%%%%%%%%%%
In this section, we present numerical comparisons between
our proposed algorithms and the following baselines:
\begin{itemize}
    \item \textbf{GPI-PRIS}: Our proposed algorithm  (Algorithm~\ref{alg:algorithm_3}).

    \item \textbf{R-GPI-PRIS}: Our proposed GPI-PRIS with the reduced-complexity approach as described in Section~\ref{subsec:reduction}.
        
    % \item \textbf{WMMSE-SAA-MO}: This robust approach leverages the WMMSE-SAA precoder~\cite{joudeh2016sum} with manifold-based RIS phase shifts optimization~\cite{pan2022overview}, which employs the SAA method to effectively handle imperfect CSIT by estimating and averaging multiple channel realizations.

    \item \textbf{WMMSE-MO}: The WMMSE precoder with manifold-based RIS phase shifts optimization \cite{guo2020weighted}.
    
    % \item \textbf{WMMSE-PI}: The WMMSE precoder \cite{shi2011iteratively} with power iteration (PI)-based RIS phase shifts optimization \cite{jin2023low}.

    \item \textbf{WMMSE-PINet}: The WMMSE precoder with PINet-based RIS phase shifts optimization, which is a model-driven deep learning approach \cite{jin2023low}.

    \item {\color{black}\textbf{FP-SE}: A fractional programming-based sum SE maximization method by setting $\xi = 0$ in Algorithm~3 of \cite{huang2019reconfigurable}.}

    \item \textbf{RZF-MO}: A regularized zero-forcing (RZF) precoder with manifold-based RIS phase shifts optimization \cite{guo2020weighted}.

    \item \textbf{Alg~\ref{alg:algorithm_1}-MO}: Our proposed Algorithm~\ref{alg:algorithm_1} precoder with manifold-based RIS phase shifts optimization \cite{guo2020weighted}.

    \item \textbf{Alg~\ref{alg:algorithm_1}-Random}: With Algorithm~\ref{alg:algorithm_1}, phase shifts of each RIS are randomly selected.
\end{itemize}
% Algorithm setting
We use RZF as an initial precoder of Algorithm~\ref{alg:algorithm_1}.
Using line search with $T_{\mu}=15$ linearly spaced points, we identify the optimal value of $\mu$ in range $\mu \in [0, 100]$.
% We set the maximum iteration counts to be $t_{1, {\rm max}} = t_{2, {\rm max}}= t_{3, {\rm max}} = 30$ and tolerance levels as $\varepsilon_1 = \varepsilon_3 = 10^{-2}$, $\varepsilon_2 = 10^{-3}$.
% For R-GPI-PRIS, we obtain $a_{k,\ell}$ by computing $a_{k,\ell} = \frac{{\rm tr}(\bC_{k,\ell})}{NM}$.

% ----------------------------------------------------------------
\subsection{Simulation Environments}
% ----------------------------------------------------------------
% Channel model
In the considered system, {\color{black}we assume that both antennas at the BS and elements of the RISs are arranged in uniform planar arrays (UPAs).}
% The small wavelength of mmWave signals restricts their ability to diffract around obstacles. 
% Consequently, mmWave channels typically display a sparse multipath structure and are often modeled using a geometric channel.
{\color{black}
In the presence of the same spatial paths $L_{\sf BR}$ for all BS-RIS links, the channel matrix between the BS and RIS $\ell$ is given by $\forall \ell \in \cL$
\begin{align}
    \label{eq:H_1_LoS}
    % \bH_{1,\ell} = \frac{1}{\sqrt{L_{\sf BR}}} \sum_{i=1}^{L_{\sf BR}} \sqrt{\gamma_{1,\ell}} \ba_{\sf B}(\vartheta^{\sf B}_{i, \ell}) \ba_{\sf R}^{\sf H} (\varphi^{a}_{i, \ell}, \varphi^{e}_{i, \ell}),
    \bH_{1,\ell} = \frac{1}{\sqrt{L_{\sf BR}}} \sum_{i=1}^{L_{\sf BR}} \sqrt{\gamma_{1,\ell}} \ba_{\sf B}(\vartheta^{a}_{i, \ell}, \vartheta^{e}_{i, \ell}) \ba_{\sf R}^{\sf H} (\varphi^{a}_{i, \ell}, \varphi^{e}_{i, \ell})\;,
\end{align}
where $\gamma_{1, \ell}$ is the path gain of $\ell$th BS-RIS link, $\vartheta^{a}_{i, \ell}\;(\varphi^{a}_{i, \ell})$ and $\vartheta^{e}_{i, \ell}\;(\varphi^{e}_{i, \ell})$ denote azimuth and elevation of AoDs (AoAs), and $\ba(\cdot)$ denotes a steering vector.
Considering $N_y$ and $N_z$ as the horizontal and vertical indices of the BS antenna, i.e., $N = N_y \times N_z$, the array response vectors at the BS and RIS $\ell$ are defined as 
% \begin{align}
%     \ba_{\sf B}(\vartheta^{a}_{i, \ell}, \vartheta^{e}_{i, \ell}) &= \ba_{\sf B}^{y}(\vartheta^{a}_{i, \ell}, \vartheta^{e}_{i, \ell}) \otimes \ba_{\sf B}^{z}(\vartheta^{e}_{i, \ell}),
%     \\
%     \ba_{\sf R}(\varphi^{a}_{i, \ell}, \varphi^{e}_{i, \ell}) &= \ba_{\sf R}^{y}(\varphi^{a}_{i, \ell}, \varphi^{e}_{i, \ell}) \otimes \ba_{\sf R}^{z}(\varphi^{e}_{i, \ell}).
% \end{align}
$\ba_{\sf B}(\vartheta^{a}_{i, \ell}, \vartheta^{e}_{i, \ell}) = \ba_{\sf B}^{y}(\vartheta^{a}_{i, \ell}, \vartheta^{e}_{i, \ell}) \otimes \ba_{\sf B}^{z}(\vartheta^{e}_{i, \ell}), \ba_{\sf R}(\varphi^{a}_{i, \ell}, \varphi^{e}_{i, \ell}) = \ba_{\sf R}^{y}(\varphi^{a}_{i, \ell}, \varphi^{e}_{i, \ell}) \otimes \ba_{\sf R}^{z}(\varphi^{e}_{i, \ell}).$
The array response vectors at the BS are defined as 
\begin{align}
    \label{eq:a_B_y}
    &\ba_{\sf B}^{y}(\vartheta^{a}_{i, \ell}, \vartheta^{e}_{i, \ell}) \!
    \\ \nonumber
    &=\! \left[\!1, e^{j 2\pi \frac{\Delta^y_{\sf B}}{\lambda_c}\!\sin(\vartheta^{a}_{i, \ell})\!\sin(\vartheta^{e}_{i, \ell})},\ldots,e^{j 2\pi (N_y \!-\! 1)  \frac{\Delta^y_{\sf B}}{\lambda_c}\!\sin(\vartheta^{a}_{i, \ell})\!\sin(\vartheta^{e}_{i, \ell})}\!\right]^{\!\sf T}\!\!\!,
    \\
    \label{eq:a_B_z}
    &\ba_{\sf B}^{z}(\vartheta^{e}_{i, \ell}) \!=\! \left[1, e^{j 2\pi \frac{\Delta^z_{\sf B}}{\lambda_c}\!\cos(\vartheta^{e}_{i, \ell})}, \!\ldots\! , e^{j 2\pi (N_z - 1)  \frac{\Delta^z_{\sf B}}{\lambda_c}\!\cos(\vartheta^{e}_{i, \ell})} \right]^{\sf T}\!\!\!,
\end{align}
where  $\Delta^y_{\sf B}$ and $\Delta^z_{\sf B}$ are the distance between adjacent BS antennas along two axes.
Similarly, with the horizontal and vertical dimensions of the RIS $M_y$ and $M_z$ (i.e., $M = M_y \times M_z$), its array response vectors, $\ba_{\sf R}^{y}(\varphi^{a}_{i, \ell}, \varphi^{e}_{i, \ell})$ and $\ba_{\sf R}^{z}(\varphi^{e}_{i, \ell})$, are defined analogously to \eqref{eq:a_B_y} and \eqref{eq:a_B_z}, respectively.}
% \begin{align}
%     \ba_{\sf R}^{y}(\varphi^{a}_{i, \ell}, \varphi^{e}_{i, \ell}) \!&=\!\left[\!1, e^{j 2\pi \frac{\Delta^y_{\sf R}}{\lambda_c}\!\sin(\varphi^{a}_{i, \ell})\!\sin(\varphi^{e}_{i, \ell})},\ldots,e^{j 2\pi (M_y \!-\! 1)  \frac{\Delta^y_{\sf R}}{\lambda_c}\!\sin(\varphi^{a}_{i, \ell})\!\sin(\varphi^{e}_{i, \ell})}\!\right]^{\!\sf T}\!\!\!,
%     \\
%     \ba_{\sf R}^{z}(\varphi^{e}_{i, \ell}) \!&=\! \left[1, e^{j 2\pi \frac{\Delta^z_{\sf R}}{\lambda_c}\!\cos(\varphi^{e}_{i, \ell})}, \!\ldots\! , e^{j 2\pi (M_z - 1)  \frac{\Delta^z_{\sf R}}{\lambda_c}\!\cos(\varphi^{e}_{i, \ell})} \right]^{\sf T}\!\!\!,
% \end{align}
% where  $\Delta^y_{\sf R}$ and $\Delta^z_{\sf R}$ are the distance between adjacent RIS reflecting elements along two axes.
We assume that the same total number of spatial paths between RIS and user is $L_{\sf RU}$ for all RIS-user links.
The RIS-user channel is
\begin{align}
    \label{eq:h_2_LoS}
    \bh_{2,k,\ell} = \frac{1}{\sqrt{L_{\sf RU}}} \sum_{i=1}^{L_{\sf RU}} \sqrt{{\gamma}_{2,k,\ell}} \alpha^{\sf sc}_{k,i,\ell} \ba_{\sf R}(\vartheta^{{\sf R},a}_{k,i,\ell}, \vartheta^{{\sf R},e}_{k,i,\ell}),
\end{align}
where $\gamma_{2, k,\ell}$ denotes the path gain of the RIS-user link, $\alpha^{\sf sc}_{k,i,\ell}$ denotes small-scale fading satisfying $\alpha^{\sf sc}_{k,i,\ell} \sim \cC \cN (0,1)$, and $(\vartheta^{{\sf R},a}_{k,i,\ell}, \vartheta^{{\sf R},e}_{k,i,\ell})$ denote and azimuth and elevation AoDs between RIS $\ell$ and user $k$.

% FIGURE %%%%%%%%%%%%%%%%%%%%%%%%%%%%%
\begin{figure}[t]    
    {\centerline{\resizebox{0.85\columnwidth}{!}{\includegraphics{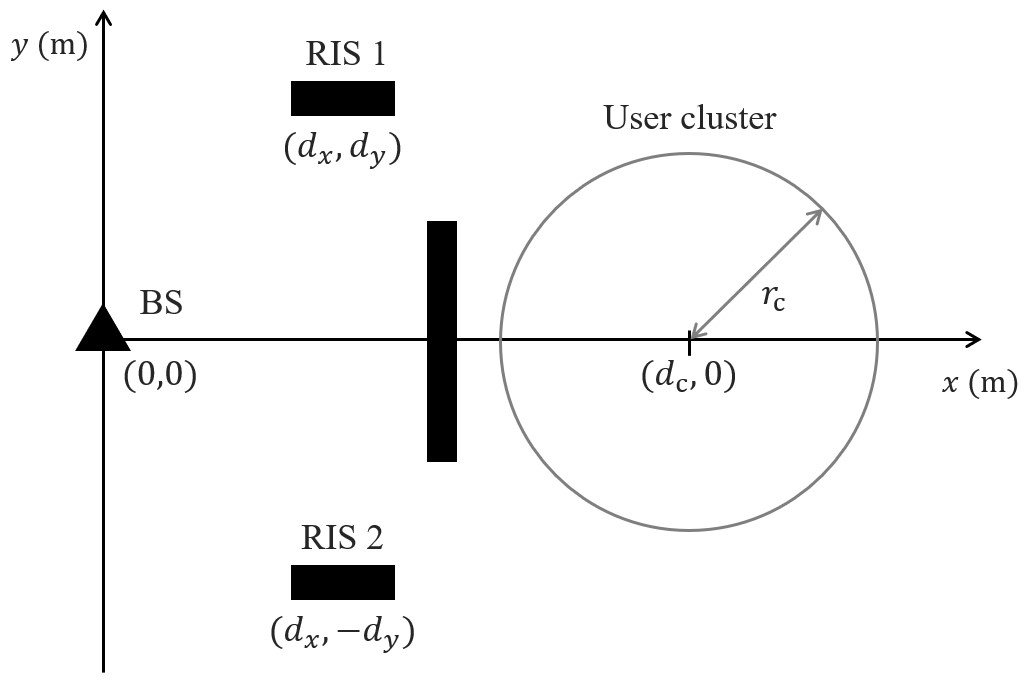}}}}
    \vspace{-0.1cm}
    \caption{A schematic of the multiuser network with two RISs}
    \label{fig:network}
\end{figure}
%%%%%%%%%%%%%%%%%%%%%%%%%%%%%%%%%%%%%%%
%%%%%%%%%%%%%%%%%%%%%%%%%%%%%%%%%%%%%%%
\begin{table}[t!]
\centering
\caption{\color{black}Comparison of Sum SE under different LogSumExp parameter}
\label{tab:LogSumExp_alpha}
\renewcommand{\arraystretch}{1.5} % 행 간격 조절
{\color{black}
\begin{tabular}{c c c c c c}
\toprule
\multirow{2}{*}{{$P$ (dBm)}} & \multicolumn{5}{c}{{Sum SE (bits/sec/Hz)}} \\
\cmidrule(lr){2-6}  
  & ${\alpha=0.001}$ & ${\alpha=0.01}$ & ${\alpha=0.1}$ & ${\alpha=1}$ & ${\alpha=10}$ \\ 
\midrule
0   & 10.84 & 10.98 & 11.09 & 11.29 & 11.25 \\
20  & NaN   & 25.98 & 26.39 & 25.68 & 23.86 \\
40  & NaN   & 31.01 & 30.97 & 29.59 & 26.74 \\
\bottomrule
\end{tabular}
}
\end{table}
%%%%%%%%%%%%%%%%%%%%%%%%%%%%%%%%%%%%%%%
%%%%%%%%%%%%%%%%%%%%%%%%%%%%%%%%%%%%%%%
\begin{table}[t]
    \caption{{\color{black}Simulation Parameters}}
    \vspace{-0.1cm}
    \label{tab:sim_params}
    \centering
    \renewcommand{\arraystretch}{1.2}  % 1이 기본값, 1.3은 30% 증가
    {\color{black}
    \begin{tabular}{l|c}
    \hlineB{2}
    \textbf{\color{black}Parameter} & \textbf{\color{black}Value} \\ \hline
    Carrier Frequency & 28 GHz \\
    Bandwidth ($W$) & 1 GHz \\
    Noise Figure ($n_f$) & 5 dB \\
    Noise Power & $-174 + 10\log_{10}W + n_f$ dBm \\
    Array Structure & uniform planar array (UPA) \\
    Antenna Spacing ($\Delta/\lambda_c$) & 0.5 \\ \hline
    BS Location & $(0, 0)$ m \\
    User Cluster Center ($d_c$) & $(60, 0)$ m \\
    User Cluster Radius ($r_c$) & 20 m \\
    RIS Locations & $(20, 20)$ m, $(20, -20)$ m \\ \hline
    Path Loss Parameters & $\alpha_{\sf PL}=61.4, \beta_{\sf PL}=2, \sigma_s=5.8$ dB \\
    Number of Paths ($L_{\sf BR}, L_{\sf RU}$) & 2 \\
    Uplink Training Length ($T_{\sf UL}$) & $M$ \\
    Uplink Training Power ($\rho_{\sf UL}$) & 0 dBm \\ \hline
    Maximum Iterations & $t_{1,\rm max} = t_{2,\rm max} = t_{3,\rm max} = 30$ \\
    Tolerance Levels & $\epsilon_1=\epsilon_3=10^{-2}, \epsilon_2=10^{-3}$ \\
    LogSumExp Parameters & $\alpha_1 = \alpha_2 = 0.1$ \\
    $\mu$ Search ($T_\mu$) & 15 (Range: $[0, 100]$) \\ 
    Initial Precoder & RZF \\
    Initial RIS Phase Shift & randomly generated $\phi \in [-\pi, \pi)$ \\ \hlineB{2}
    \end{tabular}
    }
\end{table}
%%%%%%%%%%%%%%%%%%%%%%%%%%%%%%%%%%%%%%%
% mmWave Pathloss
Considering that the system operates at a 28 GHz carrier frequency, we adopt the mmWave pathloss model \cite{akdeniz2014millimeter}.
Thus, the pathloss in dB is given by 
\begin{align}
    \label{eq:PL}
    {\sf PL}(d)= \alpha_{\sf PL} + \beta_{\sf PL} 10 \log_{10} d + \chi \ [{\rm dB}],
\end{align}
where $d$ denotes the link distance in meter, $\chi \sim \cN (0,\sigma_s^2)$ is the log-normal shadowing.
This pathloss model is applied to the path gain terms in \eqref{eq:H_1_LoS} and \eqref{eq:h_2_LoS}.
According to \cite{akdeniz2014millimeter}, the experimental data
for 28 GHz channels and 1 GHz bandwidth indicates that the parameters in \eqref{eq:PL} are set to be $\alpha_{\sf PL} = 61.4\;{\rm dB}, \beta_{\sf PL}=2\;{\rm dB}, \sigma_s^2 = 5.8\;{\rm dB}$. 
% and $\alpha_{\sf PL} = 72.0, \beta_{\sf PL}=2.92, \sigma_s^2 = 8.7$ dB for the NLoS paths.
Additionally, the noise power is given by
\begin{align}
    \label{eq:P_noise}
    P_{\sf noise}  = -174 + 10 \log_{10} W + n_f \ [{\rm dBm}],
\end{align}
where $W$ and $n_f$ are the channel bandwidth and noise figure at the BS.
{\color{black}Assuming the normalized noise variance, i.e., $\sigma^2 =1$, the large-scale path gain is computed also from normalizing the pathloss as}
\begin{align}
    \gamma = - ({\sf PL}(d) + P_{\sf noise})\  [{\rm dB}] .
\end{align}
In \eqref{eq:P_noise}, we consider $W\!=\!1$ GHz and $n_f \!=\! 5$ dB. 
In the considered channel model, we set $\Delta_{\sf B}/\lambda_c \!=\!  \Delta_y/\lambda_c \!=\! \Delta_z/\lambda_c \!=\! 0.5$, ${\color{black} L_{\sf BR} = L_{\sf RU} = 2}$, and randomly generate the signal azimuth and elevation angles of the UPA as $(\varphi^{a}_{i,\ell}, \vartheta^{{\sf R}, a}_{k,i,\ell}) \in [-\pi, \pi)$ and $(\varphi^{e}_{i,\ell},\vartheta^{{\sf R}, e}_{k,i,\ell}) \in [-\pi/2, \pi/2]$, and also the AoD of the ULA as $\vartheta^{\sf B}_{\ell} \in [0, \pi],\; \forall \ell, k, i$.
For the channel estimation discussed in Section~\ref{subsec:channel_est}, we set $T_{\sf UL} = M$ and $\rho_{\sf UL} = 0\;{\rm dBm}$.

In the simulations, we consider the systems assisted with $L=2$ RISs unless mentioned otherwise.
To illustrate the location of the entities in the considered
system, we apply a two-dimensional (2D) coordinate system as shown in Fig.~\ref{fig:network}.
The BS and origin of the circle are set as $(0, 0)$ and $(d_{{\sf c}}, 0)$, respectively.
The users are randomly generated within $r_{\sf c} = 20$ m radius of a circle, and the distance between the BS and the origin of the circle is set to be {\color{black}$d_{{\sf c}} = 60$ m}.
The locations of RISs are set as $(d_{x}, d_{y})$ and $(d_{x}, -d_{y})$, where $d_{x} \!=\! d_{y} \!=\! 20$ m.

% {\color{black} 
% To empirically justify the setting of the smoothing parameters $\alpha_1$ and $\alpha_2$ in LogSumExp, we investigate their impact on the performance of GPI-PRIS.
% These parameters govern the trade-off between the tightness of the approximation and the numerical stability of the algorithm. 
% Theoretically, as $\alpha \to 0$, LogSumExp provides a tighter approximation of the original non-smooth max/min functions.
% However, extremely small values can induce numerical instability and overly steep gradients, which may hinder convergence.
% To identify the optimal setting, for $\alpha = \alpha_1 = \alpha_2$, we evaluated the sum SE performance under different $\alpha$ values, as summarized in Table~\ref{tab:LogSumExp_alpha}. 
% The results demonstrate that setting $\alpha=0.1$ yields the highest sum SE, offering a reasonable balance between approximation accuracy and stability. 
% Conversely, reducing $\alpha$ to 0.001 caused computational failure (`NaN' in Table I). 
% This excessively tight approximation mimics the non-differentiable sharp transitions of the original functions, resulting in gradient explosions or indeterminate forms.
% Based on these empirical findings, we adopted $\alpha_1 = \alpha_2 = 0.1$ for all subsequent simulations.
% We summarize the simulation parameters in Table~\ref{tab:sim_params}.}
{\color{black} 
To empirically justify the setting of the smoothing parameters $\alpha_{1}$ and $\alpha_{2}$ in LogSumExp, we investigate their impact on the performance of GPI-PRIS. 
These parameters govern the fundamental trade-off between the tightness of the approximation and the numerical stability of the algorithm.
{\color{black} 
Consistent with the theoretical analysis in Section~\ref{subsec:RIS}, as $\alpha \to 0$, the LogSumExp approximation provides a tighter approximation of the non-smooth max/min functions.
However, this improved accuracy comes at the cost of inducing extreme curvature in the optimization landscape. 
As discussed, such sharp curvature increases the spectral radius of the iteration mapping’s Jacobian matrix, which can violate the contraction mapping condition essential for convergence.
% To identify the optimal setting, we evaluated the sum SE performance under different $\alpha$ values (where $\alpha=\alpha_{1}=\alpha_{2}$), as summarized in Table I. The results demonstrate that setting $\alpha=0.1$ yields the highest sum SE, striking a robust balance between approximation accuracy and stability. Conversely, reducing $\alpha$ to 0.001 led to computational failure (`NaN' in Table I). T
% This empirical failure confirms our theoretical insight that the loss of contraction at extremely small $\alpha$ results in algorithm divergence. Based on these findings, we adopted $\alpha_{1}=\alpha_{2}=0.1$ for all subsequent simulations.
The results demonstrate that setting $\alpha=0.1$ yields the reasonably high sum SE, offering a suitable balance between approximation accuracy and stability. 
Conversely, reducing $\alpha$ to 0.001 caused computational failure (`NaN' in Table~\ref{tab:LogSumExp_alpha}). 
% This excessively tight approximation mimics the non-differentiable sharp transitions of the original functions, resulting in gradient explosions or indeterminate forms.
This empirical failure confirms our theoretical insight that the loss of contraction at extremely small $\alpha$ results in algorithm divergence.}
Based on these empirical findings, we adopted $\alpha_1 = \alpha_2 = 0.1$ for all subsequent simulations.
}

% ----------------------------------------------------------------
\subsection{Performance Evaluation}
\label{subsec:evalution}
% ----------------------------------------------------------------
% FIGURE %%%%%%%%%%%%%%%%%%%%%%%%%%%%%
% Fig: Sum-rate vs. P
\begin{figure}[t]    
    {\centerline{\resizebox{0.95\columnwidth}{!}{\includegraphics{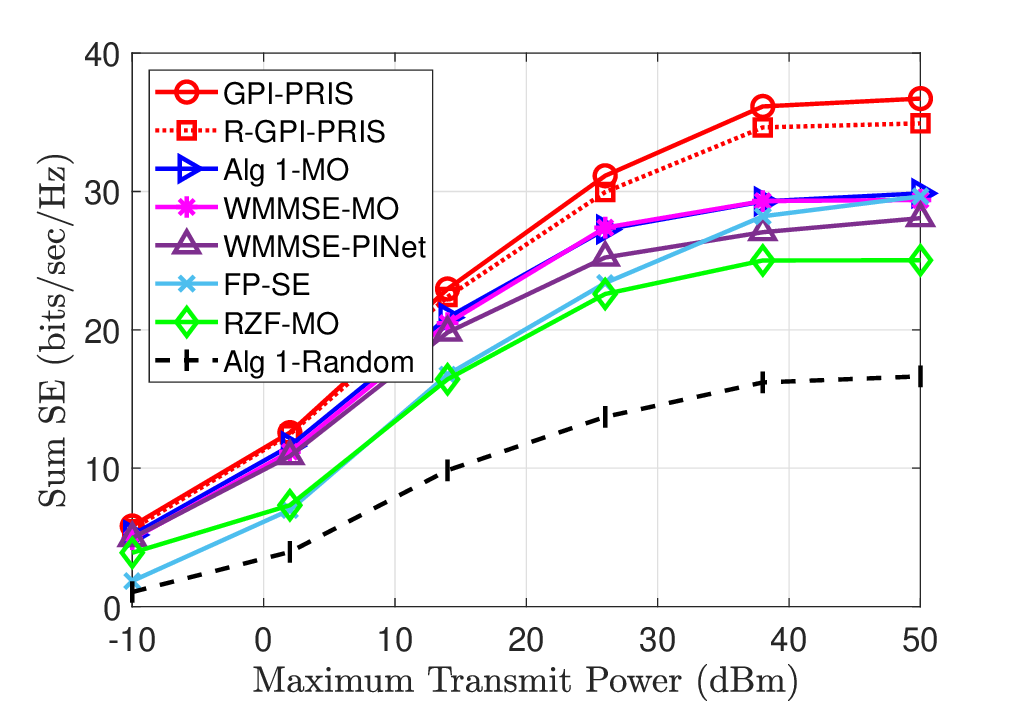}}}}
    \vspace{-0.1cm}
    \caption{The sum SE versus the maximum transmit power $P$ dBm for $N = 16$ BS antennas, $K = 4$ users, and $M = 64$ RIS phase shifts.
    \label{fig:RatevsP}}
\end{figure}
%%%%%%%%%%%%%%%%%%%%%%%%%%%%%%%%%%%%%%%

\textit{1) Maximum Transmit Power vs. Sum SE}:  
We evaluate the performance of the proposed algorithm with respect to the maximum transmit power $P$  for {\color{black}$N\!=\!16$, $K\!=\!4$, $M\!=\!64$ ($M_y \!=\! 8$, $M_z \!=\! 8$)}.
As shown in Fig.~\ref{fig:RatevsP}, our method achieves substantial SE gains  compared to all baselines across all transmit power regimes.
Specifically, the proposed algorithm exhibits a significant performance improvement over the random-phase scheme.
In addition, RZF, due to its linear nature, shows limited performance in maximizing the sum SE.
This observation suggests that SE performance can be effectively enhanced through RIS phase shift optimization.
In terms of RIS phase shift optimization, our method shows significant  performance improvement compared to Alg~\ref{alg:algorithm_1}-MO by finding a superior local optimal solution.
% Nonetheless, since our method considers the lower bound of instantaneous SE with the error covariance matrix, it proves to be comprehensive and robust under the imperfect CSIT scenario.
% In addition, Algorithm~\ref{alg:algorithm_2} excels in optimizing the RIS phase shifts, providing superior local optimal solutions. 
Consequently, our method demonstrates the effective joint optimization for the precoder and  RIS phase shift over different transmit power regime.

% FIGURE %%%%%%%%%%%%%%%%%%%%%%%%%%%%%
% Fig: Sum-rate vs. M
\begin{figure}[t]    
    {\centerline{\resizebox{0.95\columnwidth}{!}{\includegraphics{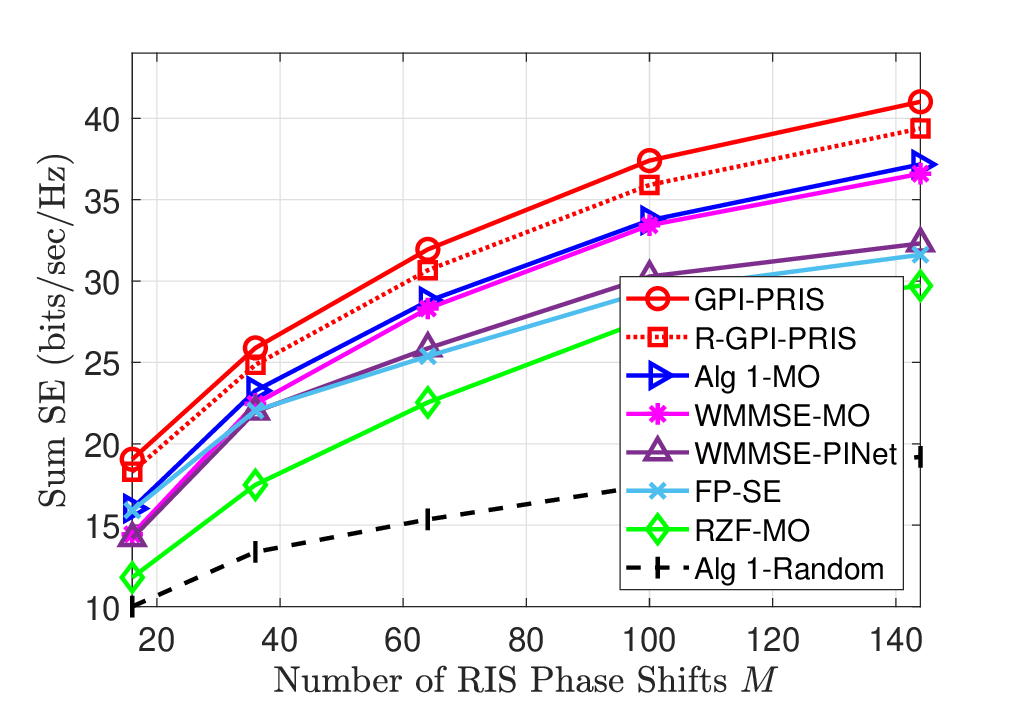}}}
    \vspace{-0.1cm}
    \caption{The sum SE versus the number of RIS phase shifts $M$ for $N = 16$ BS antennas, $K = 4$ users, and $P = 30$ dBm maximum transmit power.}
    \label{fig:RatevsM}}
\end{figure}
%%%%%%%%%%%%%%%%%%%%%%%%%%%%%%%%%%%%%%%
\textit{2) Number of RIS Phase Shifts vs. Sum SE}:
We compare the performance of the proposed algorithm and baselines in terms of the number of RIS phase shifts $M$ with $M_y \!=\! M_z$.
We depict the comparison results in Fig.~\ref{fig:RatevsM} for $N\!=\!16$, $K\!=\!4$, and $P\!=\!30$ dBm.
In Fig.~\ref{fig:RatevsM}, the proposed algorithm achieves the highest performance across all tested numbers of RIS elements.
Fig.~\ref{fig:RatevsM} also shows a significant increase in the SE performance with increasing $M$ for our method.
In contrast, the baseline methods show only marginal improvements with increasing $M$. 
This improvement is attributed to the efficient transmission strategy by properly incorporating the partial CSIT and identifying a superior local optimal point the regularized GPI method.
These results emphasize the suitability of the proposed algorithm  for  high-speed data communications and coverage expansion by employing multiple RISs with a number of elements.

% % FIGURE %%%%%%%%%%%%%%%%%%%%%%%%%%%%%
% % Fig: Sum-rate vs. N
% \begin{figure}[t]    
%     {\centerline{\resizebox{0.85\columnwidth}{!}{\includegraphics{Figures/Fig_NvsSE.eps}}}
%     % \vspace{-0.2cm}
%     \caption{The sum SE versus the number of BS antennas $N$ for $M = 32$ RIS phase shifts, $K = 4$ users, and $P = 20$ dBm maximum transmit power.}
%     \label{fig:RatevsN}} 
% \end{figure}
% %%%%%%%%%%%%%%%%%%%%%%%%%%%%%%%%%%%%%%%
% \textit{3) Number of BS antennas vs. Sum SE}:
% We evaluate the sum SE with respect to the number of BS antennas for $K\!=\!4$, $M\!=\!32\;(M_y\!=\!8,M_z\!=\!4)$, and $P\!=\!20$ dBm.
% In Fig.~\ref{fig:RatevsN}, it is shown that the proposed algorithm also achieves the highest SE performance across all tested numbers of $N$.
% We note that RZF-MO achieves comparable performance to other baselines at $N\!=\!64$ due to sufficient degree-of-freedom (DoF).
% While all baselines show similar performance trends in terms of $N$, a  performance gap exists between our method and the baselines, which becomes more pronounced with higher $N$.
% Similar to the observation from Fig.~\ref{fig:RatevsM}, fully leveraging the partial CSIT and finding a superior local optimal point through the GPI method contribute to this performance improvement.

% FIGURE %%%%%%%%%%%%%%%%%%%%%%%%%%%%%
% Fig: Sum-rate vs. rho
\begin{figure}[t]    
    {\centerline{\resizebox{0.95\columnwidth}{!}{\includegraphics{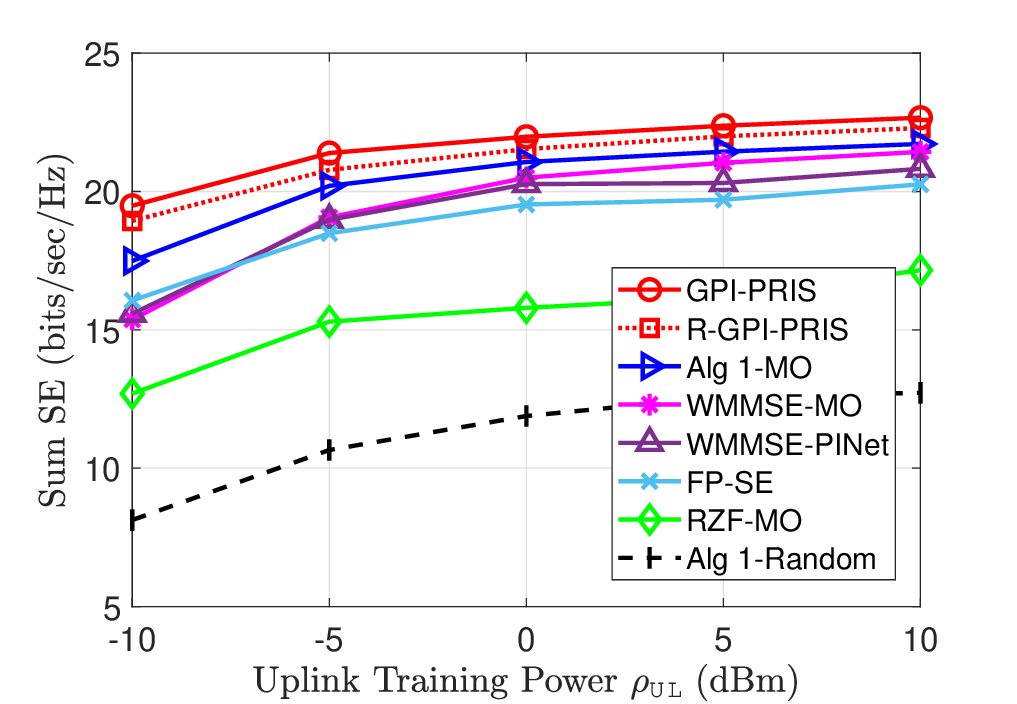}}}
    \vspace{-0.1cm}
    \caption{The sum SE versus the channel estimation parameter $\rho_{\sf UL}$ for $N = 16$ BS antennas, $K = 4$ users, $M = 32$ RIS phase shifts, and $P = 20$ dBm maximum transmit power.}
    \label{fig:CE_para}} 
\end{figure}
%%%%%%%%%%%%%%%%%%%%%%%%%%%%%%%%%%%%%%%

\textit{3) CSIT Accuracy}:  
We investigate the sum SE in relation to the accuracy of CSIT.
In this simulation, we consider $N=16$, $K=4$, $M=32$, and $P=20\;{\rm dBm}$.
Note that the increased value of $\rho_{\sf UL}$ directly correlates with improved channel estimation accuracy. 
As expected, Fig.~\ref{fig:CE_para} shows an increasing SE with $\rho_{\sf UL}$.
Fig.~\ref{fig:CE_para} demonstrates that our algorithm achieves the highest SE performance over different $\rho_{\sf UL}$.
We also observe that there exists performance gap between Alg~\ref{alg:algorithm_1}-MO and WMMSE-MO in the low $\rho_{\sf UL}$ regime from Fig.~\ref{fig:CE_para}.
This performance gap arises from Algorithm~\ref{alg:algorithm_1} that the error covariance is embedded to effectively handle channel estimation errors.
For such a reason, GPI-PRIS has robust performance in the coarse channel estimation environment owing to both Algorithm~\ref{alg:algorithm_1} and Algorithm~\ref{alg:algorithm_2} utilizing the error covariance. 
Thus, leveraging our GPI-based optimization framework and incorporating the error covariance-embedded transmission strategy, our algorithm offers robust joint beamforming solutions that maintain the highest performance.

% %%%%%%%%%%%%%%%%%%%%%%%%%%%%%%%%%%%%
\begin{figure}[!t]
    \centering
    $\begin{array}{c}
    {\resizebox{0.95\columnwidth}{!}{\includegraphics{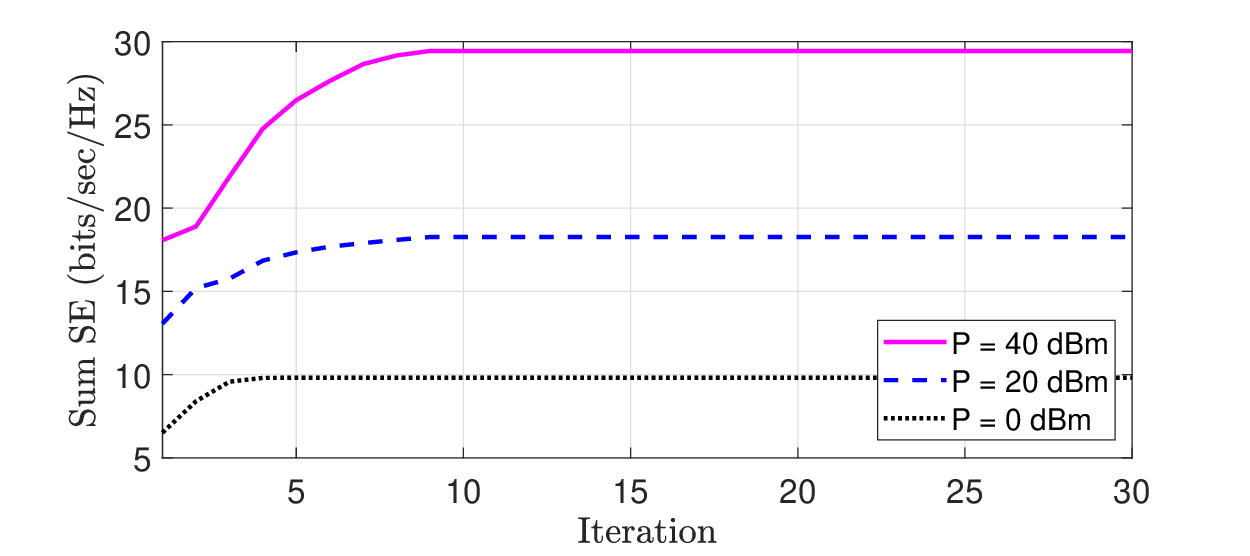}}
    } \\ \mbox{\small (a) $M=32$}
    \end{array}$
    
    $\begin{array}{c}
    {\resizebox{0.95\columnwidth}{!}{\includegraphics{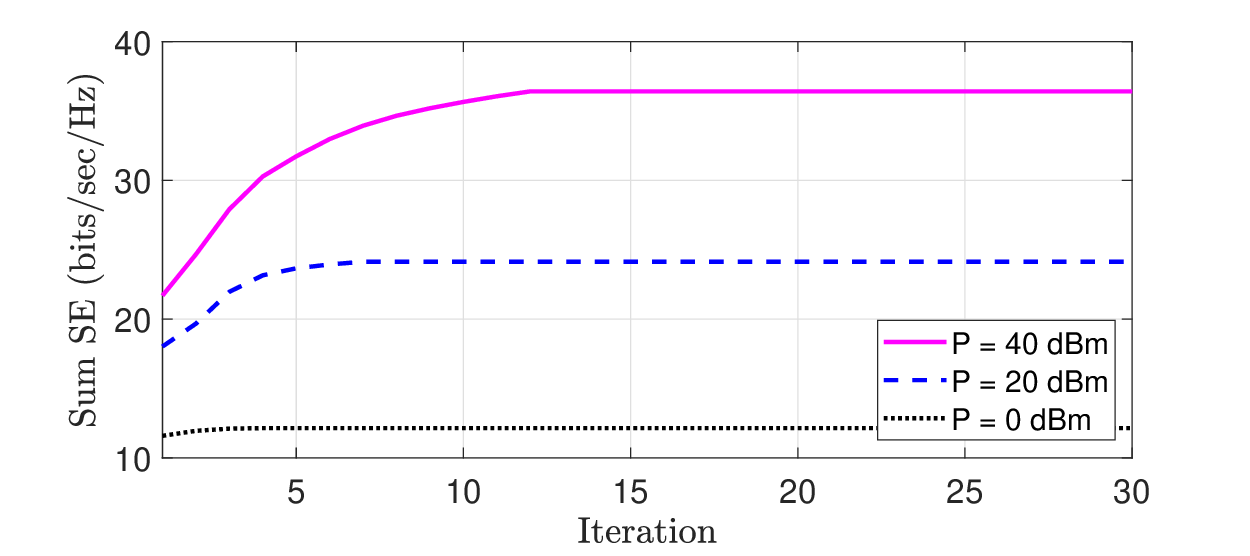}}
    }\\  \mbox{\small (b) $M=64$}
    \end{array}$
    \vspace{-0.1cm}
    \caption{Convergence behavior of GPI-PRIS (Algorithm~\ref{alg:algorithm_3}) for $N=16$ BS antennas, $K=4$ users, $M \in \{32, 64\}$ RIS phase shifts, and $P \in \{0, 20, 40\}$ dBm maximum transmit power.}
    \label{fig:Convergence}
\end{figure}
% %%%%%%%%%%%%%%%%%%%%%%%%%%%%%%%%%%%%
\textit{4) Convergence \& Initialization Behavior}:
To identify the convergence behavior of the proposed algorithm, we evaluate the proposed GPI-PRIS regarding the sum SE versus the number of outer iterations for the various transmit power regime as $P \in \{0, 20, 40\}$ dBm.
We assume $N=16$ and $K=4$ in both $M =32$ and $M =64$ cases.
% In this simulation, we identify the convergence behavior of our proposed method (Algorithm~\ref{alg:algorithm_3}).
% Considering $N=16$ and $K=4$, we evaluate Algorithm~\ref{alg:algorithm_3} regarding the sum SE versus the number of iterations for the various transmit power regime, i.e., $P \in \{0, 20, 40\}$ dBm with $M \in \{32, 64\}$.
For the $M\!=\!32$ ($M_y \!=\! 8$, $M_z \!=\! 4$) case, Fig.~\ref{fig:Convergence}(a) illustrates that the proposed algorithm converges within 11 iterations. 
Additionally, we examine a larger scale RIS configuration with $M\!=\!64$  ($M_y \!=\! 8$, $M_z \!=\! 8$) as depicted in Fig.~\ref{fig:Convergence}(b). 
This scenario demonstrates that the proposed algorithm requires 13 iteration to converge.
% The increase in the number of iterations can be attributed to the expanded search space in the higher transmit power regime and number of RIS elements.
The increase in the iteration counts can be attributed to the expanded search space for the number of RIS elements.
The results underscore that the proposed algorithm achieves fast convergence.
% , enabling the deployment of additional RISs and more reflecting elements effectively.

{\color{black}
\begin{figure}[!t]
    \centering
    $\begin{array}{c}
    {\resizebox{0.95\columnwidth}{!}{\includegraphics{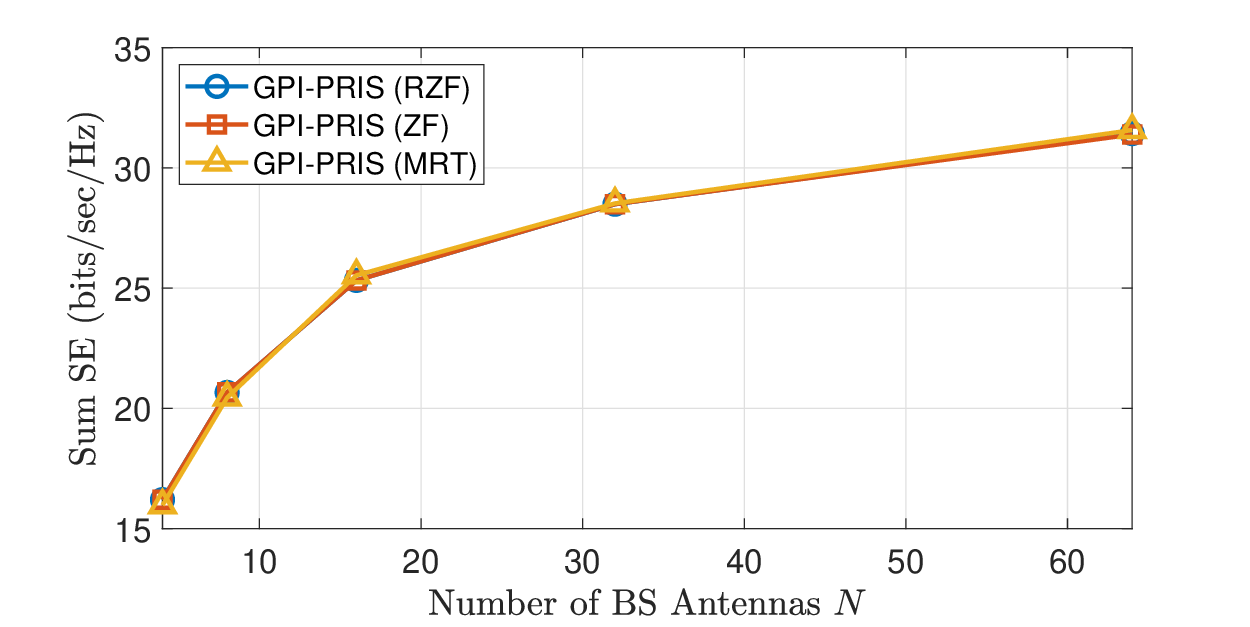}}
    } \\ \mbox{\small (a) {\color{black}$N$ vs. sum SE}}
    \end{array}$
    % \vspace{-0.25cm}
    $\begin{array}{c}
    {\resizebox{0.95\columnwidth}{!}{\includegraphics{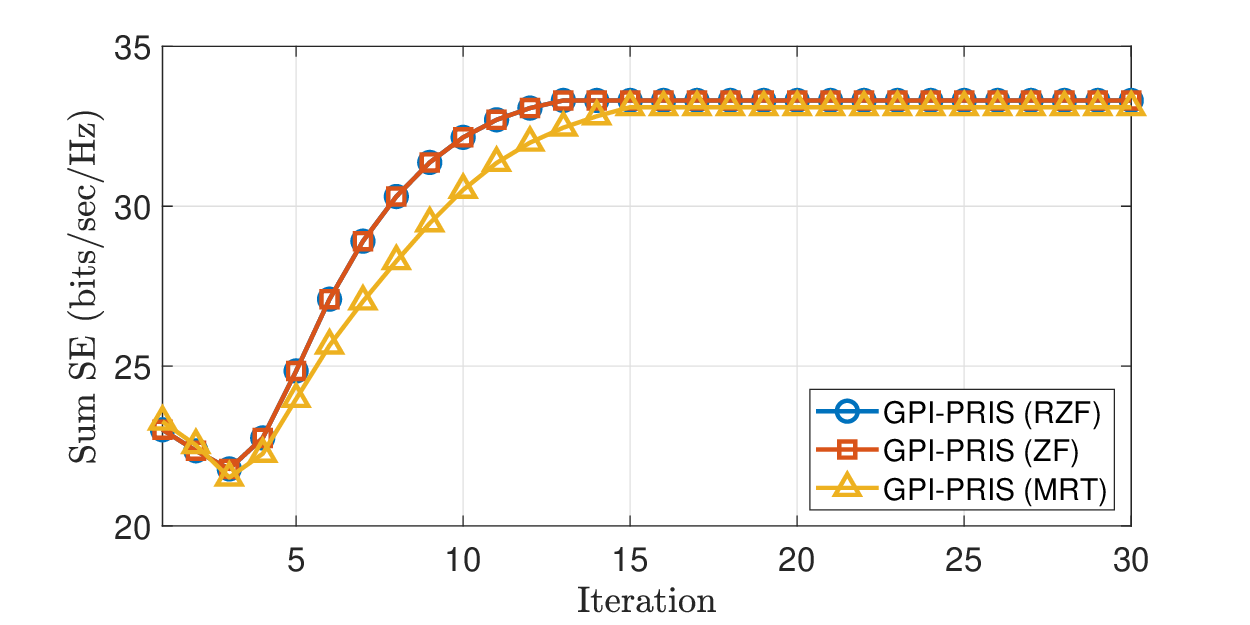}}
    }\\  \mbox{\small (b) {\color{black}Convergence}}
    \end{array}$
    \caption{{\color{black}The comparison of different initial precoders (RZF, ZF, and MRT) for GPI-PRIS.}}
    \label{fig:precoder_init}
\end{figure}
% %%%%%%%%%%%%%%%%%%%%%%%%%%%%%%%%%%%%
% In Fig.~\ref{fig:precoder_init}(a), we consider $N=16, K=4$, and $M=64$.
% As shown in Fig.~\ref{fig:precoder_init}(a), we observe that the choice of the initial precoder does not materially affect SE performance.
% Similarly, Fig.~\ref{fig:precoder_init}(b) examines the effect of varying $N$ while fixing $M=32, K=4$, and $P=40$ dBm.
% Consistent with the power analysis, the SE performance remains virtually unchanged regardless of the initial precoder selection, with all three methods showing comparable performance as $N$ increases.
% These results collectively demonstrate the robustness of the GPI-PRIS algorithm to precoder initialization.
}

{\color{black}
To verify the convergence stability of the proposed algorithm, we evaluated the sum SE performance under different initial precoders: RZF, zero-forcing (ZF), and maximum ratio transmission (MRT). 
Fig.~\ref{fig:precoder_init}(a) examines the effect of varying $N$ while fixing $M=32, K=4$, and $P=40$ dBm.
As shown in Fig.~\ref{fig:precoder_init}(a), the sum SE performance improves as the number of BS antennas $N$ increases, while exhibiting negligible differences among the different initialization schemes.
Furthermore, Fig.~\ref{fig:precoder_init}(b) illustrates the convergence behavior of the proposed algorithm for $N=16, K=4$, and $M=32$.
It is observed that despite different starting points, all initialization schemes converge to the similar steady-state sum SE value within a few iterations. 
This indicates that the proposed GPI-PRIS is robust to initialization and reliably converges to a superior local optimal solution without requiring careful tuning of the starting point.
}

\textit{5) Regularization Behavior}:
We analyze the regularization behavior of our algorithms based on a GPI method for $N\!=\!16$, $K\!=\!4$, and $M\!=\!64$ $(M_y\!=\!8, M_z\!=\!8)$ with the following metric:
\begin{align}
    \label{eq:nmse}
    % {\rm nmse} \left( LM{\rm diag}\left(\bw^{\star} \bw^{\star \sf H}\right) - {\bf 1}\right),
    % {\rm nmse} \left(\sqrt{LM}\|\bw^{\star}\|_1 - {\bf 1}\right),
    {\rm NMSE} \left(\sqrt{LM}|\bw^{\star}| - {\bf 1}\right),
\end{align}
% $\bw^{\star}$ denotes the normalized RIS phase shifts vector optimized by Algorithm~\ref{alg:algorithm_2}, 
where ${\rm NMSE}(\cdot)$ denotes a normalized mean square error (NMSE) function and ${\bf 1}$ denotes a vector whose elements are one with proper dimension.
% To study the impact of the regularization parameter $\mu \in [0, 100]$, we employ the linear search method with $N_{\mu}=15$ linearly spaced points.
Fig.~\ref{fig:NMSE}(a) illustrates that NMSE decreases for all transmit power regimes as $\mu$ increases.
This behavior can be attributed to the increasing dominance of the penalty term in \eqref{eq:AO_RIS} with higher $\mu$, which effectively enforces the unit-modulus constraint on the RIS phase shifts. 
Consequently, as $\mu$ increases, NMSE approaches zero, which indicates that the unit-modulus constraint is nearly satisfied.

%%%%%%%%%%%%%%%%%%%%%%%%%%%%%%%%%%%%%%%
\begin{figure}[!t]
    \centering
    $\begin{array}{c}
    {\resizebox{0.95\columnwidth}{!}{\includegraphics{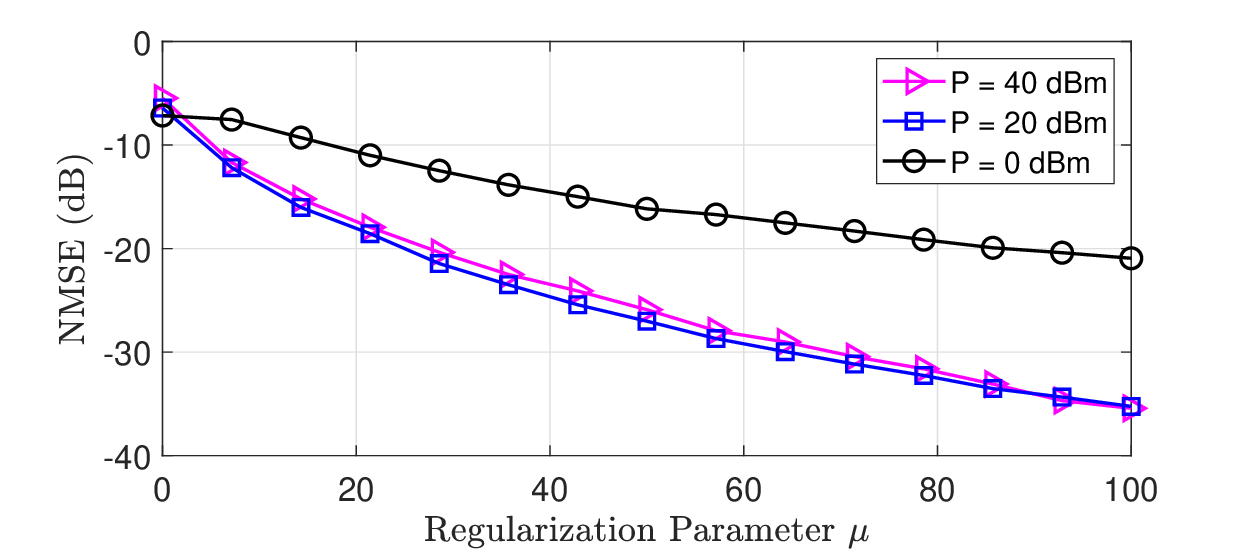}}
    } \\ \mbox{\small (a) NMSE}
    \end{array}$
    
    $\begin{array}{c}
    {\resizebox{0.95\columnwidth}{!}{\includegraphics{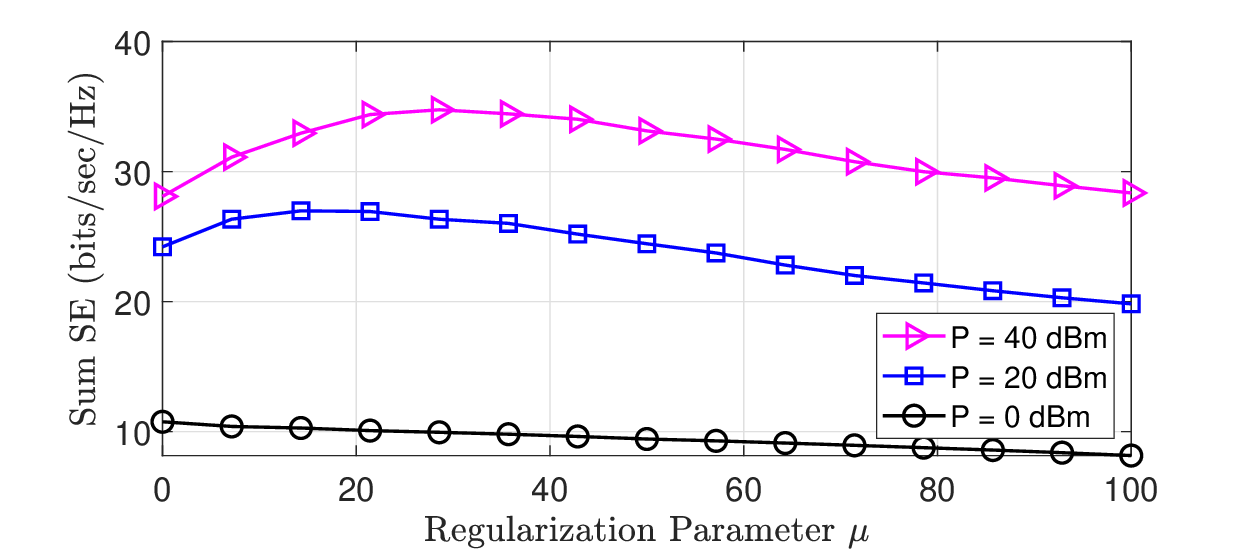}}
    }\\  \mbox{\small (b) Sum SE}
    \end{array}$
    \vspace{-0.1cm}
    \caption{The results in terms of the regularization parameter $\mu$ for $N=16$ BS antennas, $K=4$ users, $M=64$ RIS phase shifts, and $P \in \{0, 20, 40\}$ dBm maximum transmit power.}
    \label{fig:NMSE}
\end{figure}
%%%%%%%%%%%%%%%%%%%%%%%%%%%%%%%%%%%%%%%

We also verify the SE performance with respect to the value of $\mu$ with the same scenario considered above.
In Fig.~\ref{fig:NMSE}(b), it is observed that the sum SE is effectively maximized using the line search method. 
{\color{black}
% This behavior arises because even slight deviations from the unit-modulus constraint can lead to significant degradation in SE performance due to the high sensitivity of the RIS-aided system to phase inaccuracies or amplitude fluctuations after projecting a feasible solution set in the high $P$ regime.
% Let $\bar{\boldsymbol{\phi}}_{\sf d} \!=\! [{\boldsymbol{\phi}}_{{\sf d},1},\!\cdots\!, {\boldsymbol{\phi}}_{{\sf d},L}]$ be a deviation from the optimized RIS elements defined as $\bar{\boldsymbol{\phi}}_{\sf d} \!=\! e^{j{\rm arg}(\bw^{\star})} \!-\! \sqrt{LM}\bw^{\star}$.
Let $\bar{\boldsymbol{\phi}}_{\sf d} \!=\! [{\boldsymbol{\phi}}_{{\sf d},1},\!\cdots\!, {\boldsymbol{\phi}}_{{\sf d},L}]$ be a deviation from the optimized RIS elements prior to projecting onto the feasible solution set, which is defined as $\bar{\boldsymbol{\phi}}_{\sf d} \!=\! e^{j{\rm arg}(\bw^{\star})} \!-\! \sqrt{LM}\bw^{\star}$.
At low $P$, the deviation has marginal impact on SE performance, since the AWGN power is more dominant than the possible errors caused by such deviation.
% This effect can be explained by treating this deviation as another channel error in \eqref{eq:ins_SE}, 
% i.e, $\sum_{\ell =1}^{L}\!\bH^{\sf r}_{k,\ell}(\boldsymbol{\phi}_{\ell}^{\star} \!-\! {\boldsymbol{\phi}}_{{\sf d},\ell})$.
In this regard, the optimization primarily focuses on maximizing the sum SE with little consideration of the unit-modulus constraint.
{\color{black}At high $P$, 
% since the channel inaccuracies become a dominant factor to maximize the sum SE, 
the optimization  prioritizes strict adherence to the constraint, and thus the optimal $\mu$ becomes larger.
We note that depending on the transmit power, there exist optimal $\mu$.
In addition, the sum SE reveals the low sensitivity around the optimal region. 
In this regard, the suitable value of $\mu$ can be pre-defined with respect to the transmit power, which further reduces the complexity of the proposed algorithm.
Overall, we confirm that our approach enables a balanced trade-off in the regularized optimization problem over $P$.}

% %%%%%%%%%%%%%%%%%%%%%%%%%%%%%%%%%%%%
\begin{figure}[!t]
    \centering
    $\begin{array}{c}
    {\resizebox{0.95\columnwidth}{!}{\includegraphics{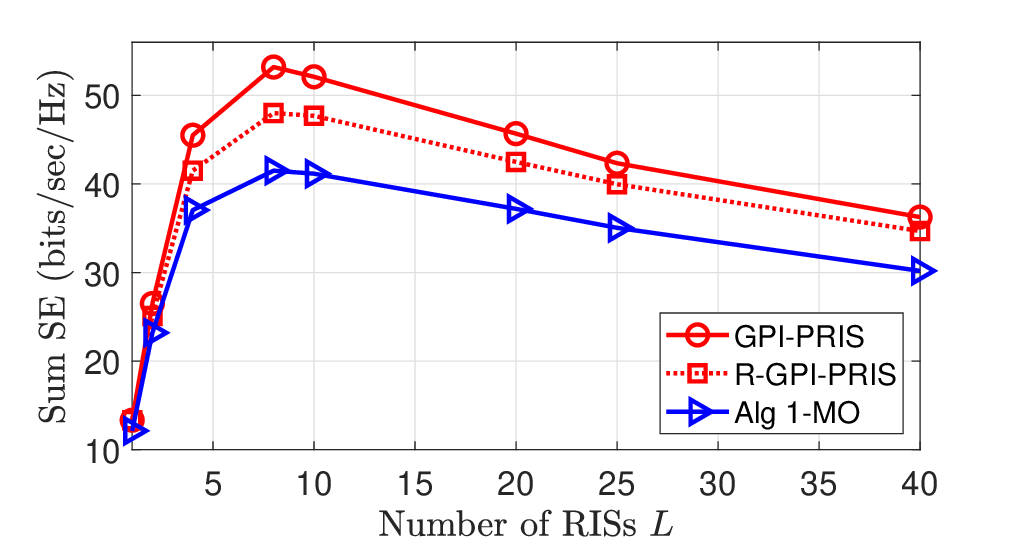}}
    } \\ \mbox{\small (a) Sum SE}
    \end{array}$
    % \vspace{-0.25cm}
    $\begin{array}{c}
    {\resizebox{0.95\columnwidth}{!}{\includegraphics{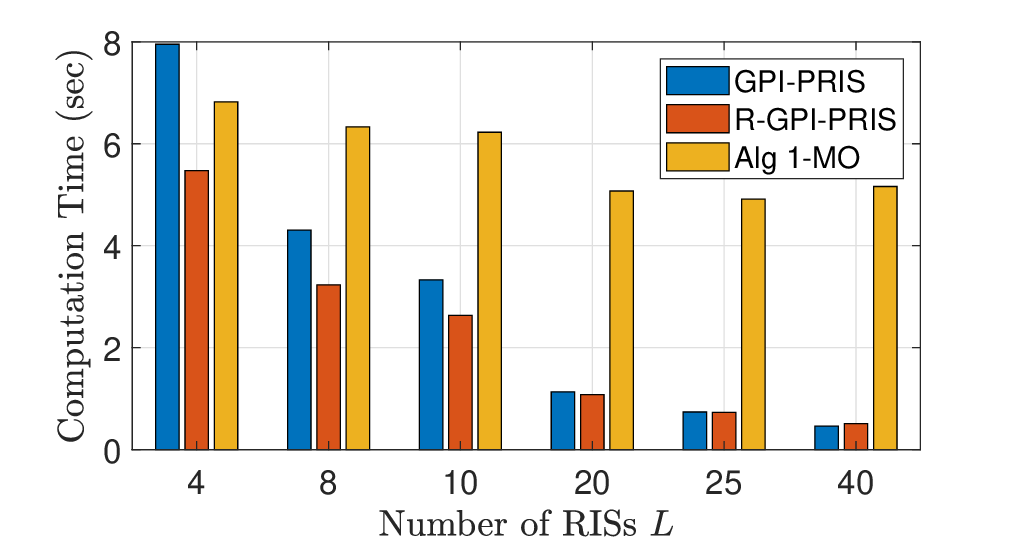}}
    }\\  \mbox{\small {(b) Computation Time}}
    \end{array}$
    % \vspace{-0.1cm}
    \caption{The comparison of the proposed algorithms and the baseline over the number of RISs $L$ with $N=32, K=8$, and ${\rm SNR}= 20 {\rm \;dB}$. The total number of RIS elements is fixed as $M_{\rm tot} = LM=200$.}
    \label{fig:Scalability}
\end{figure}
% %%%%%%%%%%%%%%%%%%%%%%%%%%%%%%%%%%%%

\textit{6) Multi-RIS Scalability}:
We evaluate the sum SE in scalable RIS-aided systems by increasing $L$ and keeping the total number of deployed RIS elements $M_{\rm tot} = L \times M=200$.
The values of $M_y$ and $M_z$ are determined such that $M_{\rm tot}/L = M_y \times M_z$.
{\color{black}We consider $N=32$ and $K=8$.
Here, we use an empirically pre-determined value of $\mu$, and consider only the small-scale fading with both randomly distributed RISs and users while fixing all path-loss terms as 1 for simplicity.
Hence, in this scenario, we use a signal-to-noise ratio (SNR) $P/\sigma^2$ which is set to be $P/\sigma^2 = 20\;{\rm dB}$ SNR.
With fixed $M_{\rm tot}$, Fig.~\ref{fig:Scalability}(a) demonstrates performance improvement as $L$ increases, reaching an optimal configuration at $L=8$ from the considered cases.
We note that GPI-PRIS achieves the highest SE performance across different values of $L$.

% We also assess computation time in MATLAB with the workstation (i9-13900K CPU and 64 GB RAM). 
% We remark that at the best $L=8$, {\color{magenta}where the highest sum SE is observed}, GPI-PRIS achieves about $28\%$ SE gain while requiring only about $68\%$ computation time of Alg \ref{alg:algorithm_1}-MO.
% {\color{magenta}
% Because unlike Alg 1-MO whose complexity is $\CMcal{O}(M_{\rm tot}^2)$, which represents a constant complexity order in $L$, our proposed GPI-PRIS achieves complexity order of $\CMcal{O}(M_{\rm tot}^3/L^2)$ with respect to $L$ as discussed in Remark~\ref{rm:complexity}.
% In Fig.~\ref{fig:Scalability}(b), the proposed algorithms exhibit a decreasing trend as $L$ increases, whereas Alg \ref{alg:algorithm_1}-MO remains nearly constant, which aligns with our discussion. 
% Moreover, Fig.~\ref{fig:Scalability}(b) shows that R‑GPI‑PRIS reduces runtime relative to GPI‑PRIS since R-GPI-PRIS scales as $\CMcal{O}(M_{\rm tot}^2/L)$, which is square complexity in $M_{\rm tot}$.}
% In particular, at $L=4$, R-GPI-PRIS achieves lower computation time than Alg 1-MO, whereas GPI-PRIS still remains more computationally expensive.
% This approach effectively reduces the computation time for large $M$ with marginal performance degradation.
% These results offer a clearer understanding of how performance and complexity are balanced by the proposed algorithms.

We also assess computation time in MATLAB with the workstation (Intel i9-13900K CPU and 64GB RAM). 
We remark that at $L=8$, where the highest sum SE is observed, GPI-PRIS achieves about $28\%$ SE gain while requiring only about $68\%$ computation time of Alg \ref{alg:algorithm_1}-MO.
As discussed in Remark~\ref{rm:complexity}, GPI-PRIS exhibits complexity $\CMcal{O}(M_{\rm tot}^3/L^2)$, which decreases with $L$, whereas Alg \ref{alg:algorithm_1}-MO has constant complexity $\CMcal{O}(M_{\rm tot}^2)$ independent of $L$.
In Fig.~\ref{fig:Scalability}(b), the computation time for GPI-PRIS decreases as $L$ increases, while Alg \ref{alg:algorithm_1}-MO remains nearly constant, which aligns with our discussion.
Furthermore, R-GPI-PRIS, with complexity $\CMcal{O}(M_{\rm tot}^2/L)$, achieves even lower computation time than GPI-PRIS.
Notably, at $L=4$, R-GPI-PRIS requires less computation time than Alg 1-MO, while GPI-PRIS remains more computationally expensive at this point.
This reduced-complexity variant effectively mitigates the computational burden for large $M$ with only marginal performance degradation.
These results offer a clearer understanding of how performance and complexity are balanced by the proposed algorithms.
}

% Although GPI-PRIS's complexity decreases as $L$ increases, the number of RISs that achieves optimal performance is fixed.
% Accordingly, it would be beneficial to investigate the trade-offs through simulations for practical systems, as these results offer a clearer understanding of how performance and complexity are balanced.}
% To numerically verify this, we assess computation time in MATLAB with the workstation equipped with i9-13900K CPU and 64 GB RAM.
% Fig.~\ref{fig:Scalability}(b) shows that GPI-PRIS exhibits a decreasing trend as $L$ increases, whereas Alg~\ref{alg:algorithm_1}-MO remains nearly constant, which aligns with our discussion.
% We remark that at the best $L=8$, GPI-PRIS achieves about $26\%$ SE gain while requiring only about $43\%$ computation time of Alg-\ref{alg:algorithm_1}-MO.

Overall, our algorithm demonstrates both the superior SE performance and multi-RIS scalability compared to the baselines, which  can provide significant benefit of multi-RIS systems for future wireless communications.

%%%%%%%%%%%%%%%%%%%%%%%%%%%%%%%%%%%%%%%%%%%%%%%%%%%%%%%
\section{Conclusion} \label{sec:conclusion}
%%%%%%%%%%%%%%%%%%%%%%%%%%%%%%%%%%%%%%%%%%%%%%%%%%%%%%%
In this paper, we proposed the scalable  and effective beamforming design for multi-RIS-aided systems under imperfect CSIT. 
Aiming to maximize the sum SE, we  reformulated the optimization problem by deriving a lower bound of the instantaneous SE with partial CSIT. 
% However, solving this problem is challenging due to its non-convexity and the unit-modulus constraint imposed by the RISs.
Dividing the problem into two subproblems, we developed the unified GPI-based optimization framework with regularization that identifies a superior local optimal solution.
% for SE maximization problems.
% Dividing the problem into two subproblems, we  derived the stationary conditions and identified the superior local optimal points with respect to both precoder and RIS phase shifts.
In particular, by leveraging the block-diagonality of the GPI matrices, the proposed algorithm achieves multi-RIS scalability with respect to the number of RISs $L$, i.e., $\sim \CMcal{O}(1/L^2)$ for the fixed total RIS elements case.
% when optimizing the RIS phase shifts, we addressed the unit-modulus constraint by proposing a regularized GPI method. 
Through simulations, we showed that our method not only outperforms conventional methods in terms of SE but also demonstrates significant computational scalability with respect to the number of RISs $L$, making it suitable for large-scale multi-RIS deployments.
Therefore, this work contributes both a robust and efficient beamforming design for practical multi-RIS-aided communication systems.
{\color{black}
For future work, it would be interesting to extend the proposed scalable framework to practical scenarios involving discrete RIS phase shifts, particularly 1-bit or 2-bit resolution phase shifts. 
Additionally, investigating the impact of near-field propagation in extremely large-scale RIS systems remains a promising direction for further research.}

%%%%%%%%%%%%%%%%%%% Appendix %%%%%%%%%%%%%%%%%%%%%%%%%%
\appendices

% -----------------------------------------------------
\section{Proof of Lemma~\ref{lem:NEP}} \label{pf:NEPv_precoder}
% -----------------------------------------------------
% From \eqref{eq:lamb_obj}, 
% \begin{align}
%     \label{eq:Lag_1}
%     \cL_1(\bar {\bf f}) = \log_2  \prod_{k=1}^{K} \left(\frac{\bar{\mathbf{f}}^{\sf H}\mathbf{A}_{k}\bar{\mathbf{f}}}{\bar{\mathbf{f}}^{\sf H}\mathbf{B}_{k}\bar{\mathbf{f}}}\right) = \log_2 \lambda_1(\bff).
% \end{align}
According to the stationary condition, the stationary points need to satisfy  $\partial \cL_{\sf BS}(\bar {\bf f})/\partial \bar {\bf f}^{\sf H} = 0$.
Thus, we take the partial derivative of $\cL_{\sf BS}(\bar {\bf f})$ with respect to ${\bf \bar f }$ and set it to be zero. 
By using the derivative of the Rayleigh quotient form as
\begin{align}
    \label{eq:derivative_matrix}
    \frac{\partial \left( \frac{{\bf \bar f }^{\sf H} {\bA}_k \bar {\bf f} }{{\bf \bar f }^{\sf H} {\bB}_k \bar {\bf f}} \right)}{\partial \bar {\bf f}^{\sf H}} = 2 \left( \frac{{\bf \bar f }^{\sf H} {\bA}_k \bar {\bf f} }{{\bf \bar f }^{\sf H} {\bB}_k \bar {\bf f}} \right) \left[  \frac{{\bA}_k \bar {\bf f} }{{\bf \bar f }^{\sf H} {\bA}_k \bar {\bf f}} - \frac{{\bB}_k \bar {\bf f} }{{\bf \bar f }^{\sf H} {\bB}_k \bar {\bf f}} \right],
\end{align}
we can calculate the partial derivative of $\cL_{\sf BS}(\bar {\bf f})$ in \eqref{eq:lamb_obj} as
% with respect to ${\bf \bar f }$ as
\begin{align}
    \label{eq:partial_L1}
    \frac{\partial \cL_{\sf BS} (\bar{\bf f})}{\partial \bar {\bf f}^{\sf H}} = \sum_{k=1}^K \frac{2}{\ln 2} \left( \frac{{\bA}_k \bar {\bf f} }{{\bf \bar f }^{\sf H} {\bA}_k \bar {\bf f}} - \frac{{\bB}_k \bar {\bf f} }{{\bf \bar f }^{\sf H} {\bB}_k \bar {\bf f}} \right).
\end{align}
Using \eqref{eq:partial_L1} the stationary condition holds if
\begin{align}
    \label{eq:kkt_cond_rearrange}
    \sum_{k=1}^{K}  \left(\frac{\bA_k}{ {\bf \bar f }^{\sf H} \bA_k \bar {\bf f}} \right)\bar {\bf f}
    =\sum_{k=1}^{K} \left(\frac{\bB_k}{ {\bf \bar f }^{\sf H} \bB_k \bar {\bf f}} \right)\bar{\bf f}.
\end{align}
The stationary condition can be regarded as the generalized eigenvalue problem with its corresponding matrices $\bar{\bA}(\bff)$ and $\bar{\bB}(\bff)$ defined in \eqref{eq:A_KKT} and \eqref{eq:B_KKT}:
\begin{align}
    \bar{\bA}(\bar {\bf f}) \bar {\bf f} = \lambda_{\sf BS}(\bar {\bf f}) \bar{\bB} (\bar {\bf f}) \bar {\bf f}.
\end{align}
Here, $\lambda_{{\sf BS}, \sf{num}}(\bar{\bff})$ and $\lambda_{{\sf BS}, \sf{den}}(\bar{\bff})$ can be any function such that $\lambda_{\sf BS}(\bar{\bff}) = \lambda_{{\sf BS}, \sf{num}}(\bar{\bff}) / \lambda_{{\sf BS}, \sf{den}}(\bar{\bff})$.
We note that $\bar{\bB}(\bw)$  can be considered to be invertible.
This completes the proof.
\hfill $\blacksquare$

% -----------------------------------------------------
\section{Proof of Corollary~\ref{lem:NEP_RIS}} \label{pf:NEPv_RIS}
% -----------------------------------------------------
Using the LogSumExp approach in \ref{subsec:RIS}, the Lagrangian function of the problem \eqref{eq:AO_RIS_2} is expressed as $\Tilde{\cL}_{\sf RIS}(\bw)$ in \eqref{eq:Rayleigh_2}.
To find the stationary points, we take derivatives of \eqref{eq:Rayleigh_2} as
\begin{align}
    \label{eq:partial_2}
    &\frac{\partial \Tilde{\CMcal{L}}_{\sf RIS}({\bw})}{\partial {\bw}^{\sf H}} = \frac{2}{R_{\Sigma} \ln{2}}\sum_{k = 1}^{K} \left[\frac{\bC_k {\bw}}{{\bw}^{\sf H}\bC_k {\bw}} - \frac{\bD_k {\bw}}{{\bw}^{\sf H}\bD_k {\bw}}\right] 
    \nonumber \\
    &- \frac{2 \mu}{\tau} \frac{\sum_{m=1}^{M}\bX_m \bw e^{{\bw}^{\sf H}\bX_{m}{\bw}}}{\sum_{m=1}^{M} e^{{\bw}^{\sf H}\bX_{m}{\bw}}} + \frac{2 \mu}{\tau} \frac{\sum_{m=1}^{M} \bX_m \bw e^{\frac{{\bw}^{\sf H}\bX_{m}{\bw}}{-\alpha_2}}}{\sum_{m=1}^{M} e^{\frac{{\bw}^{\sf H}\bX_{m}{\bw}}{-\alpha_2}}}.
\end{align}
Similar to the proof of Lemma~\ref{lem:NEP},  we find the condition for $\partial 
 \Tilde{\cL}_{\sf RIS}({\bw})/\partial {\bw}^{\sf H} = 0$.
Consequently, from \eqref{eq:partial_2}, the stationary condition can be reformulated as  
\begin{align}
    \bar{\bC}(\bw) \bw = \lambda_{\sf RIS}(\bw) \bar{\bD}(\bw) \bw,
\end{align}
where its corresponding matrices $\bar{\bC}(\bw)$ and $\bar{\bD}(\bw)$ are defined in \eqref{eq:C_KKT} and \eqref{eq:D_KKT}.
% Here, $\lambda_{\sf RIS}(\bw) = \lambda_{{\sf RIS}, \sf{num}}(\bw) / \lambda_{{\sf RIS}, \sf{den}}(\bw)$.
Here, $\lambda_{{\sf RIS}, \sf{num}}(\bw)$ and $\lambda_{{\sf RIS}, \sf{den}}(\bw)$ can be any function such that $\lambda_{\sf RIS}(\bw) = \lambda_{{\sf RIS}, \sf{num}}(\bw) / \lambda_{{\sf RIS}, \sf{den}}(\bw)$.
We note that $\bar{\bD}(\bw)$  can be considered to be invertible.
% This completes the proof.
\hfill $\blacksquare$

\bibliographystyle{IEEEtran}
\bibliography{ref}

@string{procieee="Proc. IEEE"}

@string{wcnc="Proc. IEEE Wireless Commun. and Netw. Conf."}

@string{jsac="IEEE J. Sel. Areas Commun."}

@string{jstsp="IEEE J. Sel. Topics Signal Process."}

@string{tvt="IEEE Trans. Veh. Technol."}

@string{tsp="IEEE Trans. Signal Process."}

@string{tcom="IEEE Trans. Commun."}

@string{twc="IEEE Trans. Wireless Commun."}

@string{tit="IEEE Trans. Inf. Theory"}

@string{ProcIEEE="Proc. of the IEEE"}

@string{splett="IEEE Signal Process. Lett."}

@string{commlett="IEEE Commun. Lett."}

@string{wcl = "IEEE Wireless Commun. Lett."}

@string{surv="IEEE Comm. Surveys \& Tutorials"}

@article{pan2022overview,
  title={{An overview of signal processing techniques for RIS/IRS-aided wireless systems}},
  author={Pan, Cunhua and Zhou, Gui and Zhi, Kangda and Hong, Sheng and Wu, Tuo and Pan, Yijin and Ren, Hong and Di Renzo, Marco and Swindlehurst, A Lee and Zhang, Rui and others},
  journal=jstsp,
  year={2022},
  publisher={IEEE}
}

@article{kundu2021channel,
  title={{Channel estimation for reconfigurable intelligent surface aided MISO communications: From LMMSE to deep learning solutions}},
  author={Kundu, Neel Kanth and McKay, Matthew R},
  journal={IEEE Open Jour. of the Commun. Society},
  volume={2},
  pages={471--487},
  year={2021},
  publisher={IEEE}
}

@article{choi2019joint,
  title={{Joint user selection, power allocation, and precoding design with imperfect CSIT for multi-cell MU-MIMO downlink systems}},
  author={Choi, Jiwook and Lee, Namyoon and Hong, Song-Nam and Caire, Giuseppe},
  journal=twc,
  volume={19},
  number={1},
  pages={162--176},
  year={2019},
  publisher={IEEE}
}

@article{shen2010dual,
  title={{On the dual formulation of boosting algorithms}},
  author={Shen, Chunhua and Li, Hanxi},
  journal={IEEE Trans. Pattern Analysis and Machine Intelligence},
  volume={32},
  number={12},
  pages={2216--2231},
  year={2010},
  publisher={IEEE}
}

@article{park2017dynamic,
  title={{Dynamic subarrays for hybrid precoding in wideband mmWave MIMO systems}},
  author={Park, Sungwoo and Alkhateeb, Ahmed and Heath, Robert W},
  journal=twc,
  volume={16},
  number={5},
  pages={2907--2920},
  year={2017},
  publisher={IEEE}
}

@article{jin2023low,
  title={{Low-complexity joint beamforming for RIS-assisted MU-MISO systems based on model-driven deep learning}},
  author={Jin, Weijie and Zhang, Jing and Wen, Chao-Kai and Jin, Shi and Li, Xiao and Han, Shuangfeng},
  journal=twc,
  volume={23},
  number={7},
  pages={6968--6982},
  year={2023},
  publisher={IEEE}
}

@article{shi2011iteratively,
  title={{An iteratively weighted MMSE approach to distributed sum-utility maximization for a MIMO interfering broadcast channel}},
  author={Shi, Qingjiang and Razaviyayn, Meisam and Luo, Zhi-Quan and He, Chen},
  journal=tsp,
  volume={59},
  number={9},
  pages={4331--4340},
  year={2011},
  publisher={IEEE}
}

@article{akdeniz2014millimeter,
  title={{Millimeter wave channel modeling and cellular capacity evaluation}},
  author={Akdeniz, Mustafa Riza and Liu, Yuanpeng and Samimi, Mathew K and Sun, Shu and Rangan, Sundeep and Rappaport, Theodore S and Erkip, Elza},
  journal=jsac,
  volume={32},
  number={6},
  pages={1164--1179},
  year={2014},
  publisher={IEEE}
}

@article{pan2020multicell,
  title={{Multicell MIMO communications relying on intelligent reflecting surfaces}},
  author={Pan, Cunhua and Ren, Hong and Wang, Kezhi and Xu, Wei and Elkashlan, Maged and Nallanathan, Arumugam and Hanzo, Lajos},
  journal=twc,
  volume={19},
  number={8},
  pages={5218--5233},
  year={2020},
  publisher={IEEE}
}

@article{xu2022channel,
  title={{Channel estimation for reconfigurable intelligent surface assisted high-mobility wireless systems}},
  author={Xu, Chao and An, Jiancheng and Bai, Tong and Sugiura, Shinya and Maunder, Robert G and Wang, Zhaocheng and Yang, Lie-Liang and Hanzo, Lajos},
  journal=tvt,
  volume={72},
  number={1},
  pages={718--734},
  year={2022},
  publisher={IEEE}
}

@article{li2020weighted,
  title={{Weighted sum-rate maximization for multi-IRS aided cooperative transmission}},
  author={Li, Zhengfeng and Hua, Meng and Wang, Qingxia and Song, Qingheng},
  journal=wcl,
  volume={9},
  number={10},
  pages={1620--1624},
  year={2020},
  publisher={IEEE}
}

@article{wei2021channel,
  title={{Channel estimation for RIS assisted wireless communications—Part I: Fundamentals, solutions, and future opportunities}},
  author={Wei, Xiuhong and Shen, Decai and Dai, Linglong},
  journal=commlett,
  volume={25},
  number={5},
  pages={1398--1402},
  year={2021},
  publisher={IEEE}
}

@article{yoon2023joint,
  title={{Joint user selection and beamforming design for multi-IRS aided internet-of-things networks}},
  author={Yoon, Seok-Hyun and Lim, Byungju and Vu, Mai and Ko, Young-Chai},
  journal=tvt,
  year={2023},
  publisher={IEEE}
}

@article{zheng2024cooperative,
  title={{Cooperative multi-satellite and multi-RIS beamforming: Enhancing LEO SatCom and mitigating LEO-GEO intersystem interference}},
  author={Zheng, Ziyuan and Jing, Wenpeng and Lu, Zhaoming and Wu, Qingqing and Zhang, Haijun and Gesbert, David},
  journal=jsac,
  year={2024},
  publisher={IEEE}
}

@article{toka2024ris,
  title={{RIS-Empowered LEO satellite networks for 6G: promising usage scenarios and future directions}},
  author={Toka, Mesut and Lee, Byungju and Seong, Jaehyup and Kaushik, Aryan and Lee, Juhwan and Lee, Jungwoo and Lee, Namyoon and Shin, Wonjae and Poor, H Vincent},
  journal={arXiv preprint arXiv:2402.07381},
  year={2024}
}

@article{li2023performance,
  title={{Performance analysis of multi-RIS-aided mmWave MIMO systems using Poisson point processes}},
  author={Li, Guang-Hui and Yue, Dian-Wu and Jin, Si-Nian and Hu, Qing},
  journal=splett,
  year={2023},
  publisher={IEEE}
}

@article{aldababsa2023multiple,
  title={{Multiple RISs-aided networks: Performance analysis and optimization}},
  author={Aldababsa, Mahmoud and Salhab, Anas M and Nasir, Ali Arshad and Samuh, Monjed H and da Costa, Daniel Benevides},
  journal=tvt,
  volume={72},
  number={6},
  pages={7545--7559},
  year={2023},
  publisher={IEEE}
}

@article{wang2020joint,
  title={{Joint transceiver and large intelligent surface design for massive MIMO mmWave systems}},
  author={Wang, Peilan and Fang, Jun and Dai, Linglong and Li, Hongbin},
  journal=twc,
  volume={20},
  number={2},
  pages={1052--1064},
  year={2020},
  publisher={IEEE}
}

@article{omid2023robust,
  title={{Robust beamforming design for an IRS-aided NOMA communication system with CSI uncertainty}},
  author={Omid, Yasaman and Shahabi, SM Mahdi and Pan, Cunhua and Deng, Yansha and Nallanathan, Arumugam},
  journal=twc,
  volume={23},
  number={2},
  pages={874--889},
  year={2023},
  publisher={IEEE}
}

@article{do2022line,
  title={{Line-of-sight MIMO via intelligent reflecting surface}},
  author={Do, Heedong and Lee, Namyoon and Lozano, Angel},
  journal=twc,
  year={2022},
  publisher={IEEE}
}

@article{mei2022intelligent,
  title={{Intelligent reflecting surface-aided wireless networks: From single-reflection to multireflection design and optimization}},
  author={Mei, Weidong and Zheng, Beixiong and You, Changsheng and Zhang, Rui},
  journal=ProcIEEE,
  volume={110},
  number={9},
  pages={1380--1400},
  year={2022},
  publisher={IEEE}
}

@article{zhou2020framework,
  title={{A framework of robust transmission design for IRS-aided MISO communications with imperfect cascaded channels}},
  author={Zhou, Gui and Pan, Cunhua and Ren, Hong and Wang, Kezhi and Nallanathan, Arumugam},
  journal=tsp,
  volume={68},
  pages={5092--5106},
  year={2020},
  publisher={IEEE}
}

@article{omid2021low,
  title={{Low-complexity robust beamforming design for IRS-aided MISO systems with imperfect channels}},
  author={Omid, Yasaman and Shahabi, Seyyed MohammadMahdi and Pan, Cunhua and Deng, Yansha and Nallanathan, Arumugam},
  journal=commlett,
  volume={25},
  number={5},
  pages={1697--1701},
  year={2021},
  publisher={IEEE}
}

@article{wang2023ris,
  title={{RIS-aided MIMO systems with hardware impairments: Robust beamforming design and analysis}},
  author={Wang, Jintao and Gong, Shiqi and Wu, Qingqing and Ma, Shaodan},
  journal=twc,
  volume={22},
  number={10},
  pages={6914--6929},
  year={2023},
  publisher={IEEE}
}

@article{heath2016overview,
  title={{An overview of signal processing techniques for millimeter wave MIMO systems}},
  author={Heath, Robert W and Gonzalez-Prelcic, Nuria and Rangan, Sundeep and Roh, Wonil and Sayeed, Akbar M},
  journal=jstsp,
  volume={10},
  number={3},
  pages={436--453},
  year={2016},
  publisher={IEEE}
}

@article{liu2021reconfigurable,
  title={{Reconfigurable intelligent surfaces: Principles and opportunities}},
  author={Liu, Yuanwei and Liu, Xiao and Mu, Xidong and Hou, Tianwei and Xu, Jiaqi and Di Renzo, Marco and Al-Dhahir, Naofal},
  journal=surv,
  volume={23},
  number={3},
  pages={1546--1577},
  year={2021},
  publisher={IEEE}
}

@article{zeng2022joint,
  title={{{Joint beamforming design for IRS aided multiuser MIMO with imperfect CSI}}},
  author={Zeng, Piao and Qiao, Deli and Qian, Haifeng and Wu, Qingqing},
  journal = tvt,
  volume={71},
  number={10},
  pages={10729--10743},
  year={2022},
  publisher={IEEE}
}

@article{jiang2023robust,
  title={{Robust design of IRS-aided multi-group multicast system with imperfect CSI}},
  author={Jiang, Weiheng and Xiong, Peiyun and Nie, Jiangtian and Ding, Zhiguo and Pan, Cunhua and Xiong, Zehui},
  journal=twc,
  year={2023},
  publisher={IEEE}
}

@inproceedings{xiu2021irs,
  title={{IRS-assisted millimeter wave communications: Joint power allocation and beamforming design}},
  author={Xiu, Yue and Zhao, Yang and Liu, Yang and Zhao, Jun and Yagan, Osman and Wei, Ning},
  booktitle=wcnc,
  pages={1--6},
  year={2021},
  organization={IEEE}
}

@article{zargari2020energy,
  title={{Energy efficiency maximization via joint active and passive beamforming design for multiuser MISO IRS-aided SWIPT}},
  author={Zargari, Shayan and Khalili, Ata and Zhang, Rui},
  journal=wcl,
  volume={10},
  number={3},
  pages={557--561},
  year={2020},
  publisher={IEEE}
}

@article{zhou2020intelligent,
  title={{Intelligent reflecting surface aided multigroup multicast MISO communication systems}},
  author={Zhou, Gui and Pan, Cunhua and Ren, Hong and Wang, Kezhi and Nallanathan, Arumugam},
  journal=tsp,
  volume={68},
  pages={3236--3251},
  year={2020},
  publisher={IEEE}
}

@article{shtaiwi2023sum,
  title={{Sum-rate maximization for RIS-assisted integrated sensing and communication systems with manifold optimization}},
  author={Shtaiwi, Eyad and Zhang, Hongliang and Abdelhadi, Ahmed and Swindlehurst, A Lee and Han, Zhu and Poor, H Vincent},
  journal=tcom,
  year={2023},
  publisher={IEEE}
}

@article{li2024joint,
  title={{Joint resource allocation and reflecting precoding design for RIS-assisted ISAC systems}},
  author={Li, Xiaohui and Wang, Hong and Chen, Yunpei and Sheng, Shuran},
  journal=wcl,
  year={2024},
  publisher={IEEE}
}

@article{han2020cooperative,
  title={{Cooperative double-IRS aided communication: Beamforming design and power scaling}},
  author={Han, Yitao and Zhang, Shuowen and Duan, Lingjie and Zhang, Rui},
  journal=wcl,
  volume={9},
  number={8},
  pages={1206--1210},
  year={2020},
  publisher={IEEE}
}

@article{huang2023multi,
  title={{Multi-IRS-aided millimeter-wave multi-user MISO systems for power minimization using generalized Benders decomposition}},
  author={Huang, Huan and Zhang, Ying and Zhang, Hongliang and Zhao, Zixin and Zhang, Chongfu and Han, Zhu},
  journal=twc,
  year={2023},
  publisher={IEEE}
}

@article{joudeh2016sum,
  title={{Sum-rate maximization for linearly precoded downlink multiuser MISO systems with partial CSIT: A rate-splitting approach}},
  author={Joudeh, Hamdi and Clerckx, Bruno},
  journal=tcom,
  volume={64},
  number={11},
  pages={4847--4861},
  year={2016},
  publisher={IEEE}
}

@article{choi2022energy,
  title={{Energy efficiency maximization precoding for quantized massive MIMO systems}},
  author={Choi, Jinseok and Park, Jeonghun and Lee, Namyoon},
  journal=twc,
  volume={21},
  number={9},
  pages={6803--6817},
  year={2022},
  publisher={IEEE}
}

@article{lin2021channel,
  title={{Channel estimation and user localization for IRS-assisted MIMO-OFDM systems}},
  author={Lin, Yuxing and Jin, Shi and Matthaiou, Michail and You, Xiaohu},
  journal=twc,
  volume={21},
  number={4},
  pages={2320--2335},
  year={2021},
  publisher={IEEE}
}

@article{yao2023robust,
  title={{Robust beamforming design for RIS-aided cell-free systems with CSI uncertainties and capacity-limited backhaul}},
  author={Yao, Jiacheng and Xu, Jindan and Xu, Wei and Ng, Derrick Wing Kwan and Yuen, Chau and You, Xiaohu},
  journal=tcom,
  volume={71},
  number={8},
  pages={4636--4649},
  year={2023},
  publisher={IEEE}
}

@article{al2024performance,
  title={{Performance of multi-RIS-aided cell-free massive MIMO: Do more RISs always help?}},
  author={Al-Nahhas, Bayan and Obeed, Mohanad and Chaaban, Anas and Hossain, Md Jahangir},
  journal={IEEE Trans. on Commun.},
  year={2024},
  publisher={IEEE}
}

@ARTICLE{lee2023max,
  author={Lee, Sangmin and Choi, Eunsung and Choi, Jinseok},
  journal={IEEE Wireless Commun. Lett.}, 
  title={{Max-Min Fairness Precoding for Physical Layer Security With Partial Channel Knowledge}}, 
  year={2023},
  volume={12},
  number={9},
  pages={1637-1641},
  doi={10.1109/LWC.2023.3285797}}

@article{basar2024reconfigurable,
  title={{Reconfigurable intelligent surfaces for 6G: Emerging hardware architectures, applications, and open challenges}},
  author={Basar, Ertugrul and Alexandropoulos, George C and Liu, Yuanwei and Wu, Qingqing and Jin, Shi and Yuen, Chau and Dobre, Octavia A and Schober, Robert},
  journal={IEEE Veh. Tech. Mag.},
  year={2024},
  publisher={IEEE}
}

@article{yao2023superimposed,
  title={{Superimposed RIS-phase modulation for MIMO communications: A novel paradigm of information transfer}},
  author={Yao, Jiacheng and Xu, Jindan and Xu, Wei and Yuen, Chau and You, Xiaohu},
  journal=twc,
  volume={23},
  number={4},
  pages={2978--2993},
  year={2023},
  publisher={IEEE}
}

@article{guo2020weighted,
  title={{Weighted sum-rate maximization for reconfigurable intelligent surface aided wireless networks}},
  author={Guo, Huayan and Liang, Ying-Chang and Chen, Jie and Larsson, Erik G},
  journal=twc,
  volume={19},
  number={5},
  pages={3064--3076},
  year={2020},
  publisher={IEEE}
}

@article{hassibi2003much,
  title={{How much training is needed in multiple-antenna wireless links?}},
  author={Hassibi, Babak and Hochwald, Bertrand M},
  journal=tit,
  volume={49},
  number={4},
  pages={951--963},
  year={2003},
  publisher={IEEE}
}

@article{chen2025joint,
  title={{Joint Power Allocation and Phase Shifts Design for Distributed RIS-assisted Multiuser Systems}},
  author={Chen, Zhen and Chen, Gaojie and Zhang, Xiu Yin and Tang, Jie and Jin, Shi and Wong, Kai-Kit and Chambers, Jonathon},
  journal={IEEE Trans. on Mobile Comput.},
  year={2025},
  publisher={IEEE}
}

@article{wasfi2026unified,
  title={{Unified Mathematical Framework for Channel Estimation in VLC Systems With Signal-Dependent Noise}},
  author={Wasfi, Asma and Yaseen, Maysa and Awwad, Falah and Ikki, Salama S},
  journal=commlett}

@article{huang2019reconfigurable,
  title={{Reconfigurable intelligent surfaces for energy efficiency in wireless communication}},
  author={Huang, Chongwen and Zappone, Alessio and Alexandropoulos, George C and Debbah, M{\'e}rouane and Yuen, Chau},
  journal=twc,
  volume={18},
  number={8},
  pages={4157--4170},
  year={2019},
  publisher={IEEE}
}

@article{oh2025generalized,
  title={{Generalized Power Iteration-Based Precoding With Low-Complexity Matrix Inversion Approaches for Massive MIMO Systems}},
  author={Oh, Mintaek and Yoo, Seunghyeong and Lee, Namyoon and Choi, Jinseok},
  journal=tvt,
  year={2025},
  publisher={IEEE}
}

@article{nesterov2005smooth,
  title={{Smooth minimization of non-smooth functions}},
  author={Nesterov, Yu},
  journal={Mathematical programming},
  volume={103},
  number={1},
  pages={127--152},
  year={2005},
  publisher={Springer}
}

\end{document}